\documentclass[reqno,10pt,a4paper,dvips]{amsart}

\usepackage{amssymb,mathptmx,cite,psfrag,eucal,array,setspace,amscd,color,geometry,enumitem,varwidth}
\usepackage[dvips]{graphicx}
\usepackage[margin=\parindent]{caption}
\usepackage{tikz}

\let\oldtocsection=\tocsection
\let\oldtocsubsection=\tocsubsection
\renewcommand{\tocsection}[2]{\hspace{0em} \oldtocsection{#1}{#2}}
\renewcommand{\tocsubsection}[2]{\hspace{1.7em} \oldtocsubsection{#1}{#2}} 

\DeclareSymbolFont{largesymbols}{OMX}{zplm}{m}{n} 

\numberwithin{equation}{section}
\usepackage{chngcntr}
\counterwithout{figure}{section}

\newcolumntype{C}{>{$}c<{$}} 

\setitemize{leftmargin=*} 

\allowdisplaybreaks

\let\originalleft\left
\let\originalright\right
\renewcommand{\left}{\mathopen{}\mathclose\bgroup\originalleft}
\renewcommand{\right}{\aftergroup\egroup\originalright}

\newcommand{\eps}{\varepsilon}

\newcommand{\alg}[1]{\mathfrak{#1}}
\newcommand{\grp}[1]{\mathsf{#1}}

\newcommand{\func}[2]{#1 \left( #2 \right)} 
\newcommand{\tfunc}[2]{#1 \bigl( #2 \bigr)} 

\newcommand{\brac}[1]{\left( #1 \right)}
\newcommand{\tbrac}[1]{\bigl( #1 \bigr)}
\newcommand{\sqbrac}[1]{\left[ #1 \right]}

\newcommand{\set}[1]{\left\{ #1 \right\}}
\newcommand{\tset}[1]{\bigl\{ #1 \bigr\}}

\newcommand{\st}{\mspace{5mu} : \mspace{5mu}} 

\newcommand{\abs}[1]{\left\lvert #1 \right\rvert}

\newcommand{\ZZ}{\mathbb{Z}}

\newcommand{\RR}{\mathbb{R}}

\newcommand{\CC}{\mathbb{C}}

\newcommand{\ahol}[1]{\overline{#1}} 
\newcommand{\nonch}[1]{#1}           

\newcommand{\pd}{\partial}         
\newcommand{\apd}{\ahol{\partial}} 

\newcommand{\dd}{\mathrm{d}}   
\newcommand{\ii}{\mathfrak{i}} 
\newcommand{\ee}{\mathsf{e}}   

\newcommand{\killing}[2]{\kappa \bigl( #1 , #2 \bigr)} 

\newcommand{\normord}[1]{\mbox{${} : #1 : {}$}} 

\newcommand{\comm}[2]{\bigl[ #1 , #2 \bigr]}
\newcommand{\acomm}[2]{\bigl\{ #1 , #2 \bigr\}}

\newcommand{\bra}[1]{\bigl\langle #1 \bigr\rvert}
\newcommand{\ket}[1]{\bigl\lvert #1 \bigr\rangle}
\newcommand{\braket}[2]{\bigl\langle #1 \bigr\rvert \bigl. #2 \bigr\rangle}          
\newcommand{\bracket}[3]{\bigl\langle #1 \bigr\rvert #2 \bigl\lvert #3 \bigr\rangle} 

\newcommand{\corrfn}[1]{\bigl\langle #1 \bigr\rangle}

\newcommand{\conjaut}{\mathsf{w}}             
\newcommand{\sfaut}{\sigma}                   
\newcommand{\conjmod}[1]{#1^{\star}}          
\newcommand{\sfmod}[2]{#2^{(#1)}}             

\newcommand{\affine}[1]{\widehat{#1}}
\newcommand{\SLG}[2]{\grp{#1} \bigl( #2 \bigr)}                             
\newcommand{\SLA}[2]{\alg{#1} \bigl( #2 \bigr)}                             
\newcommand{\SLSA}[3]{\alg{#1} \left( #2 \middle\vert #3 \right)}           
\newcommand{\SLSG}[3]{\grp{#1} \left( #2 \middle\vert #3 \right)}           
\newcommand{\AKMA}[2]{\affine{\alg{#1}} \left( #2 \right)}                  
\newcommand{\AKMSA}[3]{\affine{\alg{#1}} \left( #2 \middle\vert #3 \right)} 

\newcommand{\minmod}[2]{\mathsf{M} \bigl( #1 , #2 \bigr)}                   

\newcommand{\SingAlg}[1]{\mathfrak{M} \bigl( #1 \bigr)}                     
\newcommand{\TripAlg}[1]{\mathfrak{W} \bigl( #1 \bigr)}                     

\newcommand{\ExtAlg}[1]{\mathbb{W}^{(#1)}}                                  

\newcommand{\finite}[1]{\overline{#1}}

\newcommand{\GLVer}[1]{\mathcal{V}_{#1}}  
\newcommand{\GLAtyp}[1]{\mathcal{A}_{#1}} 
\newcommand{\GLStag}[1]{\mathcal{P}_{#1}} 
\newcommand{\GLBulk}[1]{\mathcal{B}_{#1}} 
\newcommand{\GLBulkatyp}[1]{\mathcal{J}_{#1}} 

\newcommand{\FinGLVer}[1]{\finite{\mathcal{V}}_{#1}}  
\newcommand{\FinGLAtyp}[1]{\finite{\mathcal{A}}_{#1}} 
\newcommand{\FinGLStag}[1]{\finite{\mathcal{P}}_{#1}} 

\newcommand{\SLIrr}[1]{\mathcal{L}_{#1}}  
\newcommand{\SLDisc}[1]{\mathcal{D}_{#1}} 
\newcommand{\SLTyp}[1]{\mathcal{E}_{#1}}  
\newcommand{\SLStag}[1]{\mathcal{S}_{#1}} 
\newcommand{\SLBulk}[1]{\mathcal{B}_{#1}} 

\newcommand{\ExtSLTyp}[2]{\mathbb{E}^{(#1)}_{#2}}

\newcommand{\BGIrr}[1]{\mathsf{L}_{#1}}  
\newcommand{\BGTyp}[1]{\mathsf{E}_{#1}}  
\newcommand{\BGStag}[1]{\mathsf{S}_{#1}} 

\newcommand{\Fock}[1]{\mathcal{F}_{#1}}     
\newcommand{\ExtFock}[1]{\mathbb{F}_{[#1]}} 

\newcommand{\SFIrr}[1]{\mathbb{L}_{#1}}  
\newcommand{\SFVer}[1]{\mathbb{V}_{#1}}  
\newcommand{\SFStag}[1]{\mathbb{S}_{#1}} 
\newcommand{\SFStagbulk}[1]{\mathbb{B}_{#1}}

\newcommand{\SingTyp}[1]{\mathcal{F}_{#1}}  
\newcommand{\SingAtyp}[1]{\mathcal{M}_{#1}} 
\newcommand{\SingStag}[1]{\mathcal{P}_{#1}} 
\newcommand{\TripIrr}[1]{\mathcal{W}_{#1}}  
\newcommand{\TripStag}[1]{\mathcal{S}_{#1}} 
\newcommand{\TripStagbulk}[1]{\mathcal{B}_{#1}}

\newcommand{\ExtSingTyp}[2]{\mathbb{F}^{(#1)}_{#2}}  
\newcommand{\ExtSingAtyp}[2]{\mathbb{M}^{(#1)}_{#2}} 

\newcommand{\VirVer}[1]{\mathcal{V}_{#1}}   
\newcommand{\VirIrr}[1]{\mathcal{L}_{#1}}   
\newcommand{\VirQuot}[1]{\mathcal{Q}_{#1}}  
\newcommand{\VirStag}[1]{\mathcal{S}_{#1}}  

\newcommand{\LMod}{\mathcal{H}^{\textup{L}}}    
\newcommand{\RMod}{\mathcal{H}^{\textup{R}}}    
\newcommand{\LDim}{h^{\textup{L}}}              
\newcommand{\RDim}{h^{\textup{R}}}              
\newcommand{\subRMod}{\mathcal{K}^{\textup{R}}} 

\DeclareMathOperator{\tr}{tr}
\DeclareMathOperator{\str}{str}
\newcommand{\traceover}[1]{\tr_{\raisebox{-3pt}{$\scriptstyle #1$}}}
\newcommand{\straceover}[1]{\str_{\raisebox{-3pt}{$\scriptstyle #1$}}}
\newcommand{\chmap}{\mathrm{ch}}
\newcommand{\schmap}{\mathrm{sch}}
\newcommand{\ch}[1]{\chmap \bigl[ #1 \bigr]}     
\newcommand{\fch}[2]{\ch{#1} \bigl( #2 \bigr)}   
\newcommand{\sch}[1]{\schmap \bigl[ #1 \bigr]}
\newcommand{\fsch}[2]{\sch{#1} \bigl( #2 \bigr)}

\newcommand{\jth}[1]{\vartheta_{#1}}            
\newcommand{\Jth}[2]{\jth{#1} \bigl( #2 \bigr)} 

\newcommand{\modS}{\mathsf{S}} 

\newcommand{\modT}{\mathsf{T}} 
\newcommand{\modC}{\mathsf{C}} 

\newcommand{\atyp}[1]{\underline{#1}} 

\newcommand{\ExtmodS}[1]{\mathbb{S}^{(#1)}} 

\newcommand{\Sch}[1]{\Bigl. \ch{#1} \Bigr\rvert_{\modS}}   
\newcommand{\Tch}[1]{\Bigl. \ch{#1} \Bigr\rvert_{\modT}}   
\newcommand{\Ssch}[1]{\Bigl. \sch{#1} \Bigr\rvert_{\modS}} 

\newcommand{\fuse}{\mathbin{\times}}                                 
\newcommand{\grfuse}{\mathbin{\dot{\times}}}                         
\newcommand{\fakefuse}{\mathbin{\text{``{}$\fuse$''}}}               
\newcommand{\fuscoeff}[2]{\mathsf{N}_{#1}^{\hphantom{#1} #2}}        
\newcommand{\fusmat}[1]{\mathsf{N}_{#1}}                             
\newcommand{\Extfusmat}[2]{\mathbb{N}^{(#1)}_{[#2]}}                 

\newdimen{\tempwidth}
\newcommand{\Extfuscoeff}[4]{
  \settowidth{\tempwidth}{$(#1)$}
  \mathsf{N}_{[#2] [#3]}^{(#1) \hspace{-\tempwidth} \hphantom{[#2] [#3]} [#4]}
}

\newcommand{\Groth}{\mathsf{Gr}}       
\newcommand{\tGr}[1]{\bigl[ #1 \bigr]}

\newcommand{\partfunc}[1]{\mathbf{Z}_{\text{#1}}}
\newcommand{\bulkstatespace}{\mathcal{H}_{\text{bulk}}}

\newcommand{\ra}{\rightarrow}
\newcommand{\lra}{\longrightarrow}
\newcommand{\ses}[3]{0 \ra #1 \ra #2 \ra #3 \ra 0}                                  
\newcommand{\dses}[5]{0 \lra #1 \overset{#2}{\lra} #3 \overset{#4}{\lra} #5 \lra 0} 

\newcommand{\eqnref}[1]{Equation~\eqref{#1}}

\newcommand{\secref}[1]{Section~\ref{#1}}
\newcommand{\secDref}[2]{Sections~\ref{#1} and \ref{#2}}
\newcommand{\secTref}[3]{Sections~\ref{#1}, \ref{#2} and \ref{#3}}
\newcommand{\appref}[1]{Appendix~\ref{#1}}
\newcommand{\figref}[1]{Figure~\ref{#1}}
\newcommand{\figDref}[2]{Figures~\ref{#1} and \ref{#2}}
\newcommand{\tabref}[1]{Table~\ref{#1}}
\newcommand{\thmref}[1]{Theorem~\ref{#1}}

\newcommand{\cft}{conformal field theory}
\newcommand{\cfts}{conformal field theories}
\newcommand{\uea}{universal enveloping algebra}
\newcommand{\lcft}{logarithmic conformal field theory}
\newcommand{\lcfts}{logarithmic conformal field theories}
\newcommand{\WZW}{Wess-Zumino-Witten}
\newcommand{\ope}{operator product expansion}
\newcommand{\opes}{operator product expansions}
\newcommand{\hws}{highest weight state}
\newcommand{\hwss}{highest weight states}
\newcommand{\sv}{singular vector}
\newcommand{\svs}{singular vectors}
\newcommand{\hwm}{highest weight module}
\newcommand{\hwms}{highest weight modules}
\newcommand{\hwsm}{highest weight submodule}

\newcommand{\pde}{partial differential equation}
\newcommand{\pdes}{partial differential equations}

\DeclareMathOperator{\id}{id}
\DeclareMathOperator{\im}{im}
\DeclareMathOperator{\Ext}{Ext}
\DeclareMathOperator{\soc}{soc}
\DeclareMathOperator{\sgn}{sgn}

\newcommand{\ExtGrp}[3]{\Ext^{#1} \left( #2 , #3 \right)}

\newcommand{\directint}{\ominus \mspace{-17.4mu} \int} 

\theoremstyle{plain}
\newtheorem*{thm}{Theorem}

\begin{document}

\title[Logarithmic Conformal Field Theory]{Logarithmic Conformal Field Theory: \\ Beyond an Introduction}

\author[T Creutzig]{Thomas Creutzig}

\address[T Creutzig]{
Fachbereich Mathematik \\
Technische Universit\"{a}t Darmstadt \\
Schlo\ss{}gartenstra\ss{}e 7 \\
64289 Darmstadt \\
Germany
}

\email{tcreutzig@mathematik.tu-darmstadt.de}

\author[D Ridout]{David Ridout}

\address[David Ridout]{
Department of Theoretical Physics \\
Research School of Physics and Engineering;
and
Mathematical Sciences Institute;
Australian National University \\
Canberra, ACT 2600 \\
Australia
}

\email{david.ridout@anu.edu.au}

\thanks{\today}

\setstretch{1.1}

\begin{abstract}
This article aims to review a selection of central topics and examples in \lcft{}.  It begins with the remarkable observation of Cardy that the horizontal crossing probability of critical percolation may be computed analytically within the formalism of boundary \cft{}.  Cardy's derivation relies on certain implicit assumptions which are shown to lead inexorably to indecomposable modules and logarithmic singularities in correlators.  For this, a short introduction to the fusion algorithm of Nahm, Gaberdiel and Kausch is provided.

While the percolation \lcft{} is still not completely understood, there are several examples for which the formalism familiar from rational \cft{}, including bulk partition functions, correlation functions, modular transformations, fusion rules and the Verlinde formula, has been successfully generalised.  This is illustrated for three examples:  The singlet model $\SingAlg{1,2}$, related to the triplet model $\TripAlg{1,2}$, symplectic fermions and the fermionic $bc$ ghost system; the fractional level \WZW{} model based on $\AKMA{sl}{2}$ at $k=-\tfrac{1}{2}$, related to the bosonic $\beta \gamma$ ghost system; and the \WZW{} model for the Lie supergroup $\SLSG{GL}{1}{1}$, related to $\SLSG{SL}{2}{1}$ at $k=-\tfrac{1}{2}$ and $1$, the Bershadsky-Polyakov algebra $W_3^{(2)}$ and the Feigin-Semikhatov algebras $W_n^{(2)}$.  These examples have been chosen because they represent the most accessible, and most useful, members of the three best-understood families of \lcfts{}:  The logarithmic minimal models $\TripAlg{q,p}$, the fractional level \WZW{} models, and the \WZW{} models on Lie supergroups (excluding $\SLSG{OSP}{1}{2n}$).

In this review, the emphasis lies on the representation theory of the underlying chiral algebra and the modular data pertaining to the characters of the representations.  Each of the archetypal \lcfts{} is studied here by first determining its irreducible spectrum, which turns out to be continuous, as well as a selection of natural reducible, but indecomposable, modules.  This is followed by a detailed description of how to obtain character formulae for each irreducible, a derivation of the action of the modular group on the characters, and an application of the Verlinde formula to compute the Grothendieck fusion rules.  In each case, the (genuine) fusion rules are known, so comparisons can be made and favourable conclusions drawn.  In addition, each example admits an infinite set of simple currents, hence extended symmetry algebras may be constructed and a series of bulk modular invariants computed.  The spectrum of such an extended theory is typically discrete and this is how the triplet model $\TripAlg{1,2}$ arises, for example.  Moreover, simple current technology admits a derivation of the extended algebra fusion rules from those of its continuous parent theory.  Finally, each example is concluded by a brief description of the computation of some bulk correlators, a discussion of the structure of the bulk state space, and remarks concerning more advanced developments and generalisations.

The final part gives a very short account of the theory of staggered modules, the (simplest class of) representations that are responsible for the logarithmic singularities that distinguish logarithmic theories from their rational cousins.  These modules are discussed in a generality suitable to encompass all the examples met in this review and some of the very basic structure theory is proven.  Then, the important quantities known as logarithmic couplings are reviewed for Virasoro staggered modules and their role as fundamentally important parameters, akin to the three-point constants of rational \cft{}, is discussed.  An appendix is also provided in order to introduce some of the necessary, but perhaps unfamiliar, language of homological algebra.
\end{abstract}

\maketitle

\newpage

\tableofcontents

\newpage

\section{Introduction} \label{sec:Intro}

Ever since the pioneering work of Belavin, Polyakov and Zamolodchikov \cite{BelInf84}, two-dimensional \cft{} has been at the forefront of much of the progress in modern mathematical physics.  Its application to the study of critical statistical models and string theory is well known, see \cite{SalExa87,WitNon84} for example, but it also provides the basic inspiration for the mathematical theory of vertex operator algebras.  The simplest \cfts{} are constructed mathematically from irreducible representations of an infinite-dimensional symmetry algebra.  However, recent attention to non-local observables for statistical models and string theories with fermionic degrees of freedom has led to the conclusion that the corresponding field-theoretic models require, in addition, certain reducible, but indecomposable, representations.  Such models have come to be known as \emph{\lcfts{}} because the type of indecomposability required leads to logarithmic singularities in correlation functions.

As a field of study, \lcft{} dates back to the works of Rozansky and Saleur on the $\SLSG{U}{1}{1}$ (or perhaps $\SLSG{GL}{1}{1}$) \WZW{} model \cite{Rozansky:1992rx,Rozansky:1992td} and that of Gurarie on a fermionic ghost system \cite{GurLog93} related to the theory now known as symplectic fermions.  Since then, things have progressed rather rapidly with many of the standard features of rational \cft{} now understood in the logarithmic setting.  In particular, there are three fine reviews of the subject \cite{FloBit03,GabAlg03,KawLog03} which focus on, among other things, modular transformations, module structure and boundary aspects, mostly for a family of theories related to symplectic fermions.

Each of these reviews are accounts of lectures given at a workshop in 2001.  The present review aims to build upon the state of knowledge summarised there, introducing the reader to some of the recent advances that seem to be converging towards a more unified picture of \lcft{}.  Unfortunately, a detailed overview would require a rather lengthy book, hence we will restrict ourselves to foundational material and in-depth examples which we believe, hopefully without controversy, are ``archetypes'' for the discipline.  We hope that our choice will give the reader a good sense of what structures \lcft{} relies upon and what one can do with it.

In particular, we work almost entirely in the continuum, expecting that the reader is familiar with rational \cft{} as described (for example) in \cite{DiFCon97}, eschewing approaches based on statistical lattice models and conjectured scaling limits (see \cite{CarLog99,ReaExa01,PeaLog06,ReaAss07}).  We also work, for the most part, with chiral algebras, even though it is well known that the holomorphic factorisation principle of rational \cft{} fails in the logarithmic setting.  Instead, we will see how to construct physically satisfactory non-chiral fields, even when logarithmic behaviour is present.  For other approaches to \lcft{}, as well as condensed matter physics and string-theoretic applications, discussions of logarithmic vertex operator algebras and other relations to mathematics, we refer to the other articles that constitute this special issue of the Journal of Physics A.

We will outline what we cover in this review shortly.  First however, we quickly remind the reader how logarithmic singularities arise in correlation functions as consequences of a non-diagonalisable action of the Virasoro zero-mode $L_0$.  Then, we digress slightly in recalling the (non-logarithmic) theory known as the free boson, in particular, its characters, their modular transformations and the relation between these and the fusion rules (the Verlinde formula \cite{VerFus88}).  This is in order to set the scene for the analysis of the ``archetypal'' logarithmic theories that follow.  We also mention simple currents for the free boson and the corresponding extended algebra theories as these ideas will also play an important role for us.

\subsection{Correlators and Logarithmic Singularities} \label{sec:LogCorr}

Conformal field theory is relatively tractable among physical models due to its infinite-dimensional algebra of symmetries.  As is well-known, this always includes the Virasoro algebra, the infinite-dimensional Lie algebra spanned by modes $L_n$, $n \in \ZZ$, and $C$ with commutation relations
\begin{equation}
\comm{L_m}{L_n} = \brac{m-n} L_{m+n} + \frac{m^3-m}{12} \delta_{m+n=0} C, \qquad \comm{L_m}{C} = 0.
\end{equation}
The central mode $C$ will act on all representations as a fixed multiple of the identity, known as the central charge $c$.  We will identify $C$ with $c$ in what follows.  The energy-momentum tensor is $\tfunc{T}{z} = \sum_{n \in \ZZ} L_n z^{-n-2}$, as usual.

In this section, we recall how the global conformal invariance of the vacuum $\ket{0}$, meaning its annihilation by $L_{-1}$, $L_0$ and $L_1$, fixes the two-point functions of (chiral) fields and gives rise to logarithmic singularities when the corresponding Virasoro representations admit a non-diagonalisable action of $L_0$.  Given any chiral field $\tfunc{\phi}{z}$, the natural action of the Virasoro modes is given by
\begin{equation}
\comm{L_n}{\func{\phi}{w}} = \oint_w \func{T}{z} \func{\phi}{w} z^{n+1} \: \frac{\dd z}{2 \pi \ii}.
\end{equation}
If $\tfunc{\phi}{z}$ is a chiral primary field of conformal weight $h$, then this action gives
\begin{equation}
\comm{L_{-1}}{\func{\phi}{w}} = \pd \func{\phi}{w}, \quad 
\comm{L_0}{\func{\phi}{w}} = h \func{\phi}{w} + w \pd \func{\phi}{w}, \quad 
\comm{L_1}{\func{\phi}{w}} = 2h w \func{\phi}{w} + w^2 \pd \func{\phi}{w}
\end{equation}
and the invariance of $\ket{0}$ then leads to the following differential equations for the two-point functions:
\begin{equation}
\begin{gathered}
\brac{\pd_z + \pd_w}  \corrfn{\func{\phi}{z} \func{\phi}{w}} = 0, \qquad 
\brac{z \pd_z + w \pd_w + 2h} \corrfn{\func{\phi}{z} \func{\phi}{w}} = 0, \\
\brac{z^2 \pd_z + w^2 \pd_w + 2h \brac{z+w}} \corrfn{\func{\phi}{z} \func{\phi}{w}} = 0.
\end{gathered}
\end{equation}
It is a straight-forward exercise to show that the general solution of these equations has the form
\begin{equation} \label{eq:PrimCorr}
\corrfn{\func{\phi}{z} \func{\phi}{w}} = \frac{A}{\brac{z-w}^{2h}},
\end{equation}
for some constant $A$ (which could be zero).  We may identify $A$ with $\braket{\phi}{\phi}$.

So far, we have repeated a standard textbook computation (see \cite{DiFCon97} for example).  We now ask what happens if the primary field $\tfunc{\phi}{z}$ corresponds to a state $\ket{\phi}$ which has a \emph{Jordan partner} $\ket{\Phi}$ under the $L_0$-action:  $L_0 \ket{\Phi} = h \ket{\Phi} + \ket{\phi}$.  Then, the constant $A$ in \eqref{eq:PrimCorr} is
\begin{equation}
A = \braket{\phi}{\phi} = \bracket{\phi}{L_0 - h}{\Phi} = 0,
\end{equation}
since $L_0 \ket{\phi} = h \ket{\phi}$.  Moreover, the partner field $\tfunc{\Phi}{z}$ has the \ope{}%
\footnote{Here, we assume for simplicity that $L_n \ket{\Phi} = 0$ for all $n>0$.  See \secDref{sec:PercLogCorr}{sec:StagLogCorr} for a more general discussion.}
\begin{equation}
\func{T}{z} \func{\Phi}{w} \sim \frac{h \func{\Phi}{w} + \func{\phi}{w}}{\brac{z-w}^2} + \frac{\pd \func{\Phi}{w}}{z-w},
\end{equation}
so that the Virasoro modes act as
\begin{equation}
\begin{gathered}
\comm{L_{-1}}{\func{\Phi}{w}} = \pd \func{\Phi}{w}, \qquad 
\comm{L_0}{\func{\Phi}{w}} = h \func{\Phi}{w} + w \pd \func{\Phi}{w} + \func{\phi}{w}, \\
\comm{L_1}{\func{\Phi}{w}} = 2h w \func{\Phi}{w} + w^2 \pd \func{\Phi}{w} + 2 w \func{\phi}{w}.
\end{gathered}
\end{equation}
We therefore obtain a set of inhomogeneous differential equations for the two-point functions:
\begin{equation}
\begin{gathered}
\brac{\pd_z + \pd_w} \corrfn{\func{\phi}{z} \func{\Phi}{w}} = 0, \qquad 
\brac{z \pd_z + w \pd_w + 2h} \corrfn{\func{\phi}{z} \func{\Phi}{w}} = -\corrfn{\func{\phi}{z} \func{\phi}{w}}, \\
\brac{z^2 \pd_z + w^2 \pd_w + 2h \brac{z+w}} \corrfn{\func{\phi}{z} \func{\Phi}{w}} = -2 w \corrfn{\func{\phi}{z} \func{\phi}{w}}, \\
\brac{\pd_z + \pd_w} \corrfn{\func{\Phi}{z} \func{\Phi}{w}}, \qquad 
\brac{z \pd_z + w \pd_w + 2h} \corrfn{\func{\Phi}{z} \func{\Phi}{w}} = -\corrfn{\func{\Phi}{z} \func{\phi}{w}} - \corrfn{\func{\phi}{z} \func{\Phi}{w}}, \\
\brac{z^2 \pd_z + w^2 \pd_w + 2h \brac{z+w}} \corrfn{\func{\Phi}{z} \func{\Phi}{w}} = -2 z \corrfn{\func{\phi}{z} \func{\Phi}{w}} - 2 w \corrfn{\func{\Phi}{z} \func{\phi}{w}}.
\end{gathered}
\end{equation}
If we assume that $\tfunc{\phi}{z}$ and $\tfunc{\Phi}{z}$ are mutually bosonic, meaning that $\corrfn{\func{\phi}{z} \func{\Phi}{w}} = \corrfn{\func{\Phi}{z} \func{\phi}{w}}$, then solving these equations leads to two-point functions of the form
\begin{equation} \label{eq:LogCorr}
\corrfn{\func{\phi}{z} \func{\phi}{w}} = 0, \qquad 
\corrfn{\func{\phi}{z} \func{\Phi}{w}} = \frac{B}{\brac{z-w}^{2h}}, \qquad 
\corrfn{\func{\Phi}{z} \func{\Phi}{w}} = \frac{C - 2B \log \brac{z-w}}{\brac{z-w}^{2h}},
\end{equation}
where $B$ and $C$ are constants.  We conclude that combining global conformal invariance with a non-diagonalisable $L_0$-action leads to logarithmic singularities in correlation functions.

We remark that $\func{\Phi}{z}$ is not uniquely specified because we may, for example, add to it any multiple of $\func{\phi}{z}$ without affecting its defining properties.  However, adding such a multiple will change the constant $C$ in \eqref{eq:LogCorr}, though $B$ will remain invariant.  Because of this, $C$ may be tuned to any desired value, so is not expected to be physical.  The constant $B = \braket{\phi}{\Phi}$, on the other hand, is expected to be physically meaningful.

\subsection{The Free Boson} \label{sec:FreeBoson}

The free boson is a $c=1$ \cft{} with chiral algebra $\AKMA{gl}{1} = \AKMA{u}{1}$ generated by modes $a_n$, $n \in \ZZ$, and a central element $K$:
\begin{equation}
\comm{a_m}{a_n} = m \delta_{m+n,0} K.
\end{equation}
As usual, $K$ is identified with a real number $k$ times the identity when acting on representations and the Virasoro modes $L_n$ then follow from the standard Sugawara construction.  Moreover, when a \hws{} in such a representation has weight ($a_0$-eigenvalue) $\lambda$, its conformal dimension is $\lambda^2 / 2k$.  Note that the algebra for $k \neq 0$ is almost always rescaled via $a_m \to a_m / \sqrt{k}$ so as to set $k$ to $1$.%
\footnote{If we wish to preserve the adjoint (reality condition) $a_m^{\dag} = a_{-m}$, then we may only rescale $k$ to $\pm 1$.  Free bosons with $k>0$ are often called \emph{euclidean} whereas those with $k<0$ are called \emph{lorentzian}.}

The irreducible ($k=1$) \hwms{} $\Fock{\lambda}$, called Fock spaces, have characters given by
\begin{equation} \label{ch:BosonNaive}
\fch{\Fock{\lambda}}{q} = \traceover{\Fock{\lambda}} q^{L_0 - c/24} = \frac{q^{\lambda^2/2}}{\func{\eta}{q}}.
\end{equation}
It is well known (see \cite{SchNon06} for example) that the S-transformations of these characters amount to a Fourier transform and that one can recover non-negative integer fusion multiplicities from them using a continuum version of the Verlinde formula.  The only problem with this is that the characters \eqref{ch:BosonNaive} do not completely distinguish the irreducible modules:  $\ch{\Fock{\lambda}}$ and $\ch{\Fock{-\lambda}}$ are identical.  Consequently, the application of the Verlinde formula cannot, strictly speaking, reproduce the structure constants of the fusion ring, but only of a quotient of the fusion ring by an action of the two-element group $\ZZ_2$.

The obvious fix is to include the affine weight $\lambda$ in the character.  Of course, then the S-transformation will produce an unwanted factor for which the standard remedy is to include $k$ in the character.  In this way, we arrive at the full character (for general $k$):
\begin{equation} \label{ch:Boson}
\fch{\Fock{\lambda}}{y;z;q} = \traceover{\Fock{\lambda}} y^k z^{a_0} q^{L_0 - c/24} = \frac{y^k z^{\lambda} q^{\lambda^2/2k}}{\func{\eta}{q}}.
\end{equation}
Writing $y = \ee^{2 \pi \ii t}$, $z = \ee^{2 \pi \ii u}$ and $q = \ee^{2 \pi \ii \tau}$, the modular S-transformation of the characters \eqref{ch:Boson} acts as $\modS \colon \left( t \middle\vert u \middle\vert \tau \right) \to \left( t - u^2 / 2 \tau \middle\vert u / \tau \middle\vert -1 / \tau \right)$, leading to
\begin{equation} \label{eq:BosonS}
\Sch{\Fock{\lambda}} = \int_{-\infty}^{\infty} \modS_{\lambda \mu} \ch{\Fock{\mu}} \: \frac{\dd \mu}{\sqrt{k}}, \qquad \modS_{\lambda \mu} = \ee^{-2 \pi \ii \lambda \mu / k}.
\end{equation}
This follows from a standard gaussian integration, convergent for $k>0$ (when $k<0$, we have to assume the standard result through an analytic continuation):
\begin{align}
\int_{-\infty}^{\infty} \modS_{\lambda \mu} \ch{\Fock{\mu}} \: \frac{\dd \mu}{\sqrt{k}} &= \frac{\ee^{2 \pi \ii k t}}{\func{\eta}{\tau}} \int_{-\infty}^{\infty} \ee^{\ii \pi \tau \mu^2 / k + 2 \pi \ii \brac{u - \lambda / k} \mu} \: \frac{\dd \mu}{\sqrt{k}} = \frac{\ee^{2 \pi \ii k t - \ii \pi k \brac{u - \lambda / k}^2 / \tau}}{\sqrt{-\ii \tau} \func{\eta}{\tau}} = \Sch{\Fock{\lambda}}.
\end{align}
We remark that the measure $\dd \mu / \sqrt{k}$ is natural given the rescaling property of the $a_n$.

For $\modT \colon \left( t \middle\vert u \middle\vert \tau \right) \to \left( t \middle\vert u \middle\vert \tau + 1 \right)$, the transformation is
\begin{equation} \label{eq:BosonT}
\Tch{\Fock{\lambda}} = \int_{-\infty}^{\infty} \modT_{\lambda \mu} \ch{\Fock{\mu}} \: \frac{\dd \mu}{\sqrt{k}}, \qquad \modT_{\lambda \mu} = \ee^{\ii \pi \brac{\lambda^2 / k - 1/12}} \func{\delta}{\frac{\lambda}{\sqrt{k}} = \frac{\mu}{\sqrt{k}}}.
\end{equation}
It is straight-forward to check that $\modS^2$ and $\brac{\modS \modT}^3$ are the conjugation permutation $\lambda \to -\lambda$, hence that the characters span a representation of the modular group $\SLG{SL}{2 ; \ZZ}$ (of uncountably-infinite dimension).

The S-matrix (or S-density) is symmetric and unitary with respect to the rescaled weights $\lambda / \sqrt{k}$:
\begin{equation}
\int_{-\infty}^{\infty} \modS_{\lambda \mu} \modS_{\mu \nu}^{\dag} \: \frac{\dd \mu}{\sqrt{k}} = \int_{-\infty}^{\infty} \ee^{-2 \pi \ii \brac{\lambda - \nu} \mu / k} \: \frac{\dd \mu}{\sqrt{k}} = \func{\delta}{\frac{\lambda}{\sqrt{k}} = \frac{\nu}{\sqrt{k}}}.
\end{equation}
It immediately follows that the diagonal partition function $\partfunc{diag.} = \int_{-\infty}^{\infty} \ahol{\ch{\Fock{\lambda}}} \ch{\Fock{\lambda}} \: \dd \lambda / \sqrt{k}$ is modular invariant (T-invariance is manifest).  Similarly, the invariance of the charge conjugation partition function $\partfunc{c.c.} = \int_{-\infty}^{\infty} \ahol{\ch{\Fock{\lambda}}} \ch{\Fock{-\lambda}} \: \dd \lambda / \sqrt{k}$ follows from unitarity and the symmetry $\modS_{\lambda, \mu} = \modS_{-\lambda, -\mu}$.

The continuum Verlinde formula states that the fusion coefficients are given by
\begin{equation}
\fuscoeff{\lambda \mu}{\nu} = \int_{-\infty}^{\infty} \frac{\modS_{\lambda \rho} \modS_{\mu \rho} \modS_{\nu \rho}^*}{\modS_{0 \rho}} \: \frac{\dd \rho}{\sqrt{k}} =  \int_{-\infty}^{\infty} \ee^{-2 \pi \ii \brac{\lambda + \mu - \nu} \rho / k} \: \frac{\dd \rho}{\sqrt{k}} = \func{\delta}{\frac{\nu}{\sqrt{k}} = \frac{\lambda}{\sqrt{k}} + \frac{\mu}{\sqrt{k}}},
\end{equation}
where we recognise that the vacuum module is $\Fock{0}$.  The predicted fusion rules are therefore
\begin{equation} \label{FR:Boson}
\Fock{\lambda} \fuse \Fock{\mu} = \int_{-\infty}^{\infty} \fuscoeff{\lambda \mu}{\nu} \Fock{\nu} \: \frac{\dd \nu}{\sqrt{k}} = \Fock{\lambda + \mu},
\end{equation}
agreeing perfectly with the known fusion rules.  Actually, what the Verlinde formula computes is the fusion rules at the level of the characters.  However, the free boson theory has the property that its irreducible modules have linearly independent characters, if we use \eqref{ch:Boson}, and every module in the spectrum is completely reducible.  It follows that character fusion and module fusion coincide for this theory.

An important feature of the spectrum of the free boson is that every irreducible module is a simple current, meaning that they have inverses in the fusion ring \cite{SchExt89,SchSim90}.  Moreover, if we exclude the fusion identity $\Fock{0}$, then the simple current $\Fock{r}$ has no (irreducible) fixed points.%
\footnote{A fixed point of a simple current is a module for which fusion with the simple current reproduces itself.}  %
We can use the group generated by a simple current $\Fock{r}$ to construct extended algebras in a canonical fashion (see \cite{RidSU206,RidMin07} for example):  The extension is obtained by promoting the fields associated to the fusion orbit $\bigoplus_{j \in \ZZ} \Fock{jr}$ to symmetry generators.  In other words, this direct sum of $\AKMA{gl}{1}$-irreducibles becomes the (irreducible) vacuum module of the extended algebra.  We restrict attention to extended algebras which are integer-moded (hence bosonic).  Comparing conformal dimensions of the states of $\Fock{jr}$ shows that this is equivalent to demanding that $r^2 \in 2 \ZZ$ (we have scaled $k$ to $1$ for simplicity).%
\footnote{The extended algebra constructed from $\Fock{r}$ is, of course, the symmetry algebra of the free boson compactified on a circle of radius $r$.}

The irreducible modules of the extended algebra are also obtained as fusion orbits.  We denote them by
\begin{equation}
\ExtFock{\lambda} = \bigoplus_{j \in \ZZ} \Fock{\lambda + jr},
\end{equation}
where $\sqbrac{\lambda} = \lambda \bmod{r}$.  Requiring that the extended algebra act with integer moding leads to a finite set of (untwisted) extended algebra modules labelled by $\lambda = m/r$, with $m \in \set{0,1,\ldots,r^2-1}$.  The S-transformations of the characters of these modules follow readily from \eqref{eq:BosonS}:
\begin{align}
\Sch{\ExtFock{m/r}} &= \sum_{j \in \ZZ} \Sch{\Fock{m/r + jr}} = \int_{-\infty}^{\infty} \ee^{-2 \pi \ii m \mu / r} \sum_{j \in \ZZ} \ee^{-2 \pi \ii jr \mu} \ch{\Fock{\mu}} \: \dd \mu \notag \\
&= \frac{1}{r} \sum_{n=0}^{r^2-1} \sum_{j \in \ZZ} \ee^{-2 \pi \ii mn / r^2} \ch{\Fock{n/r + jr}} = \frac{1}{r} \sum_{n=0}^{r^2-1} \ee^{-2 \pi \ii mn / r^2} \ch{\ExtFock{n/r}}.
\end{align}
Here, we have applied the following summation formula:
\begin{equation}
\sum_{j \in \ZZ} \ee^{-2 \pi \ii jr \mu} = \sum_{\ell \in \ZZ} \tfunc{\delta}{r \mu = \ell} = \frac{1}{r} \sum_{\ell \in \ZZ} \tfunc{\delta}{\mu = \frac{\ell}{r}} = \frac{1}{r} \sum_{n=0}^{r^2-1} \sum_{j \in \ZZ} \tfunc{\delta}{\mu = \frac{n}{r} + jr}.
\end{equation}

The extended algebra's S-matrix is therefore given by
\begin{equation}
\ExtmodS{r}_{mn} = \frac{1}{r} \ee^{-2 \pi \ii mn / r^2}, \qquad m,n \in \set{0,1,\ldots,r^2-1}.
\end{equation}
This is again symmetric and unitary, so one can construct a diagonal modular invariant partition function $\partfunc{$r$} = \sum_{m=0}^{r^2-1} \ahol{\ch{\ExtFock{m/r}}} \ch{\ExtFock{m/r}}$ and its charge conjugate version.  Expressing this in terms of Fock space characters, one obtains an infinite set of non-diagonal modular invariants for $\AKMA{gl}{1}$ with discrete spectra.  Finally, it is easy to check that the (standard) Verlinde formula for the extended algebra gives non-negative integer coefficients:  $\Extfuscoeff{r}{m/r}{n/r}{p/r} = \delta_{p = m+n \bmod{r^2}}$.

Thus far, we have seen that the modular S-transformations of the free boson characters may be used to compute the S-transformations of those of the extended algebras.  These in turn can then be used to compute the Verlinde formula for the extended theories.  In a sense though, this is overkill because simple current technology makes it possible to reproduce the extended algebra fusion rules from those of the free boson.  Na\"{\i}vely, one might try
\begin{equation}
\ExtFock{\lambda} \fakefuse \ExtFock{\mu} = \bigoplus_{i,j \in \ZZ} \brac{\Fock{\lambda + ir} \fuse \Fock{\mu + jr}} = \bigoplus_{i \in \ZZ} \bigoplus_{j \in \ZZ} \Fock{\lambda + \mu + \brac{i+j} r} = \bigoplus_{j \in \ZZ} \ExtFock{\lambda + \mu}.
\end{equation}
However, this gives an overall multiplicity of infinity, even when $\lambda = \mu = 0$. The reason is that each of the $\Fock{\lambda + ir}$ are in the same module for the extended algebra, hence each of these Fock spaces gives exactly the same contribution to the fusion product.  It is therefore necessary to choose a single representative $\Fock{\lambda + ir}$, a convenient one has $i=0$, to avoid multiply counting the same information.  This ``renormalisation'' leads to
\begin{equation}
\ExtFock{\lambda} \fuse \ExtFock{\mu} = \bigoplus_{j \in \ZZ} \brac{\Fock{\lambda} \fuse \Fock{\mu + jr}} = \bigoplus_{j \in \ZZ} \Fock{\lambda + \mu + jr} = \ExtFock{\lambda + \mu},
\end{equation}
fixing the multiplicity issue. We therefore arrive at a very powerful strategy to compute the fusion rules of extended theories which may be summarised as follows:
\begin{itemize}
\item Compute the modular S-transformation of the (non-rational) theory with continuous spectrum.
\item Deduce fusion rules using the continuum Verlinde formula.
\item Use these fusion rules to identify simple current extensions with discrete (finite) spectrum.
\item Extract the fusion rules of the extended theory from those of the non-rational theory. 
\end{itemize}

We conclude this exercise by checking these extended algebra results at the self-dual radius $r = \sqrt{2}$, for which it is well known that the extended algebra is $\AKMA{sl}{2}$ at level $1$.  There are $r^2 = 2$ extended algebra modules $\ExtFock{0}$ and $\ExtFock{1 / \sqrt{2}}$ which are easily checked to have $1$ and $2$ ground states of dimensions $0$ and $1/4$, respectively.  The S-matrix and fusion matrices are found to be
\begin{equation}
\ExtmodS{\sqrt{2}} = \frac{1}{\sqrt{2}} 
\begin{pmatrix}
1 & 1 \\
1 & -1
\end{pmatrix}
; \qquad \Extfusmat{\sqrt{2}}{0} = 
\begin{pmatrix}
1 & 0 \\
0 & 1
\end{pmatrix}
, \quad \Extfusmat{\sqrt{2}}{1 / \sqrt{2}} = 
\begin{pmatrix}
0 & 1 \\
1 & 0
\end{pmatrix}
,
\end{equation}
which does indeed reproduce the correct data for $\AKMA{sl}{2}_1$.

\subsection{Outline}

Our review commences in \secref{sec:Perc} with an overview of the $c=0$ \lcft{} that describes the critical point of the statistical lattice model known as percolation.  We describe enough of the underlying lattice theory to introduce Cardy's celebrated formula \cite{CarCri92} for a non-local observable known as the horizontal crossing probability.  While it is clear that Cardy's derivation cannot be accommodated within a unitary theory (the only unitary $c=0$ \cft{} is the trivial minimal model $\minmod{2}{3}$), we show that his derivation actually implies a logarithmic theory, following \cite{RidPer07}.  This necessitates a brief introduction to the famous fusion algorithm of Nahm, Gaberdiel and Kausch \cite{NahQua94,GabInd96}.  We compute a few fusion products explicitly before describing the results of more involved calculations that detail the structures of the indecomposable modules so constructed.  We then use the results to derive a couple of logarithmic correlators, generalising the analysis of \secref{sec:LogCorr}, before briefly discussing other non-local percolation observables and other $c=0$ models.

\secref{sec:Trip} introduces the first of our ``archetypal'' \lcfts{}, the symplectic fermions of Kausch \cite{KauCur95}.  More precisely, we discuss a family of $c=-2$ theories which include symplectic fermions, the triplet model $\TripAlg{1,2}$ studied in \cite{Kausch:1991,Flohr:1995ea,GabRat96}, and the corresponding singlet model $\SingAlg{1,2}$, itself a special case of Zamolodchikov's original W-algebra \cite{ZamInf85}.  We begin with the symplectic fermion algebra, constructing its irreducible and indecomposable (twisted) representations and verifying the non-diagonalisability of $L_0$ on the latter, before decomposing its representations into those of the subalgebras $\TripAlg{1,2}$ and $\SingAlg{1,2}$.  For the singlet, the spectrum is continuous and it is here that we derive character formulae and deduce modular transformations.  The S-matrix is found to be symmetric and unitary, so we apply a continuous version of the Verlinde formula and find that the resulting (Grothendieck) fusion coefficients are positive integers.

As far as we are aware, the modular properties of the singlet model's characters are new (the generalisation to $\SingAlg{1,p}$ will be reported in \cite{CreW1p13}).  Assuming that the continuous Verlinde formula does give the correct (Grothendieck) fusion coefficients, we also deduce many fusion rules, in particular concluding that the singlet model possesses a countable infinity of simple currents.  We identify the maximal simple current extension as symplectic fermions and the maximal bosonic extension as the triplet model.  This also seems to be new.  We moreover use our singlet results to determine what the (Grothendieck) fusion rules for the triplet model should be, finding agreement with the fusion computations of \cite{GabRat96}.  This then provides a stringent consistency check of the continuous Verlinde formula.  We also conjecture the existence of certain singlet indecomposable modules before briefly discussing the known issue with obtaining an S-matrix for the triplet model (the S-matrix entries are not constant) and how this is manifested in the simple current extension formalism we have developed.

Finally, we discuss the bulk (non-chiral) aspects of these $c=-2$ theories.  Bulk \lcfts{} are not as well understood as their chiral counterparts, though progress has been steady \cite{GabLoc99,Schomerus:2005bf,Saleur:2006tf,GabFro08,GabMod11,RunLog12,GaiAss12}.  For this, it has proven useful to study analogous situations in mathematics.  For example, the representation of a semisimple finite-dimensional associative algebra, where it is acting on itself by left-multiplication, decomposes as a direct sum of irreducibles, with every irreducible appearing with multiplicity equal to its dimension (Wedderburn's theorem).  However, the non-semisimple case gives a direct sum of projectives, where the multiplicity of each is now the dimension of the irreducible it covers.  The semisimple case is also the result for compact Lie groups $\grp{G}$ acting on the Hilbert space $\func{L^2}{\grp{G} , \mu}$ (with $\mu$ the Haar measure), whereas the non-semisimple case seems to be roughly correct for Lie supergroups and many non-compact groups (this is the \emph{minisuperspace limit} \cite{MalStr01,Schomerus:2005bf}).

This is relevant because the modular invariant partition functions that have been constructed for \lcfts{} often have the form
\begin{equation}
\partfunc{} = \sum_i \ahol{\ch{L_i}} \ch{P_i} = \sum_i \ahol{\ch{P_i}} \ch{L_i},
\end{equation}
where $i$ labels the irreducibles $L_i$ in the spectrum and $P_i$ denotes an indecomposable cover of $L_i$ which one expects to be projective in some category.  (In the rational case, each $P_i$ and $L_i$ coincide and this reduces to the standard diagonal invariant.)  Because of this, the state space of a logarithmic theory seems likely to decompose, upon restricting to the chiral or antichiral algebra, into a direct sum of projectives (each with infinite multiplicity).

Digressions over, we conclude our discussion of this family of $c=-2$ \lcfts{} by noting the well known structure for the bulk state space of the symplectic fermions theory and discuss the structure one obtains by restricting to the triplet algebra \cite{GabLoc99}.  We then indicate how one could have guessed this structure, without \emph{a priori} knowledge of the symplectic fermions structure, based on the form of the diagonal modular invariant partition function and the above analogy.  This leads to a simple proposal for constructing bulk module structures from chiral ones which we stress automatically satisfies the physical \emph{locality} requirement that bulk correlators are single-valued.  Algebraically, this will be met if the chiral and antichiral states have conformal dimensions that differ by an integer and, in a logarithmic theory, if the spin operator $L_0 - \ahol{L}_0$ acts diagonalisably.  We conclude by computing some correlation functions and briefly mentioning what is known about the more general $\TripAlg{1,p}$ and $\TripAlg{q,p}$ theories that have received so much attention in the literature.

\secref{sec:SL2} considers an example of a fractional level \WZW{} model, specifically that whose symmetry algebra is $\AKMA{sl}{2}$ at level $k = -\tfrac{1}{2}$.  The existence of such fractional level theories was first suggested by Kent \cite{KenInf86} in order to provide a unified coset construction of all Virasoro minimal models, unitary and non-unitary.  They began to be studied seriously once Kac and Wakimoto discovered \cite{KacMod88} that the levels required for Kent's cosets, the \emph{admissible levels}, were the only ones which admitted modules whose characters carried a finite-dimensional representation of the modular group $\SLG{SL}{2 ; \ZZ}$.  Assuming, naturally enough, that this meant that admissible level models were rational, Koh and Sorba computed the fusion rules given by the Verlinde formula, noting that this sometimes resulted in negative integer fusion coefficients \cite{KohFus88}.  This puzzle was subsequently addressed by many groups \cite{BerFoc90,MatFra90,AwaFus92,RamNew93,FeiFus94,AndOpe95,PetFus96,DonVer97,FurAdm97,MatPri99}, without any real progress, before Gaberdiel pointed out \cite{GabFus01} that the assumption of rationality was in error (see also \cite{AdaVer95}).  He constructed enough fusion products for $\AKMA{sl}{2}$ at level $k=-\tfrac{4}{3}$ to conclude that the theory was logarithmic, but was unable to solve the puzzle of negative fusion multiplicities.  The level $k=-\tfrac{1}{2}$ was subsequently argued to be logarithmic using a free field realisation \cite{LesSU202,LesLog04}, but a complete picture including indecomposable module structure, characters, modular properties and the Verlinde formula has only recently emerged \cite{RidSL208,RidSL210,RidFus10,CreMod12,CreMod13}.  The purpose of \secref{sec:SL2} is to explain this progress for $k=-\tfrac{1}{2}$.

We start by introducing the closely related $\beta \gamma$ ghost system and derive the current algebra $\AKMA{sl}{2}_{-1/2}$ as an orbifold.  Instead of considering the representation theory of the ghost algebra, as we did for \secref{sec:Trip}, we determine the spectrum of $\AKMA{sl}{2}_{-1/2}$ directly.  As before, we find a continuum of generically irreducible modules, but this time they are neither highest nor lowest weight.  At the parameter values where the continuum modules become reducible, four \hwms{} are constructed (these are the admissible modules of Kac and Wakimoto).  The characters of these admissibles can be meromorphically continued using Jacobi theta functions, leading to a four-dimensional representation of the modular group.  We then illustrate the paradox of negative Verlinde fusion coefficients before indicating its resolution \cite{RidSL208} using spectral flow automorphisms.

A very important point here is that one must be careful with regions of convergence of characters.  Indeed, certain non-isomorphic modules, related by spectral flow, have (up to a sign) exactly the same meromorphically-continued character.  However, the regions where these characters converge are disjoint, being separated by a common pole.  We then interpret the sum of these characters as a distribution supported at this pole.  With this formalism, we obtain modular properties, apply the Verlinde formula and construct a discrete series of modular invariants.  Here, the story is very similar to the previous section and we are again able to propose a structure for the local bulk modules.  Finally, we use a free field realisation to compute correlation functions and conclude with a brief discussion of how all this generalises to the level $k=-\tfrac{4}{3}$.
 
\secref{sec:GL11} addresses the last of the ``archetypal'' examples, the \WZW{} theory on the Lie supergroup $\SLSG{GL}{1}{1}$.  As usual, supergroup models depend upon a level $k$ and the symmetry algebra is an affine Lie superalgebra.  But, in contrast to (integer level) bosonic \WZW{} models, our understanding of these superanalogues is still rather rudimentary.  Aside from the \emph{rational} theories associated with $\SLSG{OSP}{1}{2n}$, see \cite{EnnOSP97} for example, only the theories associated to the Lie supergroups $\SLSG{GL}{1}{1}$ and $\SLSG{PSL}{1}{1}$ (which is just symplectic fermions) are completely understood.  Indeed, these were the first logarithmic conformal field theories investigated over two decades ago.

We structure this section so as to bring out the analogy with the previous two examples.  We start with the algebra and representation theory, then continue with modular data and correlation functions following \cite{Schomerus:2005bf,Saleur:2006tf,Creutzig:2011cu}.  This example, like the two that preceded it, exhibit all the features of the simplest known logarithmic conformal field theories.  There are certainly many more logarithmic theories that should be considered, some with similar indecomposable structures to our examples, and some which are more complicated.  We mention that there are many applications which involve supergroup theories of the latter class as, for example, in statistical physics \cite{Zirnbauer:1999ua,Guruswamy:1999hi} and the AdS/CFT correspondence \cite{Berkovits:1999im} --- these will therefore need a detailed investigation in the near future. 

\secref{sec:Stag} aims to briefly outline a reasonably general approach to understanding the mathematical structures that underlie \lcft{}.  It commences with a somewhat technical discussion which introduces the important idea of a \emph{staggered module}, familiar from Virasoro studies \cite{RohRed96,RidSta09}, for a large class of associative algebras.  Some very basic results are proven at this level of generality (these results have not before been published) before restricting to a discussion of the \emph{logarithmic couplings} that parametrise the isomorphism classes of staggered modules for the Virasoro algebra.  We emphasise that these numbers are as important to \lcft{} as the three-point constants are to rational theories and we detail how they arise when computing two-point functions.  We conclude with a brief analysis of an example of a Virasoro indecomposable whose structure is more complicated than that of a staggered module because $L_0$ acts with a \emph{rank $3$ Jordan block}.

\secref{sec:Conc} then summarises what we have presented, describing a proposed approach to understanding quite general classes of \lcfts{}.  Finally, we provide a short appendix in which we have collected some of the necessary basic information about homological algebra, a very useful tool (and language) for describing the structure of indecomposable but reducible representations.

\section{Percolation as a Logarithmic Conformal Field Theory} \label{sec:Perc}

Percolation may be loosely defined as a collection of closely related probabilistic models whose observed behaviour is believed to be reasonably typical for more general classes of statistical theories.  In particular, these models exhibit phase transitions as their defining parameters pass through certain critical values \cite{KesPer82}.  Moreover, percolation is particularly easy to simulate numerically, so it is a popular choice for testing predictions such as conformal invariance and universality \cite{LanUni92,LanCon94}.  In this section, we discuss how 
the hypothesis of conformal invariance, which led Cardy to his celebrated formula \cite{CarCri92} for the horizontal crossing formula, can be accommodated within the standard framework of (boundary) \cft{}.  It has long been suspected (see \cite{CarLog99} for example) that the conformal invariance of percolation requires a logarithmic theory.  Here, we follow \cite{RidPer07} to deduce from the assumptions underlying Cardy's derivation that the spectrum of percolation contains indecomposable modules on which the Virasoro mode $L_0$ acts non-diagonalisably, hence that critical percolation is described by a \lcft{}.

\subsection{Critical Percolation and the Crossing Formula} \label{sec:Crossing}

As with many other statistical models, the primary consideration of percolation is the degree to which a very large number of identical objects tend to cluster together when distributed in a random fashion.  The setup for one of the basic percolation models is as follows:  Consider a square lattice with a given edge length and choose a fixed rectangular subdomain whose sides are a union of lattice edges.  A percolation configuration is then obtained by declaring that each edge within the subdomain is open with probability $p$ and closed with probability $1-p$.  The idea is that the subdomain represents a porous material and that open edges permit the flow of a liquid medium whereas closed edges do not.  When $p=0$, all edges are closed and the material is impermeable to the liquid.  When $p=1$, all edges are open and there is no obstruction to the liquid's flow.  For $0<p<1$, one is then led to question whether the liquid is able to percolate through the material and it is this, of course, that gives the model its name.

To be more precise, we may ask for the probability that a randomly chosen configuration of edges in our rectangular subdomain contains an open path connecting a chosen side of the rectangle with the opposite side.  Such a path is called a \emph{crossing} and \figref{fig:perc} shows an example of a configuration in which one (of the many) crossings has been drawn.  Computing this crossing probability analytically is a hopeless task, though simulation can approximate it extremely well.  However, one can ask the question again in the continuum limit where the edge length tends to $0$ while the size and shape of the rectangular subdomain is kept fixed.  In this case, one has the result \cite{KesPer82} that the limit of the crossing probabilities is $0$ if $p$ is less than a critical value, which turns out to be $p_c = \frac{1}{2}$ for a square lattice, and is $1$ if $p$ is greater than $p_c$.  The only interesting value is then the limit of the crossing probabilities when $p$ is precisely this critical value.%
\footnote{Curiously, it seems that the existence of this limit when $p=p_c$ was not known until Cardy's crossing formula (see \eqref{eq:Cardy}) was rigorously proven \cite{SchSca00,SmiCri01}.}

\begin{figure}
\begin{center}
\includegraphics[width=0.3\textwidth]{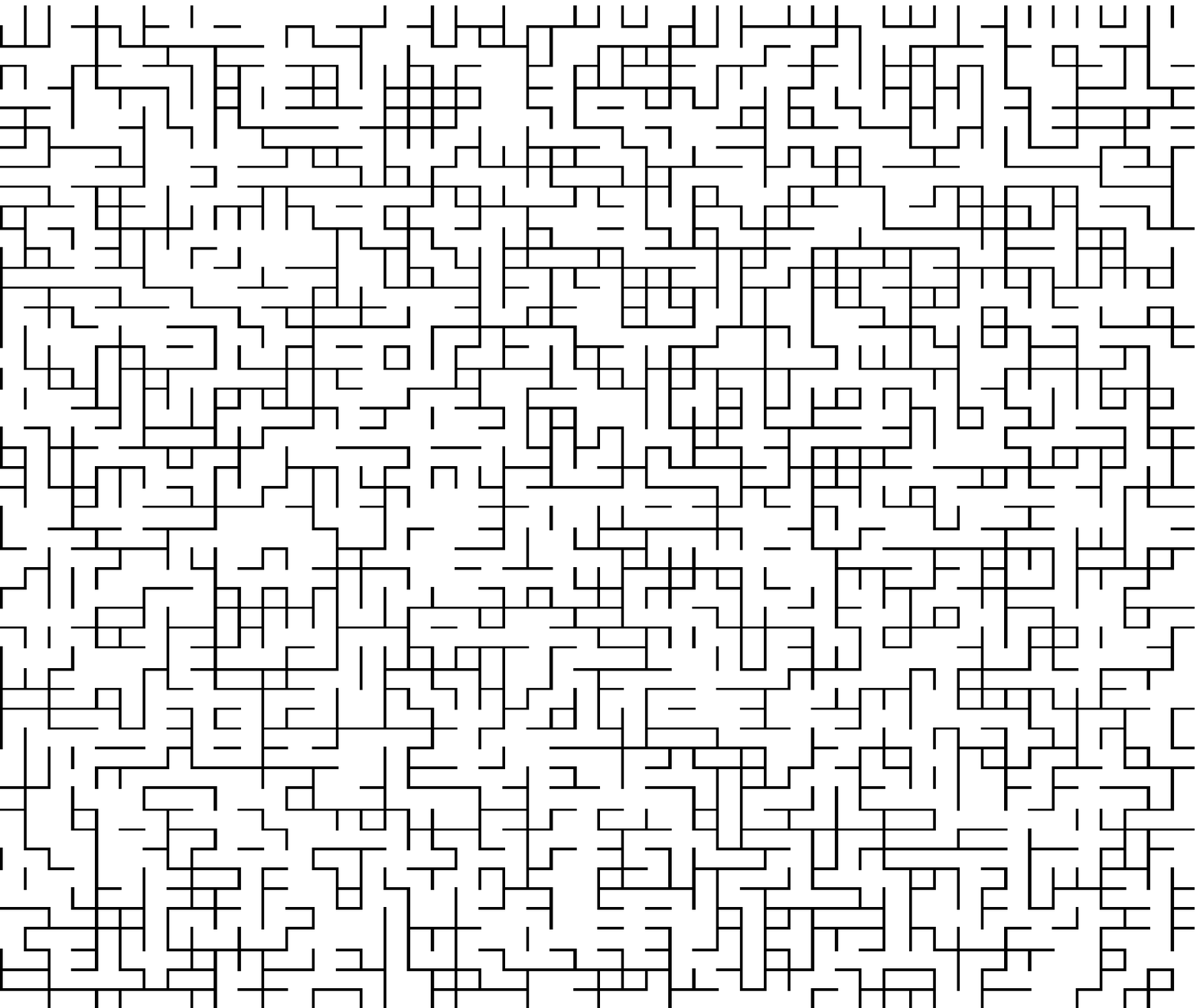} \hspace{0.1\textwidth} 
\includegraphics[width=0.3\textwidth]{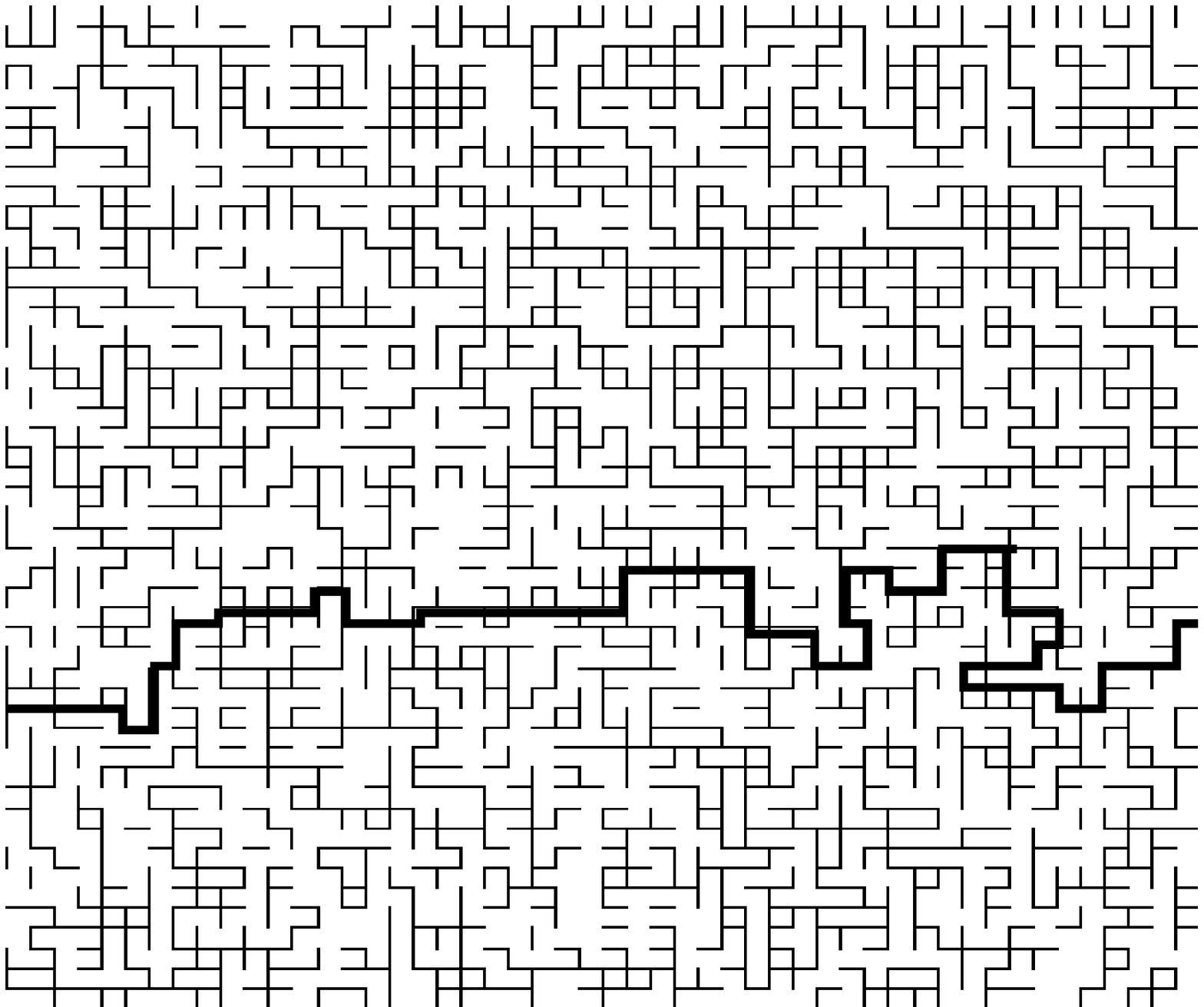}
\caption{\label{fig:perc}A typical percolation configuration (left) for a rectangular subdomain of a square lattice showing only the open edges (closed edges are omitted).  This lattice has several crossings from left to right, one of which is indicated in bold (right).} 
\end{center}
\end{figure}

This probability of a crossing being present when $p = p_c$ was famously derived by Cardy \cite{CarCri92} within the formalism of boundary \cft{} and his result is generally recognised as one of the most striking confirmations of the conjecture that the continuum limit of a statistical model is conformally invariant at its critical points.  Cardy combined the well known description of percolation as the $Q \ra 1$ limit of the $Q$-state Potts model with an inspired identification of certain boundary-changing operators in these Potts models to write the crossing probability as the four-point correlation function of a Virasoro primary field $\phi_{1,2}$, where the subscript indicates the field's Kac labels.  To apply the machinery of \cft{}, one now maps the rectangular subdomain conformally onto (a compactification of) the upper half-plane so that the fields' insertion points (the corners of the rectangle) are mapped to points $z_i$, $i=1,\ldots,4$, lying on the real axis (or to $\infty$).

The central charge $c$ of the continuum limit of the $Q$-state Potts model is well known \cite{DotCon84},
assuming of course that the limit is conformally invariant.  For percolation ($Q \ra 1$), one obtains $c=0$ and it therefore follows that $\phi_{1,2}$ has conformal dimension $0$.  Moreover, $\phi_{1,2}$ will have a singular descendant field at grade $2$ and so, according to standard \cft{} dogma, the four-point correlator representing the crossing probability will satisfy a second-order linear differential equation.  The obvious behaviour of the crossing probability as the aspect ratio of the rectangle tends to $0$ and $\infty$ then picks out a unique solution:
\begin{equation} \label{eq:Cardy}
\text{Pr} = \frac{\func{\Gamma}{\frac{2}{3}}}{\func{\Gamma}{\frac{1}{3}} \func{\Gamma}{\frac{4}{3}}} \eta^{1/3} \func{{}_2 F_1}{\frac{1}{3} , \frac{2}{3} ; \frac{4}{3} ; \eta}, \qquad \text{where} \quad \eta = \frac{\brac{z_1 - z_2} \brac{z_3 - z_4}}{\brac{z_1 - z_3} \brac{z_2 - z_4}}.
\end{equation}
The agreement between this computation and numerical data from simulations \cite{LanUni92} is impressive.

The precise formula for the crossing probability is not important for what follows.  Rather, what we wish to emphasise is that the derivation is performed with the aid of a limit $Q \ra 1$ which hides a remarkable amount of subtlety.  Indeed, one might guess that the percolation \cft{} is a minimal model, based on the usual identification of the $Q$-state Potts models for $Q=2$ and $3$ with $\minmod{3}{4}$ and $\minmod{5}{6}$, respectively.  However, the minimal model with $c=0$ is $\minmod{2}{3}$ which is trivial in the sense that its field content is limited to constant multiples of the identity.  Obviously, four-point functions in $\minmod{2}{3}$ will be constant, so this model cannot accommodate Cardy's derivation.  On the other hand, it would be distressing if Cardy's derivation turned out to be inconsistent with the principles of \cft{}.  We will therefore assume that a description of critical percolation can be accommodated within \cft{}.  As we shall see, this will require the consideration of reducible, yet indecomposable, representations.

\subsection{The Necessity of Indecomposability} \label{sec:PercIndec}

Before embarking on our explorations, let us pause to recall some useful facts concerning Virasoro modules.  This will also serve to introduce our notation.  Verma modules will be denoted by $\VirVer{h}$, where $h$ is the conformal dimensional of the \hws{}, and their irreducible quotients by $\VirIrr{h}$.  For $c=0$, we recall that the Verma module is itself irreducible unless $h = h_{r,s}$ for some $r,s \in \ZZ_+$, where
\begin{equation}
h_{r,s} = \frac{\brac{3r-2s}^2 - 1}{24}.
\end{equation}
In the latter case, $\VirVer{h} = \VirVer{h_{r,s}}$ will have a submodule generated by a singular vector at grade $rs$ (its conformal dimension will be $h_{r,s} + rs$).  If $r$ is even or $s$ is a multiple of $3$, then the maximal proper submodule of $\VirVer{h_{r,s}}$ is generated by the \sv{} of lowest (positive) grade.%
\footnote{We will often use the term ``\sv{}'' to indicate a \hws{} which is a proper descendant.  Similarly, the term ``\hws{}'' will often be used to indicate the one of lowest conformal dimension in a given module.}  %
Otherwise, it is generated by the two \svs{} of lowest and next-to-lowest grades.  It is convenient to collate the $h_{r,s}$ with $r,s \in \ZZ_+$ into an \emph{extended Kac table}, a part of which we reproduce in \tabref{tab:Kac}.  Finally, we introduce notation for certain Verma module quotients that will frequently arise in what follows:  $\VirQuot{r,s} = \VirVer{h_{r,s}} / \VirVer{h_{r,s} + rs}$.

\begin{table}
\begin{center}
\setlength{\extrarowheight}{4pt}
\begin{tabular}{|C|C|C|C|C|C|C|C|C|C|C}
\hline
0 & 0 & \frac{1}{3} & 1 & 2 & \frac{10}{3} & 5 & 7 & \frac{28}{3} & 12 & \cdots \\[1mm]
\hline
\frac{5}{8} & \frac{1}{8} & \frac{-1}{24} & \frac{1}{8} & \frac{5}{8} & \frac{35}{24} & \frac{21}{8} & \frac{33}{8} & \frac{143}{24} & \frac{65}{8} & \cdots \\[1mm]
\hline
2 & 1 & \frac{1}{3} & 0 & 0 & \frac{1}{3} & 1 & 2 & \frac{10}{3} & 5 & \cdots \\[1mm]
\hline
\frac{33}{8} & \frac{21}{8} & \frac{35}{24} & \frac{5}{8} & \frac{1}{8} & \frac{-1}{24} & \frac{1}{8} & \frac{5}{8} & \frac{35}{24} & \frac{21}{8} & \cdots \\[1mm]
\hline
\vdots & \vdots & \vdots & \vdots & \vdots & \vdots & \vdots & \vdots & \vdots & \vdots & \ddots
\end{tabular}
\vspace{3mm}
\caption{A part of the extended Kac table for $c=0$, displaying the conformal dimensions $h_{r,s}$ for which the Verma modules $\VirVer{h_{r,s}}$ are reducible.  The rows of the table are labelled by $r = 1, 2, 3, \ldots$ and the columns by $s = 1, 2, 3, \ldots$.} \label{tab:Kac}
\end{center}
\end{table}

We begin by postulating that the \cft{} describing the continuum limit of critical percolation contains a vacuum $\ket{0}$.  Equivalently, by the state-field correspondence, the identity field $I$ is present in the theory.  As $h_{1,1} = h_{1,2} = 0$, the vacuum Verma module $\VirVer{0}$ has \svs{} at grades $1$ and $2$ and these turn out to be independent in the sense that the latter is not descended from the former.%
\footnote{There are also singular vectors at grades $5,7,12,15,\ldots$ which are each descended from both the grade $1$ and grade $2$ \svs{}.}  %
In fact, the maximal proper submodule of $\VirVer{0}$ is generated by these two \svs{}.  Since $L_{-1} \ket{0}$ corresponds to the field $\pd I = 0$, we have to set $L_{-1} \ket{0} = 0$ (by quotienting $\VirVer{0}$ by the submodule it generates).  However, the grade $2$ \sv{} then corresponds to the energy-momentum tensor $\func{T}{z}$.  If this is set to $0$, then each of its modes, the Virasoro generators $L_n$, must all act as the zero operator on the states of the theory, and this leads us to the (trivial) minimal model $\minmod{2}{3}$ (or a direct sum of copies of this model).

To get a non-trivial $c=0$ theory, we must abandon the idea that \svs{} are always set to $0$.  Instead of assuming that the vacuum $\ket{0}$ generates the irreducible Virasoro module $\VirIrr{0}$, we are led to propose that the vacuum module is the reducible, but indecomposable, quotient $\VirQuot{1,1} = \VirVer{0} / \VirVer{1}$.  This is, in fact, the only remaining option because the only \sv{} to survive in $\VirQuot{1,1}$ has grade $2$ (it corresponds to $T$) and setting it to $0$ leads back to the irreducible vacuum module $\VirIrr{0}$.  In the language of \appref{app:Exact}, our proposed vacuum module $\VirQuot{1,1}$ is an extension of $\VirIrr{0}$ by the submodule generated by $\ket{T} = L_{-2} \ket{0}$, which is itself irreducible and isomorphic to $\VirIrr{2}$.  This is summarised by the exact sequence (see \appref{app:Exact})
\begin{equation}
\dses{\VirIrr{2}}{}{\VirQuot{1,1}}{}{\VirIrr{0}}.
\end{equation}

This argument shows that there is a unique choice for the vacuum module which leads to a non-trivial theory.  Moreover, this choice is reducible, but indecomposable.  To accommodate Cardy's derivation, there should also exist in the theory a primary field $\phi_{1,2}$ with a vanishing grade $2$ descendant.  This last requirement stems from the fact that the crossing probability is derived as a solution to a second order differential equation and this equation is derived from the vanishing of a grade $2$ descendant.  Because $h_{1,2} = 0$, the corresponding Verma module is again $\VirVer{0}$ with \svs{} at grades $1$ and $2$.  This time, we cannot set the grade $1$ \sv{} to $0$ because it would lead to a first order differential equation for Cardy's crossing probability (one can check that the solutions to this equation are all constant).  We therefore conclude that the reducible, but indecomposable, module $\VirQuot{1,2} = \VirVer{0} / \VirVer{2}$ is present.  Again, this is the only possibility compatible with Cardy's derivation; the exact sequence is
\begin{equation} \label{ses:Q12}
\dses{\VirIrr{1}}{}{\VirQuot{1,2}}{}{\VirIrr{0}}.
\end{equation}

This concludes the basic setup for a \cft{} which is consistent with Cardy's derivation of the crossing formula \eqref{eq:Cardy}.  One can therefore declare with confidence that the percolation (boundary) \cft{}, whatever it may be, must include the indecomposable vacuum module $\VirQuot{1,1}$ and the indecomposable module $\VirQuot{1,2}$ in appropriate boundary sectors.  It remains to explore the consequences of this conclusion.  As usual, one can try to generate new field content through fusion.  It is natural to expect that the identity field $I$ will act as the fusion identity ($I \fuse I = I$ and $I \fuse \phi_{1,2} = \phi_{1,2}$) and this is indeed the case.  One also expects that the vanishing of the grade $2$ singular descendant of $\phi_{1,2}$ will imply that
\begin{equation}
\phi_{1,2} \fuse \phi_{1,2} = I + \phi_{1,3},
\end{equation}
where $\phi_{1,3}$ is a Virasoro primary field of conformal dimension $h_{1,3} = \frac{1}{3}$.  This also turns out to be true.  However, the natural sequel to this computation,
\begin{equation}
\phi_{1,2} \fuse \phi_{1,3} = \phi_{1,2} + \phi_{1,4},
\end{equation}
where $\phi_{1,4}$ is primary of dimension $h_{1,4} = 1$, is \emph{false} as we shall see.

\subsection{The Nahm-Gaberdiel-Kausch Fusion Algorithm} \label{sec:NGK}

In standard \cft{}, where the modules are completely reducible, it is permissible to regard fusion as an operation on primary fields, remembering that the fusion rules in fact also apply to the entire family of fields descended from the respective primaries.  However, we have already surmised that there are reducible, but indecomposable, modules in the percolation spectrum.  Therefore, one needs to be much more precise about fusion and regard it not as an operation on primaries, but rather as an operation on the modules themselves.  We also need to be more careful about how fusion rules are computed.  The usual method of examining the effect of setting \svs{} to zero on three-point functions might not be practical if we do not know what type of fields to insert in the three-point functions (as we shall see, primary fields do not suffice in general).

The standard method of computing fusion rules when reducible, but indecomposable, modules are involved is known as the \emph{Nahm-Gaberdiel-Kausch} algorithm.  This was originally introduced by Nahm \cite{NahQua94} in a limited setting and was extended (and applied to indecomposable Virasoro modules at $c=-2$) by Gaberdiel and Kausch \cite{GabInd96}.  The key insight behind this algorithm is the realisation that one can concretely realise the fusion product of two modules $M$ and $N$ as a quotient of the vector space tensor product $M \otimes_{\CC} N$.  To demonstrate this, one needs to know how the action of the symmetry algebra (here, the Virasoro algebra) on $M \fuse N$ is derived from the actions on $M$ and on $N$.  This takes the form of coproduct formulae \cite{GabFus94}:
\begin{subequations} \label{eq:NGK}
\begin{align}
\func{\Delta}{L_n} &= \sum_{m=-1}^n \binom{n+1}{m+1} L_m \otimes \id + \id \otimes L_n & &\text{($n \geqslant -1$),} \label{eq:NGK1} \\
\func{\Delta}{L_{-n}} &= \sum_{m=-1}^{\infty} \brac{-1}^{m+1} \binom{n+m-1}{m+1} L_m \otimes \id + \id \otimes L_{-n} & &\text{($n \geqslant 2$),} \label{eq:NGK2} \\
L_{-n} \otimes \id &= \sum_{m=n}^{\infty} \binom{m-2}{m-n} \func{\Delta}{L_{-m}} + \brac{-1}^n \sum_{m=-1}^{\infty} \binom{n+m-1}{m+1} \id \otimes L_m & &\text{($n \geqslant 2$).} \label{eq:NGK3}
\end{align}
\end{subequations}
Actually, one derives two distinct coproducts which should coincide --- \eqref{eq:NGK3} is then deduced by imposing this equality.  Of course, there are generalisations of these formulae for other symmetry algebras \cite{GabFus94b}.

Practically, one does not compute explicitly with the entire fusion module $M \fuse N$.  Rather, one restricts attention to a subspace by setting all states of sufficiently high grade to $0$.  More precisely, if $g$ is the cutoff grade, then any state which can be written as a linear combination of states of the form $L_{-n_1} \cdots L_{-n_k} \ket{v}$, with $n_1 + \cdots n_k > g$, is set to $0$.  We will denote the result of this grade $g$ truncation of a Virasoro module $N$ by $N^{\brac{g}}$.  This truncation not only replaces the infinite-dimensional fusion product by a finite-dimensional subspace, thereby facilitating explicit computation, but it also renders the first sum in \eqref{eq:NGK3} finite (the other sums in \eqref{eq:NGK} are already effectively finite if we assume that the conformal dimensions of the states of $M$ and $N$ are bounded below).  The point is that this truncation is compatible with fusion computations because \eqref{eq:NGK} may be used to prove that $\brac{M \times N}^{\brac{g}}$ can be realised as a quotient of $M' \otimes_{\CC} N^{\brac{g}}$ \cite{GabInd96}.  Here, $M'$ denotes the \emph{special subspace}, a truncation of $M$ in which any state of the form $L_{-n_1} \cdots L_{-n_k} \ket{v}$, with $\max \set{n_1 , \ldots , n_k} > 1$, is set to $0$.  Finally, the quotient of $M' \otimes_{\CC} N^{\brac{g}}$ which realises the truncated fusion product may be identified by determining those elements of the tensor space, the so-called \emph{spurious states}, that we are forced to set to $0$ as a consequence of setting \svs{} to $0$ when forming $M$ and $N$.

It is always best to illustrate an algorithm with examples.  Let us consider the fusion of the percolation ($c=0$) module $\VirQuot{1,2}$ of \eqref{ses:Q12} with itself, setting the cutoff grade to $0$.  Then, $\VirQuot{1,2}'$ is spanned by $\ket{v}$ (the \hws{} of $\VirQuot{1,2}$) and $L_{-1} \ket{v}$, because $L_{-1}^2 \ket{v} = \frac{2}{3} L_{-2} \ket{v}$ is set to $0$, and $\VirQuot{1,2}^{\brac{0}}$ is spanned by $\ket{v}$.  There are no spurious states to find, so $\brac{\VirQuot{1,2} \fuse \VirQuot{1,2}}^{\brac{0}}$ is two-dimensional.  Applying \eqref{eq:NGK1} with $n=0$, we obtain
\begin{subequations} \label{eq:Q12xQ12g0}
\begin{align}
\func{\Delta}{L_0} \brac{\ket{v} \otimes \ket{v}} &= L_{-1} \ket{v} \otimes \ket{v} + L_0 \ket{v} \otimes \ket{v} + \ket{v} \otimes L_0 \ket{v} = L_{-1} \ket{v} \otimes \ket{v}, \\
\func{\Delta}{L_0} \brac{L_{-1} \ket{v} \otimes \ket{v}} &= L_{-1}^2 \ket{v} \otimes \ket{v} + L_0 L_{-1} \ket{v} \otimes \ket{v} + L_{-1} \ket{v} \otimes L_0 \ket{v} \notag \\
&= L_{-1} \ket{v} \otimes \ket{v} + \tfrac{2}{3} L_{-2} \ket{v} \otimes \ket{v} = L_{-1} \ket{v} \otimes \ket{v} + \tfrac{2}{3} \ket{v} \otimes L_{-1} \ket{v} \notag \\
&= \tfrac{1}{3} L_{-1} \ket{v} \otimes \ket{v}.
\end{align}
\end{subequations}
In the course of this calculation, we have combined $\func{\Delta}{L_{-1}} = \func{\Delta}{L_{-2}} = 0$ with \eqref{eq:NGK1} and \eqref{eq:NGK3} to obtain
\begin{equation}
L_{-2} \ket{v} \otimes \ket{v} = \ket{v} \otimes L_{-1} \ket{v} = -L_{-1} \ket{v} \otimes \ket{v}.
\end{equation}
It follows from \eqref{eq:Q12xQ12g0} that $\func{\Delta}{L_0}$ is diagonalisable with eigenvalues $h_{1,1} = 0$ and $h_{1,3} = \frac{1}{3}$.  From this, we deduce that the fusion product $\VirQuot{1,2} \fuse \VirQuot{1,2}$ decomposes as the direct sum of two \hwms{} whose \hwss{} have conformal dimensions $0$ and $\frac{1}{3}$, respectively.

To identify the \hwms{} appearing in this decomposition unambiguously, we need to compute to higher cutoff grades.  At grade $1$, $\VirQuot{1,2}' \otimes_{\CC} \VirQuot{1,2}^{\brac{1}}$ is four-dimensional, spanned by $\ket{v} \otimes \ket{v}$, $L_{-1} \ket{v} \otimes \ket{v}$, $\ket{v} \otimes L_{-1} \ket{v}$ and $L_{-1} \ket{v} \otimes L_{-1} \ket{v}$, and one uncovers a spurious state as follows:
\begin{align}
0 &= \func{\Delta}{L_{-1}^2} \brac{\ket{v} \otimes \ket{v}} = L_{-1}^2 \ket{v} \otimes \ket{v} + 2 L_{-1} \ket{v} \otimes L_{-1} \ket{v} + \ket{v} \otimes L_{-1}^2 \ket{v} \notag \\
&= \tfrac{2}{3} L_{-2} \ket{v} \otimes \ket{v} + 2 L_{-1} \ket{v} \otimes L_{-1} \ket{v} + \tfrac{2}{3} \ket{v} \otimes L_{-2} \ket{v} \notag \\
&= \tfrac{2}{3} \ket{v} \otimes L_{-1} \ket{v} + 2 L_{-1} \ket{v} \otimes L_{-1} \ket{v} - \tfrac{2}{3} L_{-1} \ket{v} \otimes \ket{v}.
\end{align}
This time, we have used $\func{\Delta}{L_{-1}^2} = \func{\Delta}{L_{-2}} = 0$, \eqref{eq:NGK1} and \eqref{eq:NGK3} to obtain the relations
\begin{equation}
L_{-2} \ket{v} \otimes \ket{v} = \ket{v} \otimes L_{-1} \ket{v} \qquad \text{and} \qquad \ket{v} \otimes L_{-2} v = -L_{-1} \ket{v} \otimes \ket{v}.
\end{equation}
There are no other spurious states, so the truncated fusion product is three-dimensional.  Computing $\func{\Delta}{L_0}$ as before, we find that it is diagonalisable with eigenvalues $0$, $\frac{1}{3}$ and $\frac{4}{3}$.  This refines the grade $0$ conclusion in that we now know that the \hwm{} of conformal dimension $0$ has its \sv{} at grade $1$ set to $0$, a fact which may be confirmed by checking that $\func{\Delta}{L_{-1}}$ annihilates the eigenstate with eigenvalue $0$.  This \hwm{} is therefore either $\VirIrr{0}$ or $\VirQuot{1,1}$.

To decide which, we compute to grade $2$, finding no spurious states in the six-dimensional truncated product $\VirQuot{1,2}' \otimes_{\CC} \VirQuot{1,2}^{\brac{2}}$.  Calculating as before gives $\func{\Delta}{L_0}$ as diagonalisable with eigenvalues $0$, $2$, $\frac{1}{3}$, $\frac{4}{3}$, $\frac{7}{3}$ and $\frac{7}{3}$.  The grade $2$ state may be checked to be obtained by acting with $\func{\Delta}{L_{-2}}$ on the eigenvalue $0$ state, thereby identifying one of the direct summands of the fusion product as $\VirQuot{1,1}$.  Identifying the other summand requires computing to grade $3$.  This time, there is a single spurious state and $\func{\Delta}{L_0}$ is diagonalisable with eigenvalues $0$, $2$, $3$, $\frac{1}{3}$, $\frac{4}{3}$, $\frac{7}{3}$, $\frac{7}{3}$, $\frac{10}{3}$ and $\frac{10}{3}$.  We see that the grade $3$ singular descendant of the eigenvalue $\frac{1}{3}$ state has been set to $0$, so the remaining summand is the irreducible \hwm{} $\VirQuot{1,3} = \VirIrr{1/3}$.

To summarise, we have used the Nahm-Gaberdiel-Kausch algorithm to compute the fusion rule
\begin{equation}
\VirQuot{1,2} \fuse \VirQuot{1,2} = \VirQuot{1,1} \oplus \VirIrr{1/3}.
\end{equation}
The computations beyond grade $1$ quickly become tedious and are best done using an computer (we implemented the algorithm in \textsc{Maple}).  Nevertheless, this example shows that fusion products can be identified from a finite amount of computation (although this would not be true if the result involved modules with infinitely many composition factors, Verma modules for instance).  On the other hand, the Virasoro mode $L_0$ acts diagonalisably on this fusion product, so the result is not particularly interesting so far as \lcft{} is concerned.

A more interesting computation is the fusion of $\VirQuot{1,2}$ with the newly discovered percolation module $\VirIrr{1,3}$.  At grade $0$, $\func{\Delta}{L_0}$ is diagonalisable with eigenvalues $0$ and $1$.  Because these eigenvalues differ by an integer, we cannot conclude that the result decomposes as a direct sum of two \hwms{}.  Our wariness in this matter is justified by the grade $1$ computation in which a new feature is uncovered:  $\func{\Delta}{L_0}$ is found to have eigenvalues $0$, $1$, $1$ and $2$, but is \emph{not diagonalisable} --- the eigenspace of eigenvalue $1$ corresponds to a Jordan block of rank $2$.  This is the sign of logarithmic structure that we have been looking for.

To clarify this structure, note that the eigenvalue $0$ state $\ket{\xi}$ is necessarily a \hws{}.  We can check that $\func{\Delta}{L_{-1}} \ket{\xi}$ is non-zero and is (necessarily) the $\func{\Delta}{L_0}$-eigenstate of the Jordan block.  Its Jordan partner $\ket{\theta}$ is then uniquely determined by $\tbrac{\func{\Delta}{L_0} - \id} \ket{\theta} = \func{\Delta}{L_{-1}} \ket{\xi}$, up to adding multiples of $\func{\Delta}{L_{-1}} \ket{\xi}$.  Finally, the eigenvalue $2$ state is realised by $\func{\Delta}{L_{-1}} \ket{\theta}$.  All this amounts to defining (and normalising) the states appearing at grade $1$.  What remains to be determined is the action of $L_1$:%
\footnote{There is a subtlety to this computation worth mentioning.  The action of $\func{\Delta}{L_n}$, $n>0$, at grade $g$ should be understood to map into the grade $g-n$ fusion space.  However, the latter is always a subspace (quotient) of the former.  We may therefore compute $\func{\Delta}{L_1} \ket{\theta}$ in the grade $1$ fusion product and project onto the grade $0$ subspace by setting to zero all terms with $\func{\Delta}{L_0}$-eigenvalue $1$.}%
\begin{equation} \label{eq:beta14}
\func{\Delta}{L_1} \ket{\theta} = -\tfrac{1}{2} \ket{\xi}.
\end{equation}
Because $\func{\Delta}{L_{-1}} \ket{\xi}$ is a \sv{}, this equation holds for \emph{any} choice of Jordan partner state $\ket{\theta}$.

This grade $1$ fusion calculation shows that the product $\VirQuot{1,2} \fuse \VirIrr{1,3}$ is an indecomposable non-\hwm{} which we shall denote by $\VirStag{1,4}$.  The \hws{} $\ket{\xi}$ of dimension $0$ generates a \hwsm{} of $\VirStag{1,4}$ whose \sv{} of dimension $1$, $L_{-1} \ket{\xi}$, is non-vanishing.  Using the fusion algorithm at grade $2$, we find that the singular dimension $2$ descendant of $\ket{\xi}$ vanishes, thereby identifying this \hwsm{} as $\VirQuot{1,2}$.  In the quotient module $\VirStag{1,4} / \VirQuot{1,2}$, the equivalence class $\ket{\theta} + \VirQuot{1,2}$ is a \hws{} of dimension $1$.  Checking its singular descendants of dimensions $5$ and $7$ therefore requires fusing to grade $6$ and examining the $\VirQuot{1,2}$-quotient.%
\footnote{This requires finding two spurious states in a $46$-dimensional vector space.}  %
The results --- the first singular descendant is found to vanish whereas the second does not --- indicate that the corresponding \hwm{} is isomorphic to $\VirQuot{1,4} = \VirVer{1} / \VirVer{5}$.  This then establishes the exactness of the sequence
\begin{equation}
\dses{\VirQuot{1,2}}{}{\VirStag{1,4}}{}{\VirQuot{1,4}}.
\end{equation}
The Loewy diagram for the indecomposable $\VirStag{1,4}$ is given in \figref{fig:PercLoewy} (left).  The bottom composition factor $\VirIrr{1}$ (the socle) is generated by $\func{\Delta}{L_{-1}} \ket{\xi}$.  By taking appropriate quotients, the $\VirIrr{0}$ and the top $\VirIrr{1}$ may be similarly associated with (equivalence classes of) $\ket{\xi}$ and $\ket{\theta}$, respectively.  The $\VirIrr{7}$ corresponds to the non-vanishing singular descendant of $\ket{\theta} + \VirQuot{1,2}$.

\begin{figure}
\begin{center}
\begin{tikzpicture}[thick,>=latex,
	nom/.style={circle,draw=black!20,fill=black!20,inner sep=1pt}
	]
\node (top0) at (0,1.5) [] {$\VirIrr{1}$};
\node (left0) at (-1.5,0) [] {$\VirIrr{0}$};
\node (right0) at (1.5,0) [] {$\VirIrr{7}$};
\node (bot0) at (0,-1.5) [] {$\VirIrr{1}$};
\node (top1) at (6,1.5) [] {$\VirIrr{2}$};
\node (left1) at (4.5,0) [] {$\VirIrr{0}$};
\node (right1) at (7.5,0) [] {$\VirIrr{5}$};
\node (bot1) at (6,-1.5) [] {$\VirIrr{2}$};
\node at (0,0) [nom] {$\VirStag{1,4}$};
\node at (6,0) [nom] {$\VirStag{1,5}$};
\draw [->] (top0) -- (left0);
\draw [->] (top0) -- (right0);
\draw [->] (left0) -- (bot0);
\draw [->] (right0) -- (bot0);
\draw [->] (top1) -- (left1);
\draw [->] (top1) -- (right1);
\draw [->] (left1) -- (bot1);
\draw [->] (right1) -- (bot1);
\end{tikzpicture}
\caption{\label{fig:PercLoewy}Loewy diagrams illustrating the socle series (see \appref{app:Socle}) for the indecomposable Virasoro modules $\VirStag{1,4}$ and $\VirStag{1,5}$ constructed using the Nahm-Gaberdiel-Kausch fusion algorithm.} 
\end{center}
\end{figure}
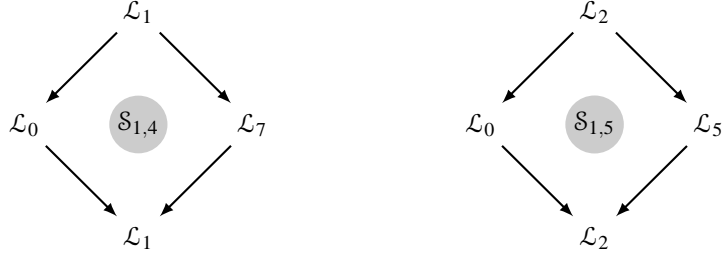

This demonstrates that the percolation \cft{} necessarily contains indecomposable modules (in some boundary sectors) on which the Virasoro zero-mode acts non-diagonalisably.  As we saw in \secref{sec:LogCorr}, this leads to logarithmic singularities in correlation functions.  Before discussing this in more detail, let us pause to explore further what fusion can tell us about the spectrum of modules.  The Nahm-Gaberdiel-Kausch algorithm may be applied to the fusion of $\VirIrr{1/3}$ with itself and computing to grade $5$ establishes that the result is the direct sum of $\VirIrr{1/3}$ and a new indecomposable $\VirStag{1,5}$ whose structure is described by the exact sequence
\begin{equation}
\dses{\VirQuot{1,1}}{}{\VirStag{1,5}}{}{\VirQuot{1,5}}.
\end{equation}
Its Loewy diagram is illustrated in \figref{fig:PercLoewy} (right).  The \hwsm{} is the (indecomposable) vacuum module containing the vacuum $\ket{0}$ and $\ket{T} = L_{-2} \ket{0}$.  The latter state (corresponding to the energy-momentum tensor) has a Jordan partner, unique up to adding multiples of $\ket{T}$, which we will denote by $\ket{t}$.  If we normalise this partner by $\tbrac{\func{\Delta}{L_0} - 2 \id} \ket{t} = \ket{T}$, then explicit computation gives
\begin{equation} \label{eq:beta15}
\func{\Delta}{L_2} \ket{t} = -\tfrac{5}{8} \ket{0}.
\end{equation}
Again, this equation is independent of the choice of $\ket{t}$.

It is possible to identify the result of many more fusion rules including \cite{EbeVir06,RidPer07}
\begin{equation}
\begin{aligned}
\VirQuot{1,2} \fuse \VirStag{1,4} &= 2 \: \VirIrr{1/3} \oplus \VirStag{1,5}, \\
\VirQuot{1,2} \fuse \VirStag{1,5} &= \VirStag{1,4} \oplus \VirIrr{10/3}, \\
\VirIrr{1/3} \fuse \VirStag{1,4} &= 2 \: \VirStag{1,4} \oplus \VirIrr{10/3}, \\
\VirIrr{1/3} \fuse \VirStag{1,5} &= 2 \: \VirIrr{1/3} \oplus \VirStag{1,7},
\end{aligned}
\qquad
\begin{aligned}
\VirStag{1,4} \fuse \VirStag{1,4} &= 4 \: \VirIrr{1/3} \oplus 2 \: \VirStag{1,5} \oplus \VirStag{1,7}, \\
\VirStag{1,4} \fuse \VirStag{1,5} &= 2 \: \VirStag{1,4} \oplus 2 \: \VirIrr{10/3} \oplus \VirStag{1,8}, \\
\VirStag{1,5} \fuse \VirStag{1,5} &= \VirIrr{1/3} \oplus 2 \: \VirStag{1,5} \oplus \VirStag{1,7} \oplus \VirIrr{28/3}.
\end{aligned}
\end{equation}
Here, the modules $\VirStag{1,7}$ and $\VirStag{1,8}$ are new indecomposables with exact sequences
\begin{equation} \label{ses:S17S18}
\dses{\VirQuot{1,5}}{}{\VirStag{1,7}}{}{\VirQuot{1,7}}, \qquad \dses{\VirQuot{1,4}}{}{\VirStag{1,8}}{}{\VirQuot{1,8}}.
\end{equation}
We remark that fusing $\VirStag{1,4}$ or $\VirStag{1,5}$ with another module requires knowing the explicit form of the (generalised) \svs{} which have been set to $0$.  Fusion computations with the new modules generated here have met with only partial success, chiefly because the computational intensity of the algorithm increases very quickly as the grade required to completely identify the fusion product grows.  Nevertheless, all such computations are consistent with the following conjecture for the fusion rules, presented algorithmically for simplicity:%
\footnote{One can convert this into a general formula, see \cite{RasFus07} for example.  However, the result does not seem particularly illuminating to us.}%
\begin{enumerate}
\item The spectrum includes irreducibles $\VirQuot{1,3k} = \VirIrr{\brac{3k-1} \brac{3k-2} / 6}$ and indecomposables $\VirStag{1,3k-1}$ and $\VirStag{1,3k-2}$, for $k \in \ZZ_+$ (we let $\VirStag{1,1} = \VirQuot{1,1}$ and $\VirStag{1,2} = \VirQuot{1,2}$).  To fuse any of these modules, first break any indecomposables into their constituent \hwms{} ($\VirQuot{1,-2} = \VirQuot{1,-1} \equiv \set{0}$):
\begin{equation} \label{eq:NotQuiteGroth}
\VirStag{1,3k-1} \lra \VirQuot{1,3k-1} \oplus \VirQuot{1,3k-5}, \quad
\VirStag{1,3k-2} \lra \VirQuot{1,3k-2} \oplus \VirQuot{1,3k-4}.
\end{equation}
\item Compute the ``fusion'' using distributivity and
\begin{equation}
\VirQuot{1,s} \fakefuse \VirQuot{1,s'} = \VirQuot{1,\abs{s-s'}+1} \oplus \VirQuot{1,\abs{s-s'}+3} \oplus \cdots \oplus \VirQuot{1,s+s'-3} \oplus \VirQuot{1,s+s'-1}.
\end{equation}
(We have enclosed the fusion operation in quotes to emphasise that this is not a true fusion rule).
\item In the result, reverse \eqref{eq:NotQuiteGroth} by replacing the combinations $\VirQuot{1,3k-1} \oplus \VirQuot{1,3k-5}$ and $\VirQuot{1,3k-2} \oplus \VirQuot{1,3k-4}$ by $\VirStag{1,3k-1}$ and $\VirStag{1,3k-2}$, respectively.  There is always a unique way of doing this.
\end{enumerate}
For example, if we wished to fuse $\VirStag{1,5}$ with $\VirIrr{10/3} = \VirQuot{1,6}$, we would instead compute that
\begin{equation}
\brac{\VirQuot{1,1} \oplus \VirQuot{1,5}} \fakefuse \VirQuot{1,6} = \VirQuot{1,6} \oplus \brac{\VirQuot{1,2} \oplus \VirQuot{1,4} \oplus \VirQuot{1,6} \oplus \VirQuot{1,8} \oplus \VirQuot{1,10}}
\end{equation}
from which we read off that
\begin{equation}
\VirStag{1,5} \fuse \VirIrr{10/3} = \VirStag{1,4} \oplus 2 \: \VirIrr{10/3} \oplus \VirStag{1,10}.
\end{equation}

\subsection{Logarithmic Correlators Again} \label{sec:PercLogCorr}

Consider now the structure of the indecomposable module $\VirStag{1,5}$.  It has a submodule generated by the vacuum $\ket{0}$, while $\VirStag{1,5}$ is itself generated by the state $\ket{t}$ satisfying
\begin{equation}
L_0 \ket{t} = 2 \ket{t} + \ket{T}, \qquad L_1 \ket{t} = 0, \qquad L_2 \ket{t} = -\tfrac{5}{8} \ket{0}, \qquad L_n \ket{t} = 0 \quad \text{for \( n>2 \).}
\end{equation}
We recall that $\ket{T} = L_{-2} \ket{0}$.  The \ope{} of the corresponding fields $\func{T}{z}$ and $\func{t}{w}$ is therefore slightly different to that considered in \secref{sec:LogCorr}:
\begin{equation}
\func{T}{z} \func{t}{w} \sim -\frac{5}{8} \frac{1}{\brac{z-w}^4} + \frac{2 \func{t}{w} + \func{T}{w}}{\brac{z-w}^2} + \frac{\func{\pd t}{w}}{z-w}.
\end{equation}
Normalising so that $\braket{0}{0} = 1$, we note that $\corrfn{\func{T}{z} \func{T}{w}} = 0$ because $\ket{T}$ is singular.  The global invariance of the vacuum then leads to the usual three \pdes{} for $\corrfn{\func{T}{z} \func{t}{w}}$ whose solution is
\begin{equation} \label{2pt:Tt}
\corrfn{\func{T}{z} \func{t}{w}} = \frac{B}{\brac{z-w}^4}, \qquad B = \bracket{0}{L_2}{t} = -\tfrac{5}{8} \braket{0}{0} = -\tfrac{5}{8}.
\end{equation}
As $\func{T}{z}$ and $\func{t}{w}$ can be shown to be mutually bosonic \cite[App.~B]{RidLog07}, we also obtain
\begin{equation} \label{2pt:tt}
\corrfn{\func{t}{z} \func{t}{w}} = \frac{A + \tfrac{5}{4} \func{\log}{z-w}}{\brac{z-w}^4},
\end{equation}
confirming the existence of logarithmic singularities in percolation correlators.  We emphasise that, unlike $B$ in \eqref{2pt:Tt}, the value of the constant $A$ depends upon the precise choice we make for $\ket{t}$.

As a second example, we consider the other module with non-diagonalisable $L_0$-action that we have studied:  $\VirStag{1,4}$.  This module is generated by a state $\ket{\theta}$ satisfying
\begin{equation}
L_0 \ket{\theta} = \ket{\theta} + L_{-1} \ket{\xi}, \qquad L_1 \ket{\theta} = -\tfrac{1}{2} \ket{\xi}, \qquad L_n \ket{\theta} = 0 \quad \text{for \(n>1\).}
\end{equation}
Here, $\ket{\xi}$ is a dimension $0$ \hws{} generating a submodule isomorphic to $\VirQuot{1,2}$.  The field $\func{\theta}{w}$ corresponding to $\ket{\theta}$ therefore has \ope{}
\begin{equation}
\func{T}{z} \func{\theta}{w} \sim -\frac{1}{2} \frac{\func{\xi}{w}}{\brac{z-w}^3} + \frac{\func{\theta}{w} + \func{\pd \xi}{w}}{\brac{z-w}^2} + \frac{\func{\pd \theta}{w}}{z-w}.
\end{equation}
Again, we take $\braket{\xi}{\xi} = 1$ and arrive at
\begin{equation}
\corrfn{\func{\pd x}{z} \func{\theta}{w}} = \frac{B}{\brac{z-w}^2}.
\end{equation}

The determination of $B$ is, however, subtle \cite{RidLog07}.  Na\"{\i}vely, we might expect that $\ket{\pd \xi} = L_{-1} \ket{\xi}$ implies that $B = \braket{\pd \xi}{\theta} = \bracket{\xi}{L_1}{\theta} = -\frac{1}{2} \braket{\xi}{\xi} = -\frac{1}{2}$, but this turns out to be incorrect.  To see why, recall that the standard definition of the outgoing state corresponding to a \emph{primary} field $\func{\phi}{z} = \sum_n \phi_n z^{-n-h}$ is
\begin{equation}
\bra{\phi} = \lim_{z \to \infty} z^{2h} \bra{0} \func{\phi}{z} \qquad \iff \qquad \phi_n^{\dag} = \phi_{-n}.
\end{equation}
We certainly want this definition to apply to $\func{\xi}{z}$, a dimension $0$ primary field.  But then,
\begin{equation}
B = \bracket{0}{\brac{\pd \xi}_1}{\theta} = -\bracket{0}{\xi_1}{\theta} = -\bracket{0}{\xi_0 L_1}{\theta} = -\bracket{\xi}{L_1}{\theta} = \tfrac{1}{2} \braket{\xi}{\xi} = \tfrac{1}{2}.
\end{equation}
This is the correct conclusion (see \secref{sec:StagLogCorr} for a more general discussion).  In any case, once $B$ is correctly determined, the computation of $\corrfn{\func{\theta}{z} \func{\theta}{w}}$ proceeds as before and one obtains
\begin{equation}
\corrfn{\func{\theta}{z} \func{\theta}{w}} = \frac{A - \func{\log}{z-w}}{\brac{z-w}^2}.
\end{equation}
Once again, $A$ depends upon the precise choice of $\ket{\theta}$, whereas $B$ does not.

\subsection{Further Developments} \label{sec:PercFuture}

We have seen that the boundary \cft{} describing critical percolation is logarithmic and that the spectrum includes the modules $\VirQuot{1,1}$, $\VirQuot{1,2}$, $\VirIrr{\brac{3k-1} \brac{3k-2} / 6}$, $\VirStag{3k+1}$ and $\VirStag{3k+2}$, for $k \in \ZZ_+$.  A natural question to ask now is whether there is more to the spectrum.  One way to look for additional modules is to consider other measurable quantities in percolation.  The most famous generalisation of Cardy's crossing probability is that which asks for the probability that a random configuration of edges (with $p = p_c$) contains a connected cluster of open edges connecting all four sides of the rectangular subdomain.  In \cite{WatCro96}, Watts notes that the four-point functions that solve the second order differential equations that lead to Cardy's formula \eqref{eq:Cardy} do not satisfy the properties one expects for this more general crossing probability.  However, a field of dimension $0$ has, at $c=0$, a singular descendant of grade $5$ and the corresponding fifth order differential equation not only has a unique solution satisfying Watts' properties, but it also beautifully interpolates the numerical data known \cite{LanUni92} for this crossing probability.  Watts' proposed solution has since been rigorously proven by Dub\'{e}dat \cite{DubExc06}.

Given what we have learned in \secref{sec:PercIndec}, the natural interpretation to propose \cite{RidPer08} is that the field appearing in Watts' four-point function corresponds to the \hws{} of the module $\VirVer{0} / \VirVer{5}$.%
\footnote{The discussion makes it clear that the singular vector at grade $5$ must be set to zero, but it is not \emph{a priori} clear why its grade $7$ partner should not be set to zero.  It is straight-forward, but tedious, to check that the seventh order differential equation that would result from setting this partner to zero does not admit Watts' crossing formula as a solution.}  %
It is rather interesting to note that this quotient module does not have the form $\VirQuot{r,s}$ for any positive integers $r$ and $s$.  Instead, one may identify it using \emph{fractional} Kac labels:  $\VirVer{0} / \VirVer{5} = \VirQuot{2,5/2} = \VirQuot{5/3,3}$.  Perhaps surprisingly, denoting this module by $\VirQuot{2,5/2}$ is convenient for discussing the modules one subsequently generates by fusing with $\VirQuot{1,2}$.  For example, one finds \cite{RidPer08} that
\begin{equation}
\VirQuot{1,2} \fuse \VirQuot{2,5/2} = \VirQuot{2,3/2} \oplus \VirQuot{2,7/2}; \qquad \VirQuot{2,3/2} = \VirVer{1/3} / \VirVer{10/3} = \VirIrr{1/3}, \quad \VirQuot{2,7/2} = \VirVer{0} / \VirVer{7}.
\end{equation}
Unfortunately, fusing $\VirQuot{2,5/2}$ with itself leads to indecomposable modules which have significantly more complicated structures and are rather poorly characterised (see \cite{RidPer08} for further details).  We remark that more general percolation crossing probabilities are considered in \cite{SimPer07} from a different perspective.

From a more abstract point of view, we have seen that Cardy's crossing probability leads to indecomposable modules which may be associated with the first row of the (extended) Kac table (\tabref{tab:Kac}), so one is led to ask whether there is a complementary observable quantity that can be associated to the first column.  In percolation, this is not so clear.  However, the statistical model known as \emph{dilute polymers} (or the \emph{self-avoiding walk}) also has a continuum limit that is (believed to be) described by a $c=0$ \cft{}.  An old proposal of Gurarie and Ludwig \cite{GurCon02} associates this latter \cft{} with modules from the first column of the extended Kac table.%
\footnote{Actually, the proposal of \cite{GurCon02} was that percolation should be associated to the first column and dilute polymers to the first row, though this statement is not repeated in the sequel \cite{GurCon04}.  This was corrected in \cite{RidPer07} for the boundary theory relevant here.}  %
We will not detail this polymer theory or its observables here, instead mentioning only that a field corresponding to the module $\VirQuot{2,1} = \VirIrr{5/8}$ is relevant and that fusing this module with itself leads to an indecomposable module which we denote by $\VirStag{3,1}$:
\begin{equation}
\VirIrr{5/8} \fuse \VirIrr{5/8} = \VirStag{3,1}, \qquad \dses{\VirQuot{1,1}}{}{\VirStag{3,1}}{}{\VirQuot{3,1}}.
\end{equation}
The Loewy diagram of $\VirStag{3,1}$ is identical to that of $\VirStag{1,5}$ (illustrated in \figref{fig:PercLoewy}) except that the composition factor $\VirIrr{5}$ is replaced by $\VirIrr{7}$.  If we regard the submodule $\VirQuot{1,1}$ as being generated by the vacuum, then $\ket{T}$ has a Jordan partner $\ket{t'} \in \VirStag{3,1}$ which can be distinguished from $\ket{t} \in \VirStag{1,5}$ by
\begin{equation} \label{eq:AnomalyNumbers}
L_2 \ket{t} = -\tfrac{5}{8} \ket{0}, \qquad L_2 \ket{t'} = \tfrac{5}{6} \ket{0}.
\end{equation}
We remark that these coefficients $b_{1,5} = -\frac{5}{8}$ and $b_{3,1} = \frac{5}{6}$, called \emph{anomaly numbers} in \cite{GurCTh99}, have recently been measured directly in the respective lattice theories (through numerical simulation) \cite{DubCon10}.  This confirms experimentally that percolation corresponds to first row modules and dilute polymers to first column modules, at least in their formulation as boundary \cfts{}.

One thing worth mentioning here is the observation (see \cite[App.~A]{GurCon04}) that the otherwise reasonable-looking two-point function $\corrfn{\func{t}{z} \func{t'}{w}}$ is inconsistent with conformal invariance.  More precisely, the three inhomogeneous \pdes{} for this correlator, which are derived from the global conformal invariance of the vacuum, admit no simultaneous solution.  This appears \cite{GurCon04,RidPer07} to rule out the possibility that both $\VirStag{1,5}$ and $\VirStag{3,1}$ can belong to the spectrum.  However, a more careful conclusion \cite{RidPer08} is that the presence of one of these indecomposables in a boundary sector labelled by boundary conditions $B_1$ and $B_2$ precludes the presence of the other in any boundary sector with label $B_1$ or $B_2$.  This does not prove that $\VirStag{1,5}$ and $\VirStag{3,1}$ can coexist in a boundary \cft{}, but it does provide a loophole whereby inconsistent two-point functions may be avoided.  Such a loophole appears to be at work in the results of \cite{RasFus07} in which boundary conditions corresponding to all the extended Kac labels $\brac{r,s}$ are constructed for a loop model variant of critical percolation.%
\footnote{Interestingly, the so-called \emph{Kac modules} $\mathcal{K}_{r,s}$ which appear here generalise the $\VirQuot{r,1}$ and $\VirQuot{1,s}$ as quotients of Feigin-Fuchs modules, rather than quotients of Verma modules, see \cite{RasCla11,RasCla13}.}  %
An extremely important open problem, in our opinion, is to determine if the conformal invariance of the vacuum leads to further, more stringent, constraints on the boundary (and bulk) spectra of \lcfts{}.

\section{Symplectic Fermions and the Triplet Model} \label{sec:Trip}

The triplet theories $\TripAlg{q,p}$, with $p,q \in \ZZ_+$, $p>q$ and $\gcd \set{p,q} = 1$, form a family of logarithmic extensions of the minimal Virasoro models. When $q=1$, the minimal model is empty, but the logarithmic theory is non-trivial (these are the original triplet models of \cite{Kausch:1991}). We will concentrate on the simplest of these models, that with $q=1$ and $p=2$,%
\footnote{The model with $p=q=1$ is just $\AKMA{sl}{2}$ at level $1$.} %
which has a free field realisation known as symplectic fermions. We start with this free theory before turning to the triplet algebra $\TripAlg{1,2}$ and then to its subalgebra, the singlet algebra $\SingAlg{1,2}$. The theories associated to these algebras are extremely closely related as we illustrate in Figure \ref{fig:rels}.  We then detail the modular transformations of the singlet characters and compute Grothendieck fusion rules for $\SingAlg{1,2}$ using a continuum Verlinde formula, before lifting the results to the triplet model and symplectic fermions.

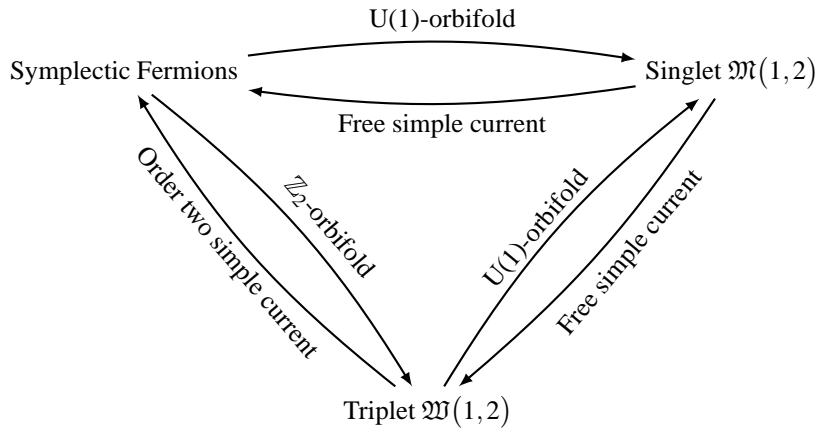
\begin{figure}
\begin{center}
\begin{tikzpicture}[thick,>=latex,
	nom/.style={circle,draw=black!20,fill=black!20,inner sep=1pt}
	]
\node (left1) at (1,0) [] {Symplectic Fermions};
\node (right1) at (9,0) [] {Singlet $\SingAlg{1,2}$};
\node (bot1) at (5,-4.5) [] {Triplet $\TripAlg{1,2}$};
\draw [->][bend left=8] (left1) to node [above] {U(1)-orbifold} (right1);
\draw [->][bend left=8] (right1) to node [below] {Free simple current} (left1);
\draw [->][bend left=10] (left1) to node [right,xshift=-10pt] {\rotatebox{310}{$\ZZ_2$-orbifold}} (bot1);
\draw [->][bend left=10] (bot1) to node [left,xshift=30pt] {\rotatebox{310}{Order two simple current}} (left1);
\draw [->][bend left=10] (bot1) to node [left,xshift=17pt] {\rotatebox{50}{U(1)-orbifold}} (right1);
\draw [->][bend left=8] (right1) to node [right,xshift=-20pt] {\rotatebox{50}{Free simple current}} (bot1);
\end{tikzpicture}
\caption{\label{fig:rels}The symplectic fermion algebra, the triplet algebra $\TripAlg{1,2}$ and the singlet algebra $\SingAlg{1,2}$ are all related by simple current extensions and orbifolds.}
\end{center}
\end{figure}

\subsection{Symplectic Fermions} \label{sec:SF}

Symplectic fermions were first introduced by Kausch \cite{KauCur95} in order to study the fermionic $bc$ ghost system of central charge $c=-2$. They also describe the \WZW{} model on the abelian supergroup $\SLSG{PSL}{1}{1}$ and should be regarded as the simplest fermionic analogue of the free boson. The action involves two non-chiral fermionic fields $\tfunc{\nonch{\theta}^\pm}{z,\ahol{z}}$:
\begin{equation} \label{eq:SFAction}
S \sqbrac{\func{\nonch{\theta}^\pm}{z,\ahol{z}}} = \frac{1}{4\pi} \int \sqbrac{\func{\pd \nonch{\theta}^+}{z,\ahol{z}} \func{\apd \nonch{\theta}^-}{z,\ahol{z}} - \func{\pd \nonch{\theta}^-}{z,\ahol{z}} \func{\apd \nonch{\theta}^+}{z,\ahol{z}}} \: \dd z \dd \ahol{z}.
\end{equation}
This action is invariant under shifts by holomorphic and anti-holomorphic fields,
and this implies, as with the free boson, the \ope{}
\begin{equation}
\func{\nonch{\theta}^+}{z, \ahol{z}} \func{\nonch{\theta}^-}{w, \ahol{w}} = A + \log \abs{z-w}^2 + \cdots,
\end{equation}
where $A$ is a constant of integration.

The equations of motion state that $\tfunc{J^\pm}{z}=\tfunc{\pd \nonch{\theta}^\pm}{z,\ahol{z}}$ and $\tfunc{\ahol{J}^\pm}{\ahol{z}}=\tfunc{\apd \nonch{\theta}^\pm}{z,\ahol{z}}$ are holomorphic and antiholomorphic, respectively, and we will take the former to generate the chiral algebra. These are the symplectic fermion currents and their \opes{} are
\begin{equation}
\func{J^+}{z}\func{J^-}{w}\sim\frac{1}{(z-w)^2}, \qquad \func{J^\pm}{z}\func{J^\pm}{w} \sim 0.
\end{equation}
Their modes then satisfy the anticommutation relations of the affine Lie superalgebra $\AKMSA{psl}{1}{1}$
(as with the case of the free boson, any non-zero level may be rescaled to $1$):
\begin{equation}
\acomm{J^+_m}{J^-_n} = m\delta_{m+n=0}, \qquad \acomm{J^\pm_m}{J^\pm_n} = 0.
\end{equation}
The Virasoro field is $\tfunc{T}{z}=\normord{\tfunc{J^-}{z}\tfunc{J^+}{z}}$ and its central charge is $c=-2$.


Let us turn to representations.  As usual, we start with \hwms{} and these can be quickly analysed by expanding the double integral
\[
\oint_0 \oint_w \func{J^+}{z} \func{J^-}{w} z^{m+1} w^n \brac{z-w}^{-1} \: \frac{\dd z}{2 \pi \ii} \frac{\dd w}{2 \pi \ii}
\]
in the usual fashion, so as to obtain the generalised commutation relation
\begin{equation} \label{eq:SFGCR}
\sum_{j=0}^{\infty} \sqbrac{J_{m-j}^+ J_{n+j}^- - J_{n-1-j}^- J_{m+1+j}^+} = \frac{1}{2} m \brac{m+1} \delta_{m+n=0} - L_{m+n}.
\end{equation}
The integer moding of the Virasoro algebra requires that $m+n \in \ZZ$, so that $m$ and $n$ must be either both integers or both half-integers.

We apply \eqref{eq:SFGCR}, with $m=n=0$, to a state $\ket{\phi}$ which is annihilated by positive modes, obtaining
\begin{equation}
L_0 \ket{\phi} = -J^+_0 J^-_0 \ket{\phi} \qquad \Rightarrow \qquad L_0^2 \ket{\phi} = 0.
\end{equation}It follows that $\ket{\phi}$ belongs to a Jordan block for $L_0$ with eigenvalue $0$ and rank at most $2$.  Repeating this, with $m = -\tfrac{1}{2}$ and $n = \tfrac{1}{2}$, gives $L_0 \ket{\phi} = -\tfrac{1}{8} \ket{\phi}$, hence the only half-integer moded \hws{} has conformal dimension $-\tfrac{1}{8}$.  It follows that we have only one module%
\footnote{Technically, there are two (graded) modules according as to the parity of its generating state.  We shall usually ignore this distinction.} %
in the half-integer moded sector, necessarily irreducible, that we shall denote by $\SFIrr{1/2}$, the label corresponding to the moding.  In the integer-moded sector, we have an irreducible $\SFIrr{0}$ (the vacuum module) as well as an indecomposable $\SFStag{0}$ generated by a dimension $0$ generalised eigenvector $\ket{\Omega}$ of $L_0$.%
\footnote{The existence of this indecomposable module follows from applying the induced module construction to the \uea{} of $\SLSA{psl}{1}{1}$, considered as a four-dimensional $\SLSA{psl}{1}{1}$-module.  Other indecomposables with integer moding may similarly be constructed \cite{RunBra13}.}  %
Its Loewy diagram is given in \figref{fig:SFTripLoewy}, where we remark that the four composition factors may be associated to the states $\ket{\Omega}$, $J^+_0 \ket{\Omega}$, $J^-_0 \ket{\Omega}$ and $J^-_0 J^+_0 \ket{\Omega} = \ket{0}$.  Finally, we note that if $J^-_0$ is regarded as a creation operator and $J^+_0$ as an annihilation operator, then $J^+_0 \ket{\Omega}$ generates the vacuum Verma module $\SFVer{0}$ (the Verma module for half-integer moding is already irreducible:  $\SFVer{1/2} = \SFIrr{1/2}$) and the indecomposable $\SFStag{0}$ is characterised by the exact sequence
\begin{equation} \label{ses:SFStag}
\dses{\SFVer{0}}{}{\SFStag{0}}{}{\SFVer{0}}.
\end{equation}

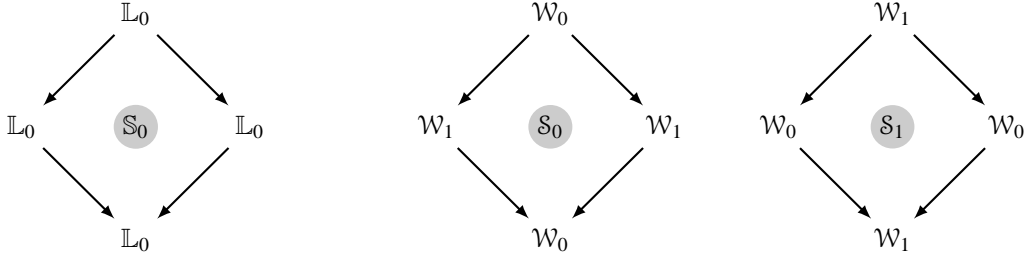
\begin{figure}
\begin{center}
\begin{tikzpicture}[thick,>=latex,
	nom/.style={circle,draw=black!20,fill=black!20,inner sep=1pt}
	]
\node (top1) at (5,1.5) [] {$\SFIrr{0}$};
\node (left1) at (3.5,0) [] {$\SFIrr{0}$};
\node (right1) at (6.5,0) [] {$\SFIrr{0}$};
\node (bot1) at (5,-1.5) [] {$\SFIrr{0}$};
\node at (5,0) [nom] {$\SFStag{0}$};
\draw [->] (top1) -- (left1);
\draw [->] (top1) -- (right1);
\draw [->] (left1) -- (bot1);
\draw [->] (right1) -- (bot1);
\end{tikzpicture}
\hspace{0.1\textwidth}
\begin{tikzpicture}[thick,>=latex,
	nom/.style={circle,draw=black!20,fill=black!20,inner sep=1pt}
	]
\node (top0) at (0,1.5) [] {$\TripIrr{0}$};
\node (left0) at (-1.5,0) [] {$\TripIrr{1}$};
\node (right0) at (1.5,0) [] {$\TripIrr{1}$};
\node (bot0) at (0,-1.5) [] {$\TripIrr{0}$};
\node (top1) at (4.5,1.5) [] {$\TripIrr{1}$};
\node (left1) at (3,0) [] {$\TripIrr{0}$};
\node (right1) at (6,0) [] {$\TripIrr{0}$};
\node (bot1) at (4.5,-1.5) [] {$\TripIrr{1}$};
\node at (0,0) [nom] {$\TripStag{0}$};
\node at (4.5,0) [nom] {$\TripStag{1}$};
\draw [->] (top0) -- (left0);
\draw [->] (top0) -- (right0);
\draw [->] (left0) -- (bot0);
\draw [->] (right0) -- (bot0);
\draw [->] (top1) -- (left1);
\draw [->] (top1) -- (right1);
\draw [->] (left1) -- (bot1);
\draw [->] (right1) -- (bot1);
\end{tikzpicture}
\caption{\label{fig:SFTripLoewy} The Loewy diagrams for the socle series of the indecomposable symplectic fermion module $\SFStag{0}$ (left) and the indecomposable $\TripAlg{1,2}$-modules $\TripStag{0}$ and $\TripStag{1}$ (right).  In each case, the non-diagonalisable action of $L_0$ maps the top factor (the head) onto the bottom factor (the socle).}
\end{center}
\end{figure}

The fusion ring generated by the irreducibles is particularly easy to work out using (a variant \cite{GabFus97} of) the Nahm-Gaberdiel-Kausch algorithm.  The vacuum module $\SFIrr{0}$ is the fusion identity and one finds that
\begin{equation} \label{FR:SF}
\SFIrr{1/2} \fuse \SFIrr{1/2} = \SFStag{0}, \qquad 
\SFIrr{1/2} \fuse \SFStag{0} = 4 \: \SFIrr{1/2}, \qquad 
\SFStag{0} \fuse \SFStag{0} = 4 \: \SFStag{0}.
\end{equation}
The characters $\fch{\SFIrr{\lambda}}{q} = \traceover{\SFIrr{\lambda}} q^{L_0 - c/24}$ of the irreducibles are likewise easily obtained:
\begin{subequations} \label{ch:SF}
\begin{equation}
\ch{\SFIrr{0}} = q^{1/12} \prod_{n=1}^{\infty} \brac{1+q^n}^2 = \frac{\Jth{2}{1;q}}{2 \func{\eta}{q}}, \qquad 
\ch{\SFIrr{1/2}} = q^{-1/24} \prod_{n=1}^{\infty} \brac{1+q^{n-1/2}}^2 = \frac{\Jth{3}{1;q}}{\func{\eta}{q}}.
\end{equation}
Note that the factor of $2$ for $\SFIrr{0}$ would disappear if we instead considered its Verma cover $\SFVer{0}$.  As symplectic fermions are described by an affine Lie superalgebra, it is natural to also consider the supercharacters in which fermionic states are counted with negative multiplicity (we assume a bosonic ground state):
\begin{equation}
\sch{\SFIrr{0}} = q^{1/12} \prod_{n=1}^{\infty} \brac{1-q^n}^2 = \func{\eta}{q}^2, \qquad 
\sch{\SFIrr{1/2}} = q^{-1/24} \prod_{n=1}^{\infty} \brac{1-q^{n-1/2}}^2 = \frac{\Jth{4}{1;q}}{\func{\eta}{q}}.
\end{equation}
\end{subequations}
Excluding $\sch{\SFIrr{0}}$, whose S-transformation involves factors of $\tau$ (we write $q = \ee^{2 \pi \ii \tau}$ as usual), these characters and supercharacters have good modular properties.  However, the S-matrix one obtains has no row or column with all entries non-zero, hence the Verlinde formula is inapplicable.


Finally, we consider general modings for the currents $J^+$ and $J^-$.  Representations on which the algebra acts with modings different to that of the vacuum module, for example $\SFIrr{1/2}$, are said to be \emph{twisted}.%
\footnote{Strictly speaking, these are only modules for an orbifold of the chiral algebra, but there is usually little harm in neglecting this.}  %
For symplectic fermions, the generic twisted module $\SFVer{\lambda} = \SFIrr{\lambda}$ is the irreducible generated by a highest weight state $\ket{\mu_{\lambda}}$ upon which the symplectic fermion currents act with mode decomposition
\begin{equation}
\func{J^\pm}{z} = \sum_{n \in \ZZ \mp \lambda} J^\pm_n z^{-n-1}.
\end{equation}
Taking $0 < \lambda < 1$ and applying \eqref{eq:SFGCR}, the conformal dimension of $\ket{\mu_{\lambda}}$ is found to be $\Delta_{\lambda} = -\tfrac{1}{2} \lambda \brac{1 - \lambda}$.  The corresponding primary field $\func{\mu_{\lambda}}{z}$ is called a \emph{twist field} \cite{Kausch:2000fu}.  The character of $\SFIrr{\lambda}$ is given by
\begin{equation} \label{ch:SFTwist}
\ch{\SFIrr{\lambda}} = q^{-\lambda \brac{1 - \lambda} / 2 + 1/12} \prod_{n=0}^\infty \brac{1+q^{n - \lambda}} \brac{1+q^{n-1 + \lambda}} = \frac{1}{\func{\eta}{q}} \sum_{m \in \ZZ} q^{\brac{m + \lambda - 1/2}^2/2}
\end{equation}
and the supercharacter is obtained by inserting a factor of $\brac{-1}^m$ into the sum.  We will return to twisted modules when we consider the singlet algebra $\SingAlg{1,2}$ in \secref{sec:Singlet}.

\subsection{The Triplet Algebra $\TripAlg{1,2}$} \label{sec:Triplet}

The triplet algebra is defined to be the bosonic subalgebra of the symplectic fermion algebra.  It is 
generated by the three fields
\begin{equation}\label{eq:stronggenerators}
\func{W^\pm}{z}=\normord{ \func{J^\pm}{z} \func{\partial J^\pm}{z}},
\qquad \func{W^0}{z}=\normord{\func{J^+}{z}\func{\partial J^-}{z}}-\normord{\func{\partial J^+}{z}\func{J^-}{z}}
\end{equation}
and the energy-momentum tensor $\func{T}{z}$.  
%
The triplet fields $\func{W^\pm}{z}$ and $\func{W^0}{z}$ are Virasoro primaries (with respect to $\tfunc{T}{z}$) and the conformal dimension of each is $3$. Their \opes{} are rather unpleasant and may be found, for example, in \cite{KauCur95}.  A complete set of (untwisted) irreducible $\TripAlg{1,2}$-modules \cite{EhoHow93} may be obtained by decomposing the symplectic fermion irreducibles $\SFIrr{0}$ and $\SFIrr{1/2}$ into their bosonic and fermionic subspaces:
\begin{equation}
\SFIrr{0} = \TripIrr{0} \oplus \TripIrr{1}, \qquad \SFIrr{1/2} = \TripIrr{-1/8} \oplus \TripIrr{3/8}.
\end{equation}
Here, we have labelled the triplet modules by the conformal dimension of their ground states.  The space of ground states is one-dimensional for $\TripIrr{0}$ and $\TripIrr{-1/8}$, but two-dimensional for $\TripIrr{1}$ and $\TripIrr{3/8}$.  The indecomposable $\SFStag{0}$ also becomes a direct sum of two indecomposables when viewed as a $\TripAlg{1,2}$-module:
\begin{equation}
\SFStag{0} = \TripStag{0} \oplus \TripStag{1}.
\end{equation}
The Loewy diagrams of $\TripStag{0}$ and $\TripStag{1}$ are given in \figref{fig:SFTripLoewy}.

The fusion ring generated by the irreducibles was first determined in \cite{GabRat96}.  The vacuum module $\TripIrr{0}$ is again the fusion identity and the other rules are
\begin{equation} \label{FR:Trip}
\begin{gathered}
\begin{aligned}
\TripIrr{1} \fuse \TripIrr{1} &= \TripIrr{0}, \\
\TripIrr{-1/8} \fuse \TripIrr{-1/8} &= \TripStag{0},
\end{aligned}
\qquad
\begin{aligned}
\TripIrr{1} \fuse \TripIrr{-1/8} &= \TripIrr{3/8}, \\
\TripIrr{-1/8} \fuse \TripIrr{3/8} &= \TripStag{1},
\end{aligned}
\qquad
\begin{aligned}
\TripIrr{1} \fuse \TripIrr{3/8} &= \TripIrr{-1/8}, \\
\TripIrr{3/8} \fuse \TripIrr{3/8} &= \TripStag{0},
\end{aligned}
\\
\TripIrr{1} \fuse \TripStag{0} = \TripStag{1}, \qquad 
\TripIrr{1} \fuse \TripStag{1} = \TripStag{0}, \\
\TripIrr{-1/8} \fuse \TripStag{0} = \TripIrr{3/8} \fuse \TripStag{0} = \TripIrr{-1/8} \fuse \TripStag{1} = \TripIrr{3/8} \fuse \TripStag{1} = 2 \: \TripIrr{-1/8} \oplus 2 \: \TripIrr{3/8}, \\
\TripStag{0} \fuse \TripStag{0} = \TripStag{0} \fuse \TripStag{1} = \TripStag{1} \fuse \TripStag{1} = 2 \: \TripStag{0} \oplus 2 \: \TripStag{1}.
\end{gathered}
\end{equation}
We note that $\TripIrr{1}$ is a simple current of order $2$.  The corresponding extended algebra is, of course, the symplectic fermion algebra (the ground states of $\TripIrr{1}$ correspond to the currents $\func{J^{\pm}}{z}$) and it is easy to check that the symplectic fermion fusion rules \eqref{FR:SF} are consistent with \eqref{FR:Trip} and this observation.  The irreducible triplet characters were first obtained in \cite{Flohr:1995ea,KauCur95}, but follow easily from averaging the characters and supercharacters of the symplectic fermion irreducibles, given in \eqref{ch:SF}:
\begin{equation}
\begin{aligned}
\ch{\TripIrr{0}} &= \frac{\Jth{2}{1;q}}{4 \func{\eta}{q}} + \frac{\func{\eta}{q}^2}{2}, \\
\ch{\TripIrr{1}} &= \frac{\Jth{2}{1;q}}{4 \func{\eta}{q}} - \frac{\func{\eta}{q}^2}{2},
\end{aligned}
\qquad
\begin{aligned}
\ch{\TripIrr{-1/8}} &= \frac{\Jth{3}{1;q} + \Jth{4}{1;q}}{2 \func{\eta}{q}}, \\
\ch{\TripIrr{3/8}} &= \frac{\Jth{3}{1;q} - \Jth{4}{1;q}}{2 \func{\eta}{q}}.
\end{aligned}
\end{equation}
The modular properties of $\TripIrr{-1/8}$ and $\TripIrr{3/8}$ are seen to be good, but those of $\TripIrr{0}$ and $\TripIrr{1}$ are not as satisfactory because the $\func{\eta}{q}^2$ gives rise to coefficients involving $\log q = 2 \pi \ii \tau$.  For example,
\begin{equation}
\Sch{\TripIrr{0}} = \frac{\Jth{4}{1;q}}{4 \func{\eta}{q}} - \frac{\ii \tau \func{\eta}{q}^2}{2} = \frac{1}{4} \brac{\ch{\TripIrr{-1/8}} - \ch{\TripIrr{3/8}}} - \frac{\ii \tau}{2} \brac{\ch{\TripIrr{0}} - \ch{\TripIrr{1}}}.
\end{equation}
Attempts have been made to interpret this, see \cite{Flohr:1995ea,FucNon04} for example.

\subsection{The Singlet Algebra $\SingAlg{1,2}$} \label{sec:Singlet}

To define the singlet algebra at $c=-2$, it is convenient to extend the $\ZZ_2$-grading of the symplectic fermion algebra, given by parity, to a $\ZZ$-grading.  This may be regarded as the ghost number in the $bc$ ghost system realisation or as the eigenvalue of a derivation $N$ extending $\AKMSA{psl}{1}{1}$.  In any case, acting with $J^{\pm}_n$ increases this grade by $\pm 1$.  We can now define the singlet algebra $\SingAlg{1,2}$ as the subalgebra of symplectic fermions whose $\ZZ$-grade matches that of the vacuum.  This is therefore a subalgebra of the triplet algebra and it is generated by $\tfunc{T}{z}$ and $\tfunc{W^0}{z}$.  In contrast to the symplectic fermion currents $\func{J^{\pm}}{z}$ and the triplet fields $\func{W^{\pm}}{z}$, the singlet generators act with integer moding on every twisted symplectic fermion module.  Decomposing the symplectic fermions' twisted Verma modules $\SFVer{\lambda}$ into $\ZZ$-graded subspaces shows that the singlet algebra possesses an uncountable set of non-isomorphic (untwisted) modules $\SingTyp{\mu}$, $\mu \in \RR$:
\begin{equation} \label{eq:SFToSing}
\SFVer{\lambda} = \bigoplus_{m \in \ZZ} \SingTyp{\lambda + m} \qquad \text{(\( 0 \leqslant \lambda < 1 \)).}
\end{equation}
The ground states of $\SingTyp{\mu}$ have conformal dimension $\Delta_{\mu} = \tfrac{1}{2} \mu \brac{\mu - 1}$.

For $\lambda \neq 0$, it turns out that the $\SingTyp{\lambda + m}$ so obtained are irreducible.  By analogy with the case of superalgebras \cite{KacCha77}, the $\SingTyp{\mu}$ with $\mu \notin \ZZ$ will therefore be referred to as \emph{typical}.  When $\lambda = 0$, irreducible $\SingAlg{1,2}$-modules are obtained by decomposing the irreducible vacuum module $\SFIrr{0} \neq \SFVer{0}$ instead:
\begin{equation}
\SFIrr{0} = \bigoplus_{r \in \ZZ} \SingAtyp{r}.
\end{equation}
These irreducibles will be referred to as \emph{atypical}.  Note that the singlet vacuum module is $\SingAtyp{0}$ and that the minimal conformal dimension for states of $\SingAtyp{r}$ is $\tfrac{1}{2} \abs{r} \brac{\abs{r} + 1}$.  We also remark that the $\SingTyp{\mu}$ with $\mu \notin \ZZ$ and the $\SingAtyp{r}$ with $r \in \ZZ$ exhaust the irreducible $\SingAlg{1,2}$-modules \cite{EhoHow93,WanCla98}.

The characters of the typical irreducibles $\SingTyp{\mu}$, $\mu \notin \ZZ$, are easily extracted from the twisted symplectic fermion characters \eqref{ch:SFTwist} once the $\ZZ$-charge is taken into account:
\begin{equation} \label{ch:SingTypNaive}
\ch{\SingTyp{\mu}} = \frac{q^{\brac{\mu - 1/2}^2 / 2}}{\func{\eta}{q}}.
\end{equation}
This formula also applies to the indecomposables $\SingTyp{r}$, $r \in \ZZ$.  Note that up to shifting $\mu$ by $-\tfrac{1}{2}$, these characters coincide with the free boson characters \eqref{ch:BosonNaive} discussed in \secref{sec:FreeBoson}.  It is also straight-forward to obtain the characters of the atypical irreducibles $\SingAtyp{r}$, $r \in \ZZ$.  However, we will not need their explicit form in what follows and we only mention that these forms involve an interesting number theoretical object called a false theta function.  

Instead, we study the structure of the indecomposable modules $\SingTyp{r}$, for $r \in \ZZ$.  These are defined as the subspaces of $\SFVer{0}$ of constant $\ZZ$-grade.  It follows from \eqref{ses:SFStag} that $\SingTyp{r}$ is an indecomposable sum of two atypicals and a little thought leads us to the (non-split) exact sequence
\begin{equation} \label{ses:singlet}
\dses{\SingAtyp{r+1}}{}{\SingTyp{r+1}}{}{\SingAtyp{r}}.
\end{equation}
Splicing the short exact sequence for $\SingAtyp{r}$ with that for $\SingAtyp{r+1}$ and iterating (see \appref{app:Splicing}), we arrive at resolutions of the atypical irreducible modules:
\begin{equation} \label{eq:SingRes}
\cdots \lra \SingTyp{r+5} \lra \SingTyp{r+4} \lra \SingTyp{r+3} \lra \SingTyp{r+2} \lra \SingTyp{r+1} \lra \SingAtyp{r} \lra 0.
\end{equation}
These imply that the characters of the irreducible atypical modules may be expressed as infinite alternating sums over the typical characters:
\begin{equation} \label{eq:rescharsing}
\ch{\SingAtyp{r}} = \sum_{j=0}^{\infty} \Bigl( \ch{\SingTyp{r+2j+1}} - \ch{\SingTyp{r+2j+2}} \Bigr) = \sum_{j=0}^{\infty} \brac{-1}^j \ch{\SingTyp{r+j+1}}.
\end{equation}
In particular, we conclude that the characters of the $\SingTyp{\mu}$ --- the irreducible typicals as well as the indecomposable atypicals --- form a (topological) basis for the vector space spanned by the characters.  We will use this to apply the Verlinde formula to irreducible $\SingAlg{1,2}$-modules and thereby deduce the fusion rules of the singlet theory.

\subsection{Modular Transformations and the Verlinde Formula} \label{sec:SingMod}

Before deriving the modular transformations, we remark that the typical singlet characters \eqref{ch:SingTypNaive} suffer from the same deficiency as the standard free boson characters \eqref{ch:BosonNaive} in that they do not completely distinguish the representations:  $\ch{\SingTyp{\mu}} = \ch{\SingTyp{1 - \mu}}$.  As in \secref{sec:FreeBoson}, the fix is to include the $\ZZ$-grading and the $\AKMSA{psl}{1}{1}$ level $k$:
\begin{equation} \label{ch:SingTyp}
\ch{\SingTyp{\mu}} = \traceover{\SingTyp{\mu}} y^k z^{\mu - 1/2} q^{L_0 - c/24} = \frac{y^k z^{\mu - 1/2} q^{\brac{\mu - 1/2}^2 / 2k}}{\func{\eta}{q}}.
\end{equation}
Here, we have finally fixed our choice for the $\ZZ$-grading used to define singlet modules:  $\SingTyp{\mu}$ is assigned the grade $\mu - \tfrac{1}{2}$ in $\SFVer{\mu'}$ (where $\mu' = \mu \bmod{1}$).  We do this because the typical singlet character \eqref{ch:SingTyp} then takes the same form as the free boson character \eqref{ch:Boson}, up to the shifts by $-\tfrac{1}{2}$.

Writing $y = \ee^{2 \pi \ii t}$, $z = \ee^{2 \pi \ii u}$ and $q = \ee^{2 \pi \ii \tau}$, as in \secref{sec:FreeBoson}, the modular S-transformation for the typical (and indecomposable atypical) characters is then immediate from \eqref{eq:BosonS}:
\begin{equation} \label{eq:SingTypS}
\Sch{\SingTyp{\lambda}} = \int_{-\infty}^{\infty} \modS_{\lambda \mu} \ch{\SingTyp{\mu}} \: \dd \mu, \qquad S_{\lambda \mu} = \ee^{-2 \pi \ii \brac{\lambda - 1/2} \brac{\mu - 1/2}}
\end{equation}
(we have set $k$ back to $1$ for convenience here).  This S-matrix is again symmetric and unitary.  Moreover, with the T-matrix $\modT_{\lambda \mu} = \ee^{\ii \pi \brac{\lambda \brac{\lambda - 1} + 1/6}} \func{\delta}{\lambda = \mu}$, it defines a representation of $\SLG{SL}{2 ; \ZZ}$ in which the conjugation permutation is $\lambda \to 1 - \lambda$.

By analogy with \secref{sec:FreeBoson}, we expect that a continuum version of the Verlinde formula will be valid.  However, we must now take into account the fact that the vacuum $\SingAlg{1,2}$-module is the \emph{atypical} irreducible $\SingAtyp{0}$.  As there is no ``atypicality'' with the free boson's representations, our story now deviates from that of \secref{sec:FreeBoson}.  Using the character formula \eqref{eq:rescharsing} for atypical irreducibles, we easily obtain their S-transformations expressed in terms of the topological basis $\set{\ch{\SingTyp{\mu}} \st \mu \in \RR}$:
\begin{align}
\Sch{\SingAtyp{r}} &= \sum_{j=0}^\infty \brac{-1}^j \Sch{\SingTyp{r+j+1}} 
= \sum_{j=0}^\infty \brac{-1}^j \int_{-\infty}^{\infty} \ee^{-2 \pi \ii \brac{r+j+1/2} \brac{\mu - 1/2}} \ch{\SingTyp{\mu}} \: \dd \mu \notag \\
&= \int_{-\infty}^{\infty} \frac{\ee^{-2 \pi \ii \brac{r+1/2} \brac{\mu - 1/2}}}{1 + e^{-2 \pi \ii \brac{\mu - 1/2}}} \ch{\SingTyp{\mu}} \: \dd \mu 
= \int_{-\infty}^{\infty} \frac{\ee^{-2 \pi \ii r \brac{\mu - 1/2}}}{2 \cos \sqbrac{\pi \brac{\mu - 1/2}}} \ch{\SingTyp{\mu}} \: \dd \mu.
\end{align}
The corresponding S-matrix entry is therefore
\begin{equation} \label{eq:SingAtypS}
\modS_{\atyp{r} \mu} = \frac{\ee^{-2 \pi \ii r \brac{\mu - 1/2}}}{2 \cos \sqbrac{\pi \brac{\mu - 1/2}}},
\end{equation}
where we have indicated when a label corresponds to an atypical irreducible by underlining it.

We can now apply the continuum Verlinde formula to compute fusion coefficients.  Actually, because characters cannot distinguish an indecomposable from the direct sum of its composition factors, what the Verlinde formula gives is the structure constants of the \emph{Grothendieck} fusion ring (\appref{app:Grothendieck}).  The easiest computation is that for the fusion of an atypical and a typical:
\begin{equation}
\fuscoeff{\atyp{r} \mu}{\nu} = \int_{-\infty}^{\infty} \frac{\modS_{\atyp{r} \rho} \modS_{\mu \rho} \modS_{\nu \rho}^*}{\modS_{\atyp{0} \rho}} \: \dd \rho = \int_{-\infty}^{\infty} \ee^{-2 \pi \ii \brac{r + \mu - \nu} \brac{\rho - 1/2}} \: \dd \rho = \func{\delta}{\nu = \mu + r}.
\end{equation}
Fusing two typicals is only slightly more involved because the denominator of $\modS_{\atyp{0} \rho}$ no longer cancels:
\begin{align}
\fuscoeff{\lambda \mu}{\nu} &= \int_{-\infty}^{\infty} \ee^{-2 \pi \ii \brac{\lambda + \mu - \nu - 1/2} \brac{\rho - 1/2}} \brac{\ee^{\ii \pi \brac{\rho - 1/2}} + \ee^{-\ii \pi \brac{\rho - 1/2}}} \: \dd \rho \notag \\
&= \func{\delta}{\nu = \lambda + \mu} + \func{\delta}{\nu = \lambda + \mu - 1}.
\end{align}
Fusing two atypicals is a little more subtle however.  Computing na\"{\i}vely, one quickly arrives at a divergent integral.  The problem here may be traced back to the derivation of \eqref{eq:SingAtypS} in which we summed a geometric series at its radius of convergence.  The fix is obvious:  Expand the geometric series once again (in the right region) as continue to integrate.  From this perspective, the dubious summation may be simply regarded as a placeholder that simplifies some computations.  With this proviso, we quickly obtain
\begin{align}
\fuscoeff{\atyp{r} \atyp{s}}{\nu} &= \int_{-\infty}^\infty \frac{\ee^{-2 \pi \ii \brac{r+s - \nu + 1/2} \brac{\rho - 1/2}}}{2 \cos \sqbrac{\pi \brac{\rho - 1/2}}} \: \dd \rho = \sum_{j=0}^{\infty} \int_{-\infty}^\infty \brac{-1}^j \ee^{-2 \pi \ii \brac{r+s+j+1 - \nu} \sigma} \: \dd \sigma \notag \\
&= \sum_{j=0}^{\infty} \brac{-1}^j \func{\delta}{\nu = r+s+j+1}.
\end{align}
This seems to say that the fusion of two atypicals leads to negative multiplicities (for $j$ odd), but in fact, this infinite alternating sum corresponds to an atypical with positive multiplicity, as we will now see.

The Grothendieck ring of characters is obtained by integrating these coefficients as in \eqref{FR:Boson}:
\begin{equation}
\begin{gathered} 
\ch{\SingAtyp{r}} \grfuse \ch{\SingTyp{\mu}} = \ch{\SingTyp{\mu+r}}, \qquad 
\ch{\SingTyp{\lambda}} \grfuse \ch{\SingTyp{\mu}} = \ch{\SingTyp{\lambda + \mu}} + \ch{\SingTyp{\lambda + \mu - 1}}, \\
\ch{\SingAtyp{r}} \grfuse \ch{\SingAtyp{s}} = \sum_{j=0}^{\infty} \brac{-1}^j \ch{\SingTyp{r+s+j+1}} = \ch{\SingAtyp{r+s}}.
\end{gathered}
\end{equation}
When one is sure that the characters cannot describe indecomposable modules, these Grothendieck fusion rules may be lifted to genuine fusion rules.  In particular, we deduce that
\begin{equation} \label{FR:Sing}
\SingAtyp{r} \fuse \SingAtyp{s} = \SingAtyp{r+s}, \quad 
\SingAtyp{r} \fuse \SingTyp{\mu} = \SingTyp{\mu + r}, \quad 
\SingTyp{\lambda} \fuse \SingTyp{\mu} = \SingTyp{\lambda + \mu} \oplus \SingTyp{\lambda + \mu - 1} \qquad \text{(\( \lambda, \mu, \lambda + \mu \notin \ZZ \)),}
\end{equation}
the last constraint arising because the conformal dimensions of the states of $\SingTyp{\lambda}$ and $\SingTyp{\lambda - 1}$ differ by $\lambda \bmod{1}$.  This means that the fusion of two irreducibles is known in every case except $\SingTyp{\lambda} \fuse \SingTyp{\mu}$ when $\lambda + \mu \in \ZZ$.

Observe now that \eqref{FR:Sing} identifies the $\SingAtyp{r}$ as simple currents of infinite order (with no fixed points).  It is easy to show that the maximal simple current extension, meaning the algebra generated by $\SingAtyp{1}$ and $\SingAtyp{-1}$, is precisely the symplectic fermion algebra $\AKMSA{psl}{1}{1}$ (indeed, the generators of $\SingAtyp{1}$ and $\SingAtyp{-1}$ have conformal dimension $1$).  From this, we may conclude that the extension by $\SingAtyp{2}$ and $\SingAtyp{-2}$ is the triplet algebra $\TripAlg{1,2}$.  It is therefore the maximal \emph{bosonic} simple current extension.

Let us now perform a consistency check on the continuum Verlinde formula by deriving (some of) the $\TripAlg{1,2}$ fusion rules from \eqref{FR:Sing}.  Our identification of the triplet algebra as a simple current extension of the singlet algebra leads to the restriction rules
\begin{equation}
\TripIrr{0} = \bigoplus_{m \in \ZZ} \SingAtyp{2m}, \quad 
\TripIrr{1} = \bigoplus_{m \in \ZZ} \SingAtyp{2m+1}, \quad 
\TripIrr{-1/8} = \bigoplus_{m \in \ZZ} \SingTyp{2m+1/2}, \quad 
\TripIrr{3/8} = \bigoplus_{m \in \ZZ} \SingTyp{2m-1/2}.
\end{equation}
As in \secref{sec:FreeBoson}, we can use these rules to compute fusion, remembering to take a single $\SingAlg{1,2}$-representative for one of the $\TripAlg{1,2}$-modules being fused.  To illustrate:
\begin{equation}
\TripIrr{0} \fuse \TripIrr{0} = \SingAtyp{0} \fuse \Bigl( \bigoplus_{m \in \ZZ} \SingAtyp{2m} \Bigr) = \bigoplus_{m \in \ZZ} \brac{\SingAtyp{0} \fuse \SingAtyp{2m}} = \bigoplus_{m \in \ZZ} \SingAtyp{2m} = \TripIrr{0}.
\end{equation}
This procedure therefore correctly normalises the extended algebra fusion.  Applying this, we can reproduce all the triplet fusion rules \eqref{FR:Trip} except for those of $\TripIrr{-1/8}$ and $\TripIrr{3/8}$ with one another.  For these latter rules, we can only compute the Grothendieck fusion, for example
\begin{align}
\ch{\TripIrr{-1/8}} \grfuse \ch{\TripIrr{3/8}} &= \ch{\SingTyp{2 \ell + 1/2}} \grfuse \Bigl( \sum_{m \in \ZZ} \ch{\SingTyp{2m-1/2}} \Bigr) = \sum_{m \in \ZZ} \Bigl( \ch{\SingTyp{2 \brac{\ell + m}}} + \ch{\SingTyp{2 \brac{\ell + m} - 1}} \Bigr) \notag \\
&= \sum_{m \in \ZZ} \Bigl( \ch{\SingAtyp{2 \brac{\ell + m}}} + 2 \: \ch{\SingAtyp{2 \brac{\ell + m} - 1}} + \ch{\SingAtyp{2 \brac{\ell + m} - 2}} \Bigr) \notag \\
&= 2 \: \ch{\TripIrr{0}} + 2 \: \ch{\TripIrr{1}},
\end{align}
where we have used \eqref{ses:singlet}.  We note that this is consistent with $\TripIrr{-1/8} \fuse \TripIrr{3/8} = \TripStag{1}$ because (see \figref{fig:SFTripLoewy})
\begin{equation} \label{eq:TripStagId}
\ch{\TripStag{0}} = \ch{\TripStag{1}} = 2 \: \ch{\TripIrr{0}} + 2 \: \ch{\TripIrr{1}}.
\end{equation}

Of course, the Grothendieck fusion of the $\SingAlg{1,2}$-typicals $\SingTyp{\lambda}$ and $\SingTyp{\mu}$, $\lambda + \mu \in \ZZ$, gives us the composition factors of the fusion product:
\begin{equation}
\ch{\SingTyp{\lambda} \fuse \SingTyp{\mu}} = \ch{\SingTyp{\lambda}} \grfuse \ch{\SingTyp{\mu}} = \ch{\SingAtyp{\lambda + \mu}} + 2 \: \ch{\SingAtyp{\lambda + \mu - 1}} + \ch{\SingAtyp{\lambda + \mu - 2}}.
\end{equation}
Because the corresponding fusion products for the triplet and symplectic fermion algebras are indecomposable, we propose that this is true for the singlet algebra as well.  We therefore conjecture that
\begin{equation}
\SingTyp{\lambda} \fuse \SingTyp{r+1 - \lambda} = \SingStag{r} \qquad \text{(\( r \in \ZZ \)),}
\end{equation}
where $\SingStag{r}$ is an indecomposable $\SingAlg{1,2}$-module whose Loewy diagram is given in \figref{fig:SingLoewy}.  To prove this, one would have to either construct the fusion product explicitly (which seems very demanding), or deduce the existence of such indecomposables abstractly and show that they describe the restriction of the $\TripAlg{1,2}$-indecomposables $\TripStag{0}$ and $\TripStag{1}$ to $\SingAlg{1,2}$.  Either approach is beyond the scope of this review.

\begin{figure}
\begin{center}
\begin{tikzpicture}[thick,>=latex,
	nom/.style={circle,draw=black!20,fill=black!20,inner sep=1pt}
	]
\node (top1) at (5,1.5) [] {$\SingAtyp{r}$};
\node (left1) at (3.5,0) [] {$\SingAtyp{r+1}$};
\node (right1) at (6.5,0) [] {$\SingAtyp{r-1}$};
\node (bot1) at (5,-1.5) [] {$\SingAtyp{r}$};
\node at (5,0) [nom] {$\SingStag{r}$};
\draw [->] (top1) -- (left1);
\draw [->] (top1) -- (right1);
\draw [->] (left1) -- (bot1);
\draw [->] (right1) -- (bot1);
\end{tikzpicture}
\caption{\label{fig:SingLoewy} Conjectured Loewy diagram for the proposed indecomposable $\SingAlg{1,2}$-module $\SingStag{r}$.}
\end{center}
\end{figure}

\subsection{Bulk Modular Invariants} \label{sec:SingModInv}

Because of the symmetries of the S-matrix, the singlet theory has two obvious bulk modular invariants, corresponding to the diagonal and charge-conjugate partition functions:
\begin{equation}
\func{\partfunc{diag.}}{q,\ahol{q}} = \int_{-\infty}^{\infty} \ahol{\ch{\SingTyp{\lambda}}} \ch{\SingTyp{\lambda}} \: \dd \lambda, \qquad \func{\partfunc{c.c.}}{q,\ahol{q}} = \int_{-\infty}^{\infty} \ahol{\ch{\SingTyp{1-\lambda}}} \ch{\SingTyp{\lambda}} \: \dd \lambda.
\end{equation}
As usual, simple current extensions allow one to construct more.  Specifically, we have seen that each atypical $\SingAlg{1,2}$-module $\SingAtyp{n}$ is a simple current, so an extended algebra $\ExtAlg{n}$ may be constructed by promoting all the fields in the fusion orbit of the vacuum module to symmetry generators:
\begin{equation}
\ExtAlg{n} = \bigoplus_{m \in \ZZ} \SingAtyp{mn}.
\end{equation}
We have already remarked that $\ExtAlg{1}$ is the symplectic fermion algebra and $\ExtAlg{2}$ is the triplet algebra.  The other extended algebras likewise give rise to rational \lcfts{} and may be described as orbifolds of the symplectic fermion theory \cite{KauCur95}.  The fusion orbits through other $\SingAlg{1,2}$-modules likewise give rise to (twisted) $\ExtAlg{n}$-modules, for example
\begin{equation}
\ExtSingAtyp{n}{r} = \bigoplus_{m \in \ZZ} \SingAtyp{r + mn}, \qquad \ExtSingTyp{n}{\lambda} = \bigoplus_{m \in \ZZ} \SingTyp{\lambda + mn}.
\end{equation}
The untwisted extended algebra modules are precisely the $\ExtSingAtyp{n}{r}$ and the $\ExtSingTyp{n}{\lambda}$ with $\lambda \in \frac{1}{n} \ZZ$.  By restricting to these, we arrive at new modular invariants.

The modular S-transformation of an untwisted typical extended algebra module is easy to calculate:
\begin{align}
\Sch{\ExtSingTyp{n}{j/n}} &= \sum_{m \in \ZZ} \Sch{\SingTyp{j/n+mn}} = 
\sum_{m \in \ZZ} \int_{-\infty}^{\infty} \ee^{-2 \pi \ii \brac{j/n+mn-1/2} \brac{\mu - 1/2}} \ch{\SingTyp{\mu}} \: \dd \mu \notag \\
&=\int_{-\infty}^{\infty} \sum_{m \in \ZZ} \ee^{-2 \pi \ii mn \brac{\mu - 1/2}} \ee^{-2 \pi \ii \brac{j/n-1/2} \brac{\mu - 1/2}} \ch{\SingTyp{\mu}} \: \dd \mu \notag \\
&= \frac{1}{n} \int_{-\infty}^{\infty} \sum_{\ell \in \ZZ} \func{\delta}{\mu - 1/2 = \ell / n} \ee^{-2 \pi \ii \brac{j/n-1/2} \brac{\mu - 1/2}} \ch{\SingTyp{\mu}} \: \dd \mu \notag \\
&= \frac{1}{n} \sum_{\ell \in \ZZ} \ee^{-2 \pi \ii \brac{j/n-1/2} \ell / n} \ch{\SingTyp{\ell / n + 1/2}} \notag \\
&= \frac{1}{n} \sum_{m \in \ZZ} \sum_{k=0}^{n^2-1} \ee^{-2 \pi \ii \brac{j/n-1/2} \brac{k/n+mn}} \ch{\SingTyp{k/n+mn+1/2}}.
\end{align}
Here, we pause to note that if $n$ is even, then the exponential factor in the above sum is independent of $m$.  We may therefore perform the $m$-summation and then shift $k$ to $k - \tfrac{1}{2} n$, obtaining
\begin{equation} \label{eq:ExtSingS}
\Sch{\ExtSingTyp{n}{j/n}} = \sum_{k=0}^{n^2-1} \ExtmodS{n}_{jk} \ch{\ExtSingTyp{n}{k/n}}, \quad \ExtmodS{n}_{jk} = \frac{1}{n} \ee^{-2 \pi \ii \brac{j/n-1/2} \brac{k/n-1/2}} \qquad \text{(\( n \) even).}
\end{equation}
The \emph{typical} extended characters, for $n$ even, therefore carry a finite-dimensional representation of the modular group.  If $n$ is odd however, then the exponential factor includes a factor $\brac{-1}^m$ and the $\SingTyp{\lambda}$ do not combine to give an untwisted extended algebra module.  One is instead forced to consider supercharacters and twisted modules, the final result being that the modular invariant one constructs from $\ExtAlg{n}$, with $n$ odd, is equivalent to that constructed from $\ExtAlg{2n}$.  This is consistent with expectations because the extended algebra generators are fermionic for $n$ odd and bosonic for $n$ even.

As the S-matrix \eqref{eq:ExtSingS} is unitary, the diagonal and charge-conjugate partition functions are modular invariant:
\begin{equation}
\partfunc{diag.}^{\brac{n}} = \sum_{j=0}^{n^2-1} \ahol{\ch{\ExtSingTyp{n}{j/n}}} \ch{\ExtSingTyp{n}{j/n}}, \qquad \partfunc{c.c.}^{\brac{n}} = \sum_{j=0}^{n^2-1} \ahol{\ch{\ExtSingTyp{n}{1-j/n}}} \ch{\ExtSingTyp{n}{j/n}}.
\end{equation}
Expressing these in terms of singlet characters finally gives new modular invariants for $\SingAlg{1,2}$.  We will not write them out in generality, noting only that for the triplet algebra $\TripAlg{1,2} = \ExtAlg{2}$, the diagonal and charge conjugate modular invariant coincide.  Since $\ch{\ExtSingTyp{2}{0}} = \ch{\ExtSingTyp{2}{1}} = \ch{\TripIrr{0}} + \ch{\TripIrr{1}}$, they are 
\begin{equation} \label{eq:TripModInv}
\partfunc{diag.}^{\brac{2}} = \abs{\ch{\TripIrr{-1/8}}}^2 + \abs{\ch{\TripIrr{3/8}}}^2 + 2 \: \abs{\ch{\TripIrr{0}} + \ch{\TripIrr{1}}}^2.
\end{equation}
Note that the obvious candidate for the non-chiral vacuum module $\TripIrr{0} \otimes \TripIrr{0}$ contributes with multiplicity $2$.  We will see that this is explained by the bulk vacuum $\ket{0} \otimes \ket{0}$ having a non-chiral logarithmic partner $\ket{\nonch{\Omega}}$.

Finally, we remark that the modular transformations of the atypical extended characters are problematic.  Repeating the calculation that led to the typical extended S-matrix \eqref{eq:ExtSingS}, assuming again that $n$ is even, leads to the sum
\begin{equation}
\Sch{\ExtSingAtyp{n}{r}} = \frac{1}{n} \sum_{k=0}^{n^2-1} \frac{\ee^{-2 \pi \ii rk/n}}{2 \cos \sqbrac{\pi k/n}} \ch{\ExtSingTyp{n}{k/n + 1/2}}.
\end{equation}
The pole at $k = \tfrac{1}{2} n$ cannot be swept aside in this summation as it was when we were integrating.  Of course, we know for the case $n=2$ ($\TripAlg{1,2}$) that the S-transformation of the atypical characters involves factors of $\log q = 2 \pi \ii \tau$, so we should not expect that the above approach will work.  This can be traced back to the fact that the atypical extended characters can, unlike their singlet counterparts, no longer be written as an (infinite) linear combination of the typical extended characters.  The extended version of the resolution \eqref{eq:SingRes} is periodic in the typical labels and, consequently, the character formula \eqref{eq:rescharsing} is divergent.

\subsection{Bulk State Spaces} \label{sec:TripBulk}

Let us consider the symplectic fermions once again.  Recall that the action \eqref{eq:SFAction} is defined in terms of (non-chiral) fermions $\func{\nonch{\theta}^{\pm}}{z , \ahol{z}}$.  From these, we can construct the following field:
\begin{equation}
\func{\nonch{\Omega}}{z,\ahol{z}} = \normord{\func{\nonch{\theta}^+}{z,\ahol{z}} \func{\nonch{\theta}^-}{z,\ahol{z}}}.
\end{equation}
The symplectic fermion currents, both holomorphic and antiholomorphic, act on $\nonch{\Omega}$ as follows:
\begin{equation} \label{eq:SFLogOPEs}
\begin{aligned}
\func{J^\pm}{z} \func{\nonch{\Omega}}{w,\ahol{w}} &\sim \mp \frac{\func{\nonch{\theta}^\pm}{w,\ahol{w}}}{z-w}, \\
\func{\ahol{J}^\pm}{\ahol{z}} \func{\nonch{\Omega}}{w,\ahol{w}} &\sim \mp \frac{\func{\nonch{\theta}^\pm}{w,\ahol{w}}}{\ahol{z}-\ahol{w}},
\end{aligned}
\qquad
\begin{aligned}
\func{J^\pm}{z} \func{\nonch{\theta}^\mp}{w,\ahol{w}} &\sim \frac{1}{z-w}, \\
\func{\ahol{J}^\pm}{\ahol{z}} \func{\nonch{\theta}^\mp}{w,\ahol{w}} &\sim \frac{1}{\ahol{z}-\ahol{w}}.
\end{aligned}
\end{equation}
From this, we deduce the structure of the bulk indecomposable symplectic fermion module $\SFStagbulk{}$, see \figref{fig:SFbulk}.  We remark that its character is \emph{twice} the atypical contribution to the $\TripAlg{1,2}$ modular invariant \eqref{eq:TripModInv}.

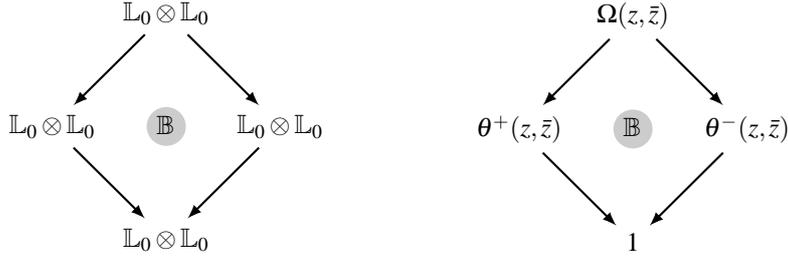
\begin{figure}
\begin{center}
\begin{tikzpicture}[thick,>=latex,
	nom/.style={circle,draw=black!20,fill=black!20,inner sep=1pt}
	]
\node (top1) at (5,1.5) [] {$\SFIrr{0}\otimes \SFIrr{0}$};
\node (left1) at (3.5,0) [] {$\SFIrr{0}\otimes \SFIrr{0}$};
\node (right1) at (6.5,0) [] {$\SFIrr{0}\otimes \SFIrr{0}$};
\node (bot1) at (5,-1.5) [] {$\SFIrr{0}\otimes \SFIrr{0}$};
\node at (5,0) [nom] {$\SFStagbulk{}$};
\draw [->] (top1) -- (left1);
\draw [->] (top1) -- (right1);
\draw [->] (left1) -- (bot1);
\draw [->] (right1) -- (bot1);
\end{tikzpicture}
\hspace{0.1\textwidth}
\begin{tikzpicture}[thick,>=latex,
	nom/.style={circle,draw=black!20,fill=black!20,inner sep=1pt}
	]
\node (top0) at (0,1.5) [] {$\func{\nonch{\Omega}}{z,\bar z}$};
\node (left0) at (-1.5,0) [] {$\func{\nonch{\theta}^+}{z,\bar z}$};
\node (right0) at (1.5,0) [] {$\func{\nonch{\theta}^-}{z,\bar z}$};
\node (bot0) at (0,-1.5) [] {$\nonch{1}$};
\node at (0,0) [nom] {$\SFStagbulk{}$};
\draw [->] (top0) -- (left0);
\draw [->] (top0) -- (right0);
\draw [->] (left0) -- (bot0);
\draw [->] (right0) -- (bot0);
\end{tikzpicture}
\caption{\label{fig:SFbulk}On the left, the Loewy diagram of the bulk indecomposable module $\SFStagbulk{}$ (over the direct sum of two copies of the symplectic fermion algebra).  On the right, the same diagram but with the bulk composition factors replaced by the fields naturally associated to them.}
\end{center}
\end{figure}

The symplectic fermion bulk module is graded by the (total) fermion number, hence it decomposes into the direct sum of two bulk modules over $\TripAlg{1,2}$ (or rather over two copies of it).  Only one of these modules $\TripStagbulk{}$ contains the triplet vacuum module $\TripIrr{0} \otimes \TripIrr{0}$ as a submodule.  We illustrate its structure in \figref{fig:TripBulkMod}.  As its character coincides with the atypical contribution to the partition function \eqref{eq:TripModInv}, we conclude that the bulk space of states corresponding to this modular invariant is
\begin{equation} \label{eq:TripBulkSpace}
\bulkstatespace = \brac{\TripIrr{-1/8} \otimes \TripIrr{-1/8}} \oplus \brac{\TripIrr{3/8} \otimes \TripIrr{3/8}} \oplus \TripStagbulk{}.
\end{equation}
This conclusion essentially defines the bulk triplet theory as the bosonic subtheory of the bulk symplectic fermions theory in which the only non-local fields admitted are those on which the fermions act with half-integer moding.  This confirms the triplet model as the $\ZZ_2$-orbifold of symplectic fermions.

\begin{figure}
\begin{center}
\begin{tikzpicture}[thick,>=latex]
\node (t11) at (10,4) [] {$ \TripIrr{1} \otimes \TripIrr{1}$};
\node (t33) at (4,4) [] {$ \TripIrr{0} \otimes \TripIrr{0}$};
\node (b11) at (10,0) [] {$ \TripIrr{1} \otimes \TripIrr{1}$};
\node (b33) at (4,0) [] {$ \TripIrr{0} \otimes \TripIrr{0}$};
\node (01') at (12,2) [] {$ \TripIrr{1} \otimes \TripIrr{0}$};
\node (31') at (8,2) [] {$ \TripIrr{1} \otimes \TripIrr{0}$};
\node (13') at (6,2) [] {$ \TripIrr{0} \otimes \TripIrr{1}$};
\node (63') at (2,2) [] {$ \TripIrr{0} \otimes \TripIrr{1}$};
\draw [->] (t11) -- (01');
\draw [->] (01') -- (b11);
\draw [->] (t11) -- (31');
\draw [->] (31') -- (b11);
\draw [->] (t33) -- (13');
\draw [->] (13') -- (b33);
\draw [->] (t33) -- (63');
\draw [->] (63') -- (b33);
\draw [->,dotted] (t11) -- (63');
\draw [->,dotted] (13') -- (b11);
\draw [->,dotted] (t33) -- (31');
\draw [->,dotted] (31') -- (b33);
\draw [->,dotted] (t11) -- (13');
\draw [->,dotted] (t33) -- (01');
\draw [->,dotted] (01') -- (b33);
\draw [->,dotted] (63') -- (b11);
\end{tikzpicture}
\caption{\label{fig:TripBulkMod}The Loewy diagram for the indecomposable bulk module $\TripStagbulk{}$ of $\TripAlg{1,2} \oplus \TripAlg{1,2}$.  The solid/dotted arrows indicate the action of the holomorphic/antiholomorphic triplet algebra.  It is clear that restricting to each chiral subalgebra results in $\TripStagbulk{}$ decomposing as $\brac{\TripIrr{0} \otimes \TripStag{0}} \oplus \brac{\TripIrr{1} \otimes \TripStag{1}}$ and $\brac{\TripStag{0} \otimes \TripIrr{0}} \oplus \brac{\TripStag{1} \otimes \TripIrr{1}}$, respectively.}
\end{center}
\end{figure}

It is worth thinking for a few moments how one could have arrived at the triplet model's bulk state space structure \eqref{eq:TripBulkSpace} if the non-chiral information concerning symplectic fermions was not so readily available.  First, the modular invariant in \eqref{eq:TripModInv} is very suggestive, especially when we may rewrite the atypical contribution, using the character identity \eqref{eq:TripStagId}, as
\begin{equation}
2 \: \abs{\ch{\TripIrr{0}} + \ch{\TripIrr{1}}}^2 = \ahol{\ch{\TripIrr{0}}} \ch{\TripStag{0}} + \ahol{\ch{\TripIrr{1}}} \ch{\TripStag{1}} = \ahol{\ch{\TripStag{0}}} \ch{\TripIrr{0}} + \ahol{\ch{\TripStag{1}}} \ch{\TripIrr{1}}.
\end{equation}
This ties in nicely with the idea discussed in the introduction that natural representations often decompose in a manner whereby each irreducible is paired with its projective cover.  Indeed, it is known \cite{NagTri11} that the irreducibles $\TripIrr{-1/8}$ and $\TripIrr{3/8}$ are projective and that the projective covers of the irreducibles $\TripIrr{0}$ and $\TripIrr{1}$ are $\TripStag{0}$ and $\TripStag{1}$, respectively (in an appropriate category of vertex algebra modules).

A simple guess, which works in this case, is therefore to draw the Loewy diagram of the direct sum $\brac{\TripIrr{0} \otimes \TripStag{0}} \oplus \brac{\TripIrr{1} \otimes \TripStag{1}}$.  This gives a holomorphic module structure to the bulk module $\TripStagbulk{}$.  Then, complete the diagram by adding (dotted) arrows representing the antiholomorphic module structure so that they trace out the Loewy diagram of $\brac{\TripStag{0} \otimes \TripIrr{0}} \oplus \brac{\TripStag{1} \otimes \TripIrr{1}}$.  One quickly finds that there is a unique way to do this.  Moreover, one can check that the resulting diagram is manifestly \emph{local}, meaning that the corresponding bulk correlation functions will all be single-valued.  For this, we recall \cite{GabLoc99} that field locality has an algebraic reformulation which requires that $L_0 - \ahol{L}_0$ be diagonalisable with integer eigenvalues.  The second constraint is clearly met and the first follows from the non-diagonalisable action for symplectic fermions (see \eqref{eq:SFGCR} and \eqref{eq:SFLogOPEs}):
\begin{equation}
L_0 \ket{\Omega} = J_0^- J_0^+ \ket{\Omega} = - J_0^- \ket{\Omega} = -\ket{0}, \qquad 
\ahol{L}_0 \ket{\Omega} = \ahol{J}_0^- \ahol{J}_0^+ \ket{\Omega} = - \ahol{J}_0^- \ket{\Omega} = -\ket{0}.
\end{equation}

\subsection{Correlation Functions} \label{sec:corr}

Correlation functions and \opes{} for symplectic fermions, including twist fields, were computed by 
Kausch \cite{Kausch:2000fu}.  The \ope{} of the logarithmic partner $\tfunc{\nonch{\Omega}}{z,\ahol{z}}$ of the vacuum with itself is logarithmic,
\begin{equation}
\func{\nonch{\Omega}}{z,\ahol{z}} \func{\nonch{\Omega}}{w,\ahol{w}} = -\brac{A + \log \abs{z-w}^2}^2 - 2 \brac{A + \log \abs{z-w}^2} \func{\nonch{\Omega}}{w,\ahol{w}} + \cdots,
\end{equation}
and the same is true for the twist field $\tfunc{\nonch{\mu}_{1/2}}{z,\ahol{z}}$ with itself:
\begin{equation}\label{eq:OPEmumu}
\func{\nonch{\mu}_{1/2}}{z,\ahol{z}} \func{\nonch{\mu}_{1/2}}{w, \ahol{w}} = -\abs{z-w}^{1/2} \brac{\func{\nonch{\Omega}}{w,\ahol{w}} + \log \abs{z-w}^2 + \text{const}} + \cdots
\end{equation}
In the triplet theory, the twist field generates the bulk triplet module $\TripIrr{-1/8}\otimes \TripIrr{-1/8}$, while the logarithmic partner of the identity is associated to the top composition factor $\TripIrr{0}\otimes \TripIrr{0}$ of the bulk indecomposable.  These \opes{} are, of course, consistent with the fusion rules and the bulk state space explained in last section.  We will now outline an efficient means of computing the correlation functions that imply \eqref{eq:OPEmumu} (using a different approach to Kausch).

Viewing symplectic fermions as the \WZW{} model of the Lie supergroup $\SLSG{PSL}{1}{1}$, it is natural to describe the theory by passing to a first order formulation as in, for example, \cite{Quella:2007hr}.  For symplectic fermions, this idea has been applied in \cite{Creutzig:2008an}. The picture is sketched as follows:
\begin{equation}
\pd \nonch{\theta}^+ \apd \nonch{\theta}^- \xrightarrow{\text{\(1^{\text{st}}\) order}} \nonch{b}^+ \pd \nonch{\theta}^+ + \nonch{b}^- \apd \nonch{\theta}^- + \nonch{b}^+ \nonch{b}^- \xrightarrow{\text{bosonisation}} -\pd \nonch{\varphi} \apd \nonch{\varphi} + \ee^{-\nonch{\varphi}} + \text{linear dilaton}.
\end{equation}
This first order formulation, or its equivalent bosonisation, is well suited to computing correlation functions perturbatively.  The symplectic fermion fields are recovered by decomposing $\ee^{\pm \func{\varphi}{z ,\ahol{z}}} = \ee^{\pm \func{\varphi_L}{z}} \ee^{\pm \func{\varphi_R}{\ahol{z}}}$ with
\begin{equation}
e^{\func{\varphi_L}{z}}e^{-\func{\varphi_L}{w}}\sim \frac{1}{(z-w)}
\end{equation}
and so on.  The fields $\tfunc{J^+}{z} = \pd \ee^{\varphi_L(z)}$ and $\tfunc{J^-}{z} = \ee^{-\varphi_L(z)}$ commute with the zero-mode of $\ee^{-\varphi_L(z)}$ and they have the same \ope{} as symplectic fermions.  The (holomorphic) Virasoro field is
\begin{equation}
\func{T}{z} = \frac{1}{2} \normord{\pd \func{\nonch\varphi}{z, \ahol{z}} \pd \func{\nonch\varphi}{z, \ahol{z}}} + \frac{1}{2} \pd^2 \func{\nonch\varphi}{z, \ahol{z}}.
\end{equation}
The twist fields $\mu_{\lambda}$ are identified with the $\ee^{\lambda \varphi_L(z)}$ because the corresponding states have dimension $-\lambda \brac{1 - \lambda} / 2$ and the $J^\pm$ act on them with moding in $\ZZ \pm \lambda$.  The corresponding bulk field will be denoted by $\tfunc{\nonch V_{\lambda}}{z,\ahol{z}} = \ee^{\lambda \nonch\varphi(z, \ahol{z})}$. 

In this formalism, correlation functions are defined by
\begin{equation} \label{eq:DefCorr}
\corrfn{\func{\nonch V_{\alpha_1}}{z_1,\ahol{z}_1} \cdots \func{\nonch V_{\alpha_n}}{z_n,\ahol{z}_n}} = \corrfn{\func{\nonch V_{\alpha_1}}{z_1,\ahol{z}_1} \cdots \func{\nonch V_{\alpha_n}}{z_n,\ahol{z}_n} \ee^{-Q}}_0,
\end{equation}
where $\ee^{-Q}$ should be interpreted as the power series $1 - Q + \tfrac{1}{2} Q^2 - \cdots$,
\begin{equation} \label{eq:DefScreen}
Q = \int_\CC \ee^{-\nonch\varphi (z,\ahol{z})} \: \frac{\dd z \dd \ahol{z}}{2 \pi},
\end{equation}
and the correlators in the free theory, denoted with a subscript $0$, are standard free boson correlators subject to the charge conservation condition
\begin{equation}
\corrfn{\func{\nonch V_{\alpha_1}}{z_1,\ahol{z}_1} \cdots \func{\nonch V_{\alpha_n}}{z_n,\ahol{z}_n}}_0 = -\delta_{\alpha_1 + \cdots + \alpha_n = 1} \prod_{i<j} \abs{z_i-z_j}^{2 \alpha_i \alpha_j}. 
\end{equation}
Because $Q$ carries charge $-1$, the interacting correlators can only be non-zero when the labels $\alpha_i$ sum to a positive integer and then only one term, that with the appropriate power of $Q$, contributes.  In this manner, we obtain the non-zero one-point and two-point functions
\begin{equation}
\corrfn{\func{\nonch V_{1}}{z_1,\ahol{z}_1}} =-1,\qquad
\corrfn{\func{\nonch V_{\alpha}}{z_1,\ahol{z}_1} \func{\nonch V_{1 - \alpha}}{z_2,\ahol{z}_2}} = -\abs{z_1-z_2}^{2\alpha \brac{1-\alpha}}.
\end{equation}

The logarithmic partner $\tfunc{\nonch\Omega}{z,\ahol{z}}$ of the identity $\tfunc{V_0}{z, \ahol{z}}$ cannot be obtained directly within this free field realisation.  However, we can use a regularisation trick to compute correlators with logarithmic singularities from which the existence of $\tfunc{\nonch\Omega}{z,\ahol{z}}$ can be surmised.  To do this, we consider the field $\func{\nonch V_1}{z, \ahol{z}}$ and the three-point functions of the form $\func{C_{\alpha \beta \gamma}}{z,\ahol{z}} = \corrfn{\func{\nonch V_{\alpha}}{z,\ahol{z}} \func{\nonch V_{\beta}}{1,1} \func{\nonch V_{\gamma}}{0,0}}$, with $\alpha + \beta + \gamma = 2$.  The latter may be computed using the famous Fateev-Dotsenko integral \cite{DotCon84}:
\begin{align} \label{eq:SFcor}
\func{C_{\alpha \beta \gamma}}{z,\ahol{z}} &= -\abs{z}^{2 \alpha \gamma} \abs{z-1}^{2 \alpha \beta} \int_\CC \abs{z-w}^{-2 \alpha} \abs{1-w}^{-2 \beta} \abs{w}^{-2 \gamma} \: \frac{\dd z \dd \ahol{z}}{2 \pi } \notag \\
&= -\abs{z}^{2 \brac{\alpha \gamma + \beta - 1}} \abs{1-z}^{2 \brac{\alpha \beta + \gamma - 1}} \frac{\func{\Gamma}{1-\alpha} \func{\Gamma}{1-\beta} \func{\Gamma}{1-\gamma}}{\func{\Gamma}{\alpha} \func{\Gamma}{\beta} \func{\Gamma}{\gamma}}.
\end{align}
These correlators diverge when $\func{\nonch V_1}{z, \ahol{z}}$ is one of their fields, in which case, we regularise by assuming that the physical field is the $\eps \to 0$ limit of an $\eps$-dependent linear combination of $V_{1+\eps}$ and $V_{\eps}$.  For example,
\begin{align}
\func{C_{\alpha, 1-\alpha, 1}}{z,\ahol{z}}_{\text{reg.}} 
&= -\lim_{\eps \to 0} \Bigl( \func{C_{\alpha, 1-\alpha-\eps, 1+\eps}}{z,\ahol{z}} - \func{\Gamma}{-\eps} \func{C_{\alpha, 1-\alpha-\eps, \eps}}{z,\ahol{z}} \Bigr) \notag \\
&= -\lim_{\eps \to 0} \biggl( \abs{z}^{2 \alpha \eps} \abs{1-z}^{2 \alpha \brac{1-\alpha-\eps}}
\Bigl( \frac{\func{\Gamma}{1-\alpha} \func{\Gamma}{\alpha + \eps} \func{\Gamma}{-\eps}}{\func{\Gamma}{\alpha} \func{\Gamma}{1-\alpha-\eps} \func{\Gamma}{1+\eps}} \frac{\abs{1-z}^{2 \eps}}{\abs{z}^{2 \eps}} - \func{\Gamma}{-\eps} \Bigr) \biggr) \notag \\
&= -\lim_{\eps \to 0} \biggl( \abs{1-z}^{2 \alpha \brac{1-\alpha}} \eps \func{\Gamma}{-\eps} \Bigl( \func{\psi}{\alpha} + \func{\psi}{1-\alpha} - \func{\psi}{1} + \brac{\alpha-1} \log \abs{z}^2 \Bigr. \biggr. \notag \\
&\mspace{240mu} \biggl. \Bigl. + \brac{1-\alpha} \log \abs{1-z}^2 - \alpha \log \abs{z}^2 + \alpha \log \abs{1-z}^2 \Bigr) \biggr) \notag \\
&= -\abs{1-z}^{2 \alpha \brac{1-\alpha}} \Bigl( \func{\psi}{\alpha} + \func{\psi}{1-\alpha} - \func{\psi}{1} - \log \abs{z}^2 + \log \abs{1-z}^2 \Bigr).
\end{align}
Here, $\func{\psi}{x} = \func{\Gamma'}{x} / \func{\Gamma}{x}$ is the logarithmic derivative of $\Gamma(x)$.  
Setting $\alpha$ to $\tfrac{1}{2}$, we recover the correlator which gives the coefficient of the identity in the \ope{} \eqref{eq:OPEmumu}.  We remark that this regularisation prescription cancels the pole in the divergent correlator by subtracting a correlator with the identity.  This is the same procedure that Gurarie and Ludwig employed to sidestep the well known ``$c \to 0$ catastrophe'' in \cite{GurCon04}.

\subsection{Further Developments} \label{sec:TripFuture}

The singlet algebras $\SingAlg{1,p}$ and triplet algebras $\TripAlg{1,p}$, with $p = 2, 3, \ldots$ and central charge $c = 1 - 6 \brac{p-1}^2 / p$, were first discovered by Kausch \cite{Kausch:1991}.  They are generated by $1$ and $3$ fields of dimension $2p-1$, respectively (along with the energy-momentum tensor).  The singlet, being a subalgebra of the triplet, has received relatively little attention, but the $\TripAlg{1,p}$-models have been considered by a variety of groups \cite{FucNon04,FeiMod06,CarNon06,GabFro08,PeaInt08,BusLus09,PeaGro10}.  In particular, we now have a complete picture of the $\TripAlg{1,p}$ spectrum, characters and fusion rules \cite{AdaTri08,NagTri11,TsuTen12}:  There are two projective irreducibles as well as $2 \brac{p-1}$ non-projective irreducibles whose projective covers have Loewy diagrams similar to those in \figref{fig:SFTripLoewy}.  These projective covers carry a non-diagonalisable action of $L_0$, so the $\TripAlg{1,p}$-models are \lcfts{}.

The representations of the singlet algebras $\SingAlg{1,p}$ were considered in \cite{AdaCla03,AdaLog07}.  Here, the results are a little less comprehensive:  There is a continuum of irreducibles whose characters are known, but fusion rules and indecomposable structures do not seem to have been settled.  Some modules with non-diagonalisable $L_0$-actions have been constructed, so the $\SingAlg{1,p}$-models are also logarithmic.  The modular properties of the singlet models will appear elsewhere \cite{CreW1p13}.

The $\brac{1,p}$ triplet algebras were generalised to $\TripAlg{q,p}$, with $p,q \in \ZZ_+$, $p>q \geqslant 2$, $\gcd \set{p,q} = 1$ and $c = 1 - 6 \brac{p-q}^2 / pq$, in \cite{FeiLog06}.  This is the central charge of the minimal model $\minmod{q}{p}$ and $\TripAlg{q,p}$ has, unlike $\TripAlg{1,p}$, a reducible but indecomposable vacuum module.  Indeed, the unique irreducible quotient of the triplet vacuum module is the minimal model vacuum module.  Perhaps because of this indecomposable structure, these triplet models have been intensively studied \cite{FeiKaz06,SemNot07,RasWEx08,GabFus09,WooFus10,AdaWal09,GabMod11,AdaExp12,TsuExt13}.  Again, there are two projective irreducibles, but now the (conjectured) projective covers of the other irreducibles can have Loewy diagrams which are significantly more complicated than the diamonds we have seen here, see \cite[App.~A.1]{GabMod11} for example.  One interesting feature here is that fusion rules involving certain irreducibles, notably the irreducible quotient of the vacuum, do not behave as expected \cite{GabFus09} (fusing with these does not define an ``exact functor'').  This means that the fusion operation does not define a ring structure on the Grothendieck \emph{group} generated by the irreducible characters.  Instead, it seems that one can quotient this group by the ideal of minimal model irreducibles and only then impose fusion.  It would be interesting to try to reconstruct this Grothendieck ring (or group) using the continuum Verlinde approach advocated here.

\newpage

\section{The Fractional Level \WZW{} Model $\AKMA{sl}{2}_{-1/2}$} \label{sec:SL2}

Our next example is based on an affine Kac-Moody symmetry.  As is well known, affine algebras give rise to rational \cfts{} when the level $k$ is a non-negative integer.  To get a logarithmic theory, we turn to non-integer levels.  In particular, we consider the Kac-Moody algebra $\AKMA{sl}{2}$ at fractional level $k=-\tfrac{1}{2}$.  This is one of the simplest of the admissible levels introduced by Kac and Wakimoto \cite{KacMod88} and, like the triplet model of \secref{sec:Trip}, the theory may be described as a subalgebra of a ghost theory.  We will start from these ghosts before turning to the relevant representations of $\AKMA{sl}{2}_{-1/2}$ and their characters.  Modular transformations are derived, leading once more to Grothendieck fusion rules which are compared to the fusion rules computed in \cite{RidFus10}.

\subsection{The $\beta \gamma$ Ghost System and its $\ZZ_2$-Orbifold $\AKMA{sl}{2}_{-1/2}$}

The $\beta \gamma$ ghost system is generated by four bosonic fields $\beta$, $\gamma$, $\ahol{\beta}$ and $\ahol{\gamma}$ whose action takes the form
\begin{equation}
S \Bigl[ \beta , \gamma ; \ahol{\beta} , \ahol{\gamma} \Bigr] = g \int \brac{\beta \apd \gamma + \ahol{\beta} \pd \ahol{\gamma}} \: \dd z \dd \ahol{z}.
\end{equation}
The equations of motion require $\beta$ and $\gamma$ to be holomorphic, while $\ahol{\beta}$ and $\ahol{\gamma}$ become antiholomorphic.  The usual symmetries under shifting by (anti)holomorphic fields lead to the \opes{}
\begin{equation}
\func{\beta}{z} \func{\beta}{w} \sim 0, \qquad \func{\beta}{z} \func{\gamma}{w} \sim \frac{1}{z-w}, \qquad \func{\gamma}{z} \func{\gamma}{w} \sim 0
\end{equation}
and their antiholomorphic analogues (which we shall mostly ignore).  The (holomorphic) energy-momentum tensor is given by
\begin{equation} \label{eq:BetaGammaT}
\func{T}{z} = \frac{1}{2} \sqbrac{\normord{\func{\beta}{z} \func{\pd \gamma}{z}} - \normord{\func{\pd \beta}{z} \func{\gamma}{z}}}
\end{equation}
and the central charge is $c=-1$.

The affine algebra $\AKMA{sl}{2}$ is recovered as the $\ZZ_2$-orbifold of the ghost theory.  We define
\begin{equation} \label{eq:BetaGammaSL2}
\func{e}{z} = \frac{1}{2} \normord{\func{\beta}{z} \func{\beta}{z}}, \qquad 
\func{h}{z} = -\normord{\func{\beta}{z} \func{\gamma}{z}}, \qquad 
\func{f}{z} = \frac{1}{2} \normord{\func{\gamma}{z} \func{\gamma}{z}}
\end{equation}
and note that these composite fields obey the \opes{}
\begin{equation}
\begin{aligned}
\func{e}{z} \func{e}{w} &\sim 0, \vphantom{\frac{-1}{\brac{z-w}^2}} \\
\func{f}{z} \func{f}{w} &\sim 0, \vphantom{\frac{1/2}{\brac{z-w}^2}}
\end{aligned}
\qquad
\begin{aligned}
\func{h}{z} \func{e}{w} &\sim \frac{2 \func{e}{w}}{z-w}, \vphantom{\frac{-1}{\brac{z-w}^2}} \\
\func{h}{z} \func{f}{w} &\sim \frac{-2 \func{f}{w}}{z-w}, \vphantom{\frac{1/2}{\brac{z-w}^2}}
\end{aligned}
\qquad
\begin{aligned}
\func{h}{z} \func{h}{w} &\sim \frac{-1}{\brac{z-w}^2}, \\
\func{e}{z} \func{f}{w} &\sim \frac{1/2}{\brac{z-w}^2} - \frac{\func{h}{w}}{z-w}.
\end{aligned}
\end{equation}
This indeed corresponds to $\AKMA{sl}{2}$ at level $k=-\tfrac{1}{2}$, but with respect to a basis $\set{e,h,f}$ of $\SLA{sl}{2}$ which differs from the standard basis in that $\comm{e}{f} = -h$ and $\killing{e}{f} = -1$ (where $\killing{\cdot}{\cdot}$ denotes the Killing form in the fundamental representation).%
\footnote{We also note that the energy-momentum tensor \eqref{eq:BetaGammaT} may be recovered from the standard Sugawara construction and \eqref{eq:BetaGammaSL2}.}  %
Whereas the standard basis defines a triangular decomposition of $\SLA{sl}{2}$ with the adjoint induced by the real form $\SLA{su}{2}$, this basis corresponds to $\SLA{sl}{2 ; \RR}$ \cite{RidSL208}.

The automorphisms of $\AKMA{sl}{2}$ which preserve the Cartan subalgebra spanned by $h_0$, $k$ and the Virasoro zero-mode $L_0$ are generated by the conjugation automorphism $\conjaut$ and the \emph{spectral flow} automorphism $\sfaut$.  The former is the lift of the non-trivial Weyl reflection of $\SLA{sl}{2}$, whereas the latter is a square root of the translation subgroup of the affine Weyl group (it corresponds to translating by a dual root rather than a coroot).  Both leave the level $k=-\tfrac{1}{2}$ invariant and the action on the generators of $\AKMA{sl}{2}$ is as follows:
\begin{equation} \label{eq:SL2Auts}
\begin{aligned}
\func{\conjaut}{e_n} &= f_n, \\ 
\func{\sfaut^{\ell}}{e_n} &= e_{n - \ell},
\end{aligned}
\qquad
\begin{aligned}
\func{\conjaut}{h_n} &= -h_n, \\
\func{\sfaut^{\ell}}{h_n} &= h_n + \tfrac{1}{2} \ell \delta_{n,0},
\end{aligned}
\qquad
\begin{aligned}
\func{\conjaut}{f_n} &= e_n, \\
\func{\sfaut^{\ell}}{f_n} &= f_{n+\ell},
\end{aligned}
\qquad
\begin{aligned}
\func{\conjaut}{L_0} &= L_0, \\
\func{\sfaut^{\ell}}{L_0} &= L_0 - \tfrac{1}{2} \ell h_0 - \tfrac{1}{8} \ell^2
\end{aligned}
\end{equation}
($\sfaut$ is referred to as spectral flow because it does not preserve conformal dimensions).  We remark that $\sfaut$ is induced from the spectral flow of the $\beta \gamma$ ghost algebra:  $\func{\sfaut^2}{\beta_n} = \beta_{n-1}$, $\func{\sfaut^2}{\gamma_n} = \gamma_{n+1}$.  The map $\sfaut$ would take the ghost algebra to a twisted sector, consistent with $\AKMA{sl}{2}_{-1/2}$ being an orbifold.

Twisting the action of $\AKMA{sl}{2}$ on a module $\mathcal{M}$ by $\conjaut$ or $\sfaut^{\ell}$, we obtain new modules denoted by $\conjmod{\mathcal{M}}$ and $\sfmod{\ell}{\mathcal{M}}$, respectively.  The first is precisely the module conjugate to $\mathcal{M}$ and we shall refer to the second as a ``spectral flow'' of $\mathcal{M}$.  Explicitly, the twisted algebra action defining $\conjmod{\mathcal{M}}$ and $\sfmod{\ell}{\mathcal{M}}$ is given by defining new states $\func{\conjaut}{\ket{v}} \in \conjmod{\mathcal{M}}$ and $\func{\sfaut^{\ell}}{\ket{v}} \in \sfmod{\ell}{\mathcal{M}}$, for each $\ket{v} \in \mathcal{M}$, and letting $\AKMA{sl}{2}$ act via
\begin{equation} \label{eq:SL2AutActions}
J \cdot \func{\conjaut}{\ket{v}} = \func{\conjaut}{\func{\conjaut^{-1}}{J} \ket{v}}, \qquad J \cdot \func{\sfaut^{\ell}}{\ket{v}} = \func{\sfaut^{\ell}}{\func{\sfaut^{-\ell}}{J} \ket{v}} \qquad \text{(\( J \in \AKMA{sl}{2} \)).}
\end{equation}

\subsection{Representation Theory of $\AKMA{sl}{2}_{-1/2}$} \label{sec:SL2Rep}

The vacuum module of the theory is taken to be the irreducible \hwm{} of $\AKMA{sl}{2}_{-1/2}$ with highest weight ($h_0$-eigenvalue) $0$.  The other \hwss{} consistent with this may be found using the \svs{} of the corresponding Verma module \cite{FeiAnn92}.  One such is $f_0 \ket{0}$ which never constrains the spectrum.  The only other \sv{} which is not a descendant of $f_0 \ket{0}$ turns out to have weight $4$ and conformal dimension $4$.  It is given explicitly by
\begin{equation}
\ket{\chi} = \brac{156 e_{-3} e_{-1} - 71 e_{-2}^2 + 44 e_{-2} h_{-1} e_{-1} - 52 h_{-2} e_{-1}^2 + 16 f_{-1} e_{-1}^3 - 4 h_{-1}^2 e_{-1}^2} \ket{0}.
\end{equation}
Because we insist upon the irreducibility of the vacuum module, the corresponding field and its modes must therefore decouple from the theory (vanish identically on physical states).  For the calculation to follow, it is convenient to consider instead the descendant $f_0^2 \ket{\chi}$ for which the corresponding field is
\begin{multline} \label{eq:SL2VacNullField}
64 \normord{e e f f} + 16 \normord{e h h f} - 136 \normord{e h \partial f} + 128 \normord{e \partial h f} - 12 \normord{e \partial^2 f} - 8 \normord{h h h h} \\
+ 200 \normord{\partial e h f} - 108 \normord{\partial e \partial f} + 8 \normord{\partial h h h} - 38 \normord{\partial h \partial h} + 156 \normord{\partial^2 e f} + 24 \normord{\partial^2 h h} - \partial^3 h.
\end{multline}
Acting with the zero-mode of this field on a \hws{} $\ket{v_{\lambda}}$ of weight $\lambda$, we arrive at a constraint on such states to be physical:
\begin{equation}
h_0 \brac{h_0 - 1} \brac{2 h_0 + 1} \brac{2 h_0 + 3} \ket{v_{\lambda}} = 0.
\end{equation}
It follows that a physical highest weight must be one of $\lambda = 0$, $1$, $-\tfrac{1}{2}$ or $-\tfrac{3}{2}$.  As this analysis also applies to \svs{}, one can conclude that the physical \hwms{} must also be irreducible.

The physical \hwms{} are therefore characterised by the $\SLA{sl}{2}$-module spanned by their ground states.  This ground state module has dimension $1$ and $2$ for $\lambda = 0$ and $1$, respectively.  We denote the corresponding irreducible $\AKMA{sl}{2}_{-1/2}$-modules by $\SLIrr{0}$ and $\SLIrr{1}$.  For $\lambda = -\tfrac{1}{2}$ and $-\tfrac{3}{2}$, the ground states span an infinite-dimensional (irreducible) module of $\SLA{sl}{2}$, a discrete series representation in fact, hence we will denote the $\AKMA{sl}{2}_{-1/2}$-irreducibles by $\SLDisc{-1/2 ; +}$ and $\SLDisc{-3/2 ; +}$, with the ``$+$'' serving to indicate that the space of ground states is generated by a \hws{} (for $\SLA{sl}{2}$).  We remark that the four irreducibles $\SLIrr{0}$, $\SLIrr{1}$, $\SLDisc{-1/2 ; +}$ and $\SLDisc{-3/2 ; +}$ exhaust the \emph{admissible} modules of Kac and Wakimoto \cite{KacMod88} for $k=-\tfrac{1}{2}$.  The conformal dimensions of their ground states are $0$, $\tfrac{1}{2}$, $-\tfrac{1}{8}$ and $-\tfrac{1}{8}$, respectively.

There are several reasons to be dissatisfied with this spectrum of physical \hwms{}.  The first is that while the characters of these admissible modules have good modular properties \cite{KacMod88b}, an application of the Verlinde formula leads to negative fusion coefficients \cite{KohFus88} (see \secref{sec:SL2Mod}).  The second, which is far more elementary, is that this spectrum is not closed under conjugation.  Indeed, while $\SLIrr{0}$ and $\SLIrr{1}$ are readily seen to be self-conjugate, the application of $\conjaut$ to $\SLDisc{-1/2 ; +}$ and $\SLDisc{-3/2 ; +}$ leads to new modules:
\begin{equation}
\conjmod{\SLDisc{-1/2 ; +}} = \SLDisc{1/2 ; -}, \qquad \conjmod{\SLDisc{-3/2 ; +}} = \SLDisc{3/2 ; -}.
\end{equation}
These new modules are labelled with ``$-$'' symbols to indicate that their ground states span lowest weight discrete series representations with lowest weights $\tfrac{1}{2}$ and $\tfrac{3}{2}$, respectively.  A third reason for dissatisfaction is that, as we shall see in \secref{sec:SL2Fus}, the spectrum of admissible \hwms{} is not closed under fusion.

To check the physicality of the \hwms{} $\SLDisc{1/2 ; -}$ and $\SLDisc{3/2 ; -}$, we should let the zero-mode of the null field \eqref{eq:SL2VacNullField} act on their ground states.  In fact, we may as well analyse this constraint for a general space of ground states, allowing the possibility of principal series $\SLA{sl}{2}$-representations as well.%
\footnote{The corresponding $\AKMA{sl}{2}$-modules are sometimes known as \emph{relaxed} \hwms{} \cite{FeiEqu98,SemEmb97}.  We will see that these relaxed modules are essential to the consistency of the theory.  For now, we only remark that relaxed modules are natural for (twisted) $\beta \gamma$ representations because the ghost zero-modes $\beta_0$ and $\gamma_0$ are bosonic, hence do not square to $0$.}  %
For this, we parametrise the $\SLA{sl}{2}$-representation comprising the ground states by $\mu \in \RR$ and label its states by $\lambda \in \RR$.  The action of the zero-modes (spanning $\SLA{sl}{2}$) on the ground states is then given by
\begin{equation}
\begin{aligned}
e_0 \ket{v_{\lambda}^{\mu}} &= \tfrac{1}{2} \brac{\lambda - \mu} \ket{v_{\lambda + 2}^{\mu}}, \\
f_0 \ket{v_{\lambda}^{\mu}} &= \tfrac{1}{2} \brac{\lambda + \mu} \ket{v_{\lambda - 2}^{\mu}},
\end{aligned}
\qquad
\begin{aligned}
h_0 \ket{v_{\lambda}^{\mu}} &= \lambda \ket{v_{\lambda}^{\mu}}, \\
L_0 \ket{v_{\lambda}^{\mu}} &= \tfrac{1}{6} \mu \brac{\mu + 2} \ket{v_{\lambda}^{\mu}}.
\end{aligned}
\end{equation}
We remark that ground states with $\lambda = \mu$ are highest weight, whereas those with $\lambda = -\mu$ are lowest weight.  It follows that if $\lambda \neq \pm \mu \bmod 2$, then the $\ket{v_{\lambda}^{\mu}}$, with $\mu$ fixed and $\lambda$ fixed modulo $2$, span an irreducible principal series representation of $\SLA{sl}{2}$.  It is not hard to verify that for a given $\lambda \neq \pm \mu \bmod 2$, the parameters $\mu$ and $-2 - \mu$ define isomorphic representations.  Thus, principal series irreducibles are distinguished by the value of $\lambda \bmod{2}$ and the eigenvalue of $L_0$ (or the quadratic Casimir).

Applying the zero-mode of the null field \eqref{eq:SL2VacNullField} to the state $\ket{v_{\lambda}^{\mu}}$ now gives
\begin{align} \label{eqnRelHWSCon}
0 &= \brac{64 f_0^2 e_0^2 + 16 f_0 h_0^2 e_0 - 192 f_0 h_0 e_0 + 180 f_0 e_0 - 8 h_0^4 - 8 h_0^3 + 10 h_0^2 + 6 h_0} \ket{v_{\lambda}^{\mu}} \notag \\
&= \brac{2 \mu + 1} \brac{2 \mu + 3} \brac{\mu \brac{\mu + 2} - 3 \lambda^2} \ket{v_{\lambda}^{\mu}}.
\end{align}
If $3 \lambda^2 = \mu \brac{\mu + 2}$, we note that each choice for $\mu$ yields at most two possibilities for $\lambda$.  The space of ground states is therefore one-dimensional, in which case $\lambda = 0$ and $\mu = 0$ or $-2$, or two-dimensional, in which case $\lambda = \pm 1$ and $\mu = 1$ or $-3$.  The conformal dimensions are $0$ and $\tfrac{1}{2}$, respectively.  If $\mu = -\tfrac{1}{2}$ or $-\tfrac{3}{2}$, then $\lambda \in \RR / 2 \ZZ$ is free, while the conformal dimension of the ground states is fixed to be $\tfrac{1}{6} \mu \brac{\mu + 2} = -\tfrac{1}{8}$.

In this way, we recover the admissible $\AKMA{sl}{2}_{-1/2}$-irreducibles $\SLIrr{0}$ and $\SLIrr{1}$.  We also deduce the physicality of \emph{any} irreducible module whose ground states have conformal dimension $-\tfrac{1}{8}$.  This includes the admissibles $\SLDisc{-1/2 ; +}$ and $\SLDisc{-3/2 ; +}$ as well as their conjugates $\SLDisc{1/2 ; -}$ and $\SLDisc{3/2 ; -}$.  However, we also obtain an uncountable family of physical relaxed \hwms{} whose ground states form a principal series $\SLA{sl}{2}$-representation of conformal dimension $-\tfrac{1}{8}$.  These modules are characterised by $\lambda \in \RR / 2 \ZZ$ and will be denoted by $\SLTyp{\lambda}$.  They are irreducible when $\lambda \neq \mu = \pm \tfrac{1}{2} \bmod{2}$, so we will refer to the $\SLTyp{\lambda}$ with this range of parameters $\lambda$ as being \emph{typical}.  The remaining $\SLTyp{\lambda}$, along with the other physical $\AKMA{sl}{2}_{-1/2}$-modules $\SLIrr{0}$, $\SLIrr{1}$, $\SLDisc{-1/2 ; +}$, $\SLDisc{-3/2 ; +}$, $\SLDisc{1/2 ; -}$ and $\SLDisc{3/2 ; -}$, are defined to be \emph{atypical}.

Finally, we remark that the atypical $\SLTyp{\lambda}$ (with $\lambda = \pm \tfrac{1}{2} \bmod{2}$) are indecomposable and their structure is not completely fixed by $\lambda$ and the conformal dimension of their ground states.  This deficiency may be overcome by affixing a label ``$+$'' or ``$-$'' to communicate whether the indecomposable has a highest or lowest weight $\SLA{sl}{2}$-state among its ground states.  In this way, we arrive at four distinct atypical indecomposables whose structures are specified by the following exact sequences:
\begin{equation} \label{ses:SL2Typ}
\begin{gathered}
\dses{\SLDisc{-1/2 ; +}}{}{\SLTyp{-1/2 ; +}}{}{\SLDisc{+3/2 ; -}}, \quad 
\dses{\SLDisc{-3/2 ; +}}{}{\SLTyp{+1/2 ; +}}{}{\SLDisc{+1/2 ; -}}, \\
\dses{\SLDisc{+1/2 ; -}}{}{\SLTyp{+1/2 ; -}}{}{\SLDisc{-3/2 ; +}}, \quad 
\dses{\SLDisc{+3/2 ; -}}{}{\SLTyp{-1/2 ; -}}{}{\SLDisc{-1/2 ; +}}.
\end{gathered}
\end{equation}

\subsection{Spectral Flow and Fusion} \label{sec:SL2Fus}

At positive-integer level $k$, the spectral flow automorphism $\sfaut$ acts on the set of integrable modules as the involution $\lambda \to k - \lambda$.  This is no longer true for fractional levels as we shall see.  For $k = -\tfrac{1}{2}$, one can check that the atypical irreducibles are related by spectral flow as follows:
\begin{equation} \label{eq:SL2AtypSF}
\sfmod{1}{\SLIrr{0}} = \SLDisc{-1/2 ; +}, \qquad \sfmod{1}{\SLIrr{1}} = \SLDisc{-3/2 ; +}, \qquad \sfmod{-1}{\SLIrr{0}} = \SLDisc{1/2 ; -}, \qquad \sfmod{-1}{\SLIrr{1}} = \SLDisc{3/2 ; -}.
\end{equation}
This is most easy checked by computing the action of $\sfaut^{\pm 1}$ on the \emph{extremal} states of a module, though one can also use the character formulae discussed in \secref{sec:SL2Mod}.  Extremal states are defined to be those whose conformal dimensions are minimal among all states sharing their weight.  For example, $e_{-1}^j \ket{0}$ and $f_{-1}^j \ket{0}$ are the extremal states of $\SLIrr{0}$ and \eqref{eq:SL2Auts} and \eqref{eq:SL2AutActions} give
\begin{equation}
\begin{split}
h_0 \func{\sfaut^{\pm 1}}{e_{-1}^j \ket{0}} &= \func{\sfaut^{\pm 1}}{\brac{h_0 \mp \tfrac{1}{2}} e_{-1}^j \ket{0}} = \brac{2j \mp \tfrac{1}{2}} \func{\sfaut^{\pm 1}}{e_{-1}^j \ket{0}}, \\
L_0 \func{\sfaut^{\pm 1}}{e_{-1}^j \ket{0}} &= \func{\sfaut^{\pm 1}}{\brac{L_0 \pm \tfrac{1}{2} h_0 - \tfrac{1}{8}} e_{-1}^j \ket{0}} = \brac{j \pm j - \tfrac{1}{8}} \func{\sfaut^{\pm 1}}{e_{-1}^j \ket{0}}.
\end{split}
\end{equation}
Thus, $\sfaut^{\pm 1}$ shifts the weight uniformly by $\mp \tfrac{1}{2}$, whereas $\sfaut$ increases the conformal dimension of $e_{-1}^j \ket{0}$ by $2j - \tfrac{1}{8}$ and $\sfaut^{-1}$ increases it by $-\tfrac{1}{8}$.  A similar calculation describes the spectral flow of the states $f_{-1}^j \ket{0}$.

More interesting is the question of what happens if we iterate the spectral flow.  For example,
\begin{equation}
L_0 \func{\sfaut^2}{f_{-1}^j \ket{0}} = -\brac{j + \tfrac{1}{2}} \func{\sfaut^2}{f_{-1}^j \ket{0}}
\end{equation}
shows that the conformal dimensions of the states of $\sfmod{2}{\SLIrr{0}}$ are \emph{unbounded below}.  The same is true for $\sfmod{\ell}{\SLIrr{0}}$ and $\sfmod{\ell}{\SLIrr{1}}$, when $\ell \neq 0, \pm 1$, and $\sfmod{\ell}{\SLTyp{\lambda}}$, when $\ell \neq 0$.  Because algebra automorphisms map physical modules to physical modules, we conclude that almost all of the physical modules of $\AKMA{sl}{2}_{-1/2}$ have the property that the conformal dimensions of their states have no lower bound.  We illustrate the weights and conformal dimensions of the states of these modules in \figref{fig:SLSpec}.

{
\psfrag{L0}[][]{$\SLIrr{0}$}
\psfrag{L1}[][]{$\SLIrr{1}$}
\psfrag{E0}[][]{$\SLTyp{\lambda}$}
\psfrag{La}[][]{$\SLDisc{-1/2 ; +}$}
\psfrag{Lb}[][]{$\SLDisc{-3/2 ; +}$}
\psfrag{La*}[][]{$\SLDisc{1/2 ; -}$}
\psfrag{Lb*}[][]{$\SLDisc{3/2 ; -}$}
\psfrag{g}[][]{$\sfaut$}
\psfrag{00}[][]{$\scriptstyle \brac{0,0}$}
\psfrag{aa}[][]{$\scriptstyle \tbrac{-\tfrac{1}{2},-\tfrac{1}{8}}$}
\psfrag{bb}[][]{$\scriptstyle \tbrac{\tfrac{1}{2},-\tfrac{1}{8}}$}
\psfrag{cc}[][]{$\scriptstyle \tbrac{-1,-\tfrac{1}{2}}$}
\psfrag{dd}[][]{$\scriptstyle \tbrac{1,-\tfrac{1}{2}}$}
\psfrag{ee}[][]{$\scriptstyle \tbrac{1,\tfrac{1}{2}}$}
\psfrag{ff}[][]{$\scriptstyle \tbrac{-1,\tfrac{1}{2}}$}
\psfrag{gg}[][]{$\scriptstyle \tbrac{-\tfrac{3}{2},-\tfrac{1}{8}}$}
\psfrag{hh}[][]{$\scriptstyle \tbrac{\tfrac{1}{2},\tfrac{7}{8}}$}
\psfrag{ii}[][]{$\scriptstyle \tbrac{\tfrac{3}{2},-\tfrac{1}{8}}$}
\psfrag{jj}[][]{$\scriptstyle \tbrac{-\tfrac{1}{2},\tfrac{7}{8}}$}
\psfrag{kk}[][]{$\scriptstyle \brac{-2,-1}$}
\psfrag{ll}[][]{$\scriptstyle \brac{0,1}$}
\psfrag{mm}[][]{$\scriptstyle \brac{2,-1}$}
\psfrag{0e}[][]{$\scriptstyle \tbrac{\lambda,-\tfrac{1}{8}}$}
\psfrag{aq}[][]{$\scriptstyle \tbrac{\lambda-\tfrac{1}{2},\tfrac{\lambda}{2}-\tfrac{1}{4}}$}
\psfrag{cq}[][]{$\scriptstyle \tbrac{\lambda-1,\lambda-\tfrac{5}{8}}$}
\psfrag{bq}[][]{$\scriptstyle \tbrac{\lambda+\tfrac{1}{2},-\tfrac{\lambda}{2}-\tfrac{1}{4}}$}
\psfrag{dq}[][]{$\scriptstyle \tbrac{\lambda+1,-\lambda-\tfrac{5}{8}}$}
\begin{figure}
\begin{center}
\includegraphics[width=\textwidth]{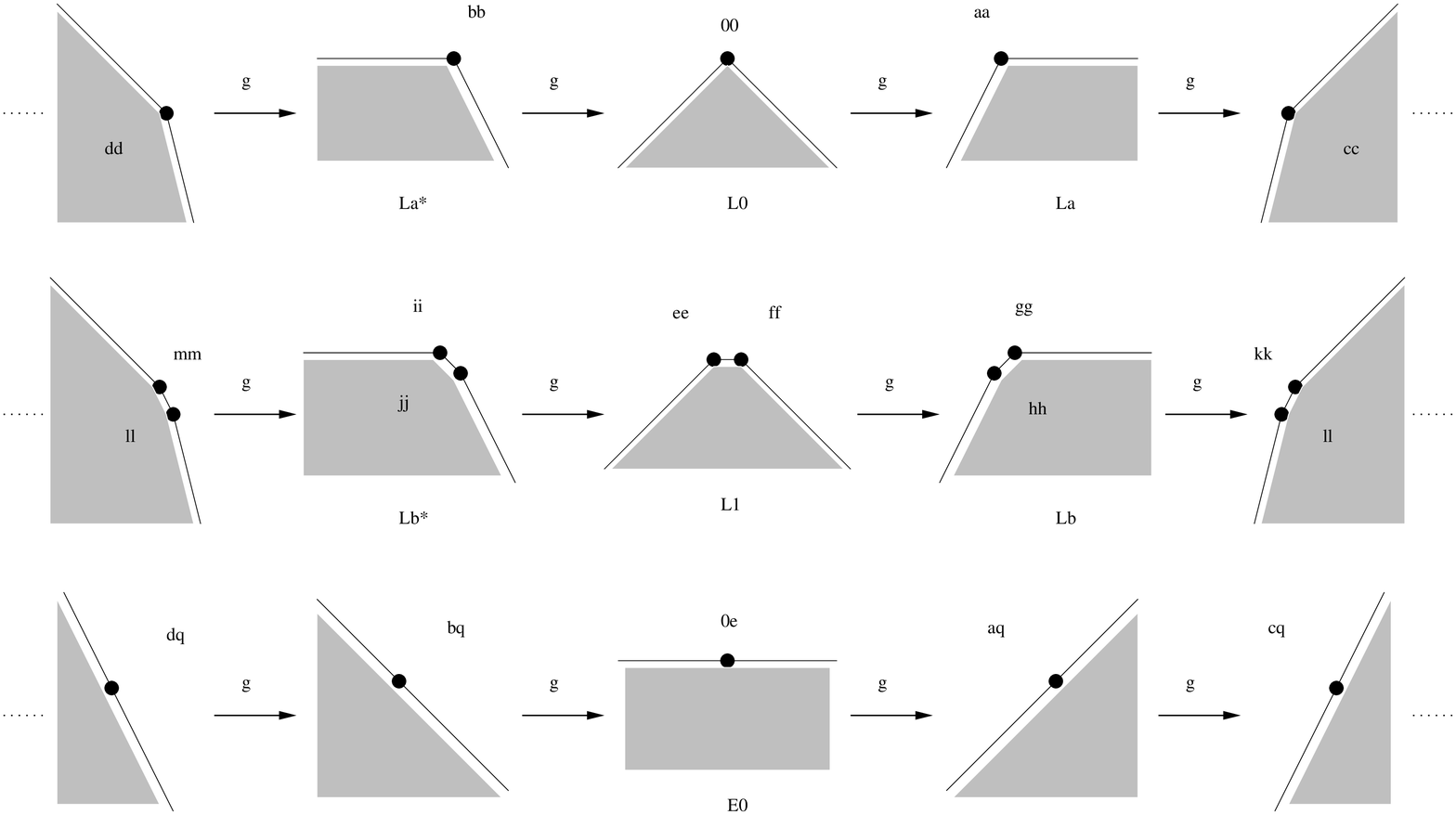}
\caption{\label{fig:SLSpec}Depictions of the physical irreducible $\AKMA{sl}{2}_{-1/2}$-modules.  Each labelled state declares its weight and conformal dimension (in that order).  Conformal dimensions increase from top to bottom and $\SLA{sl}{2}$-weights increase from right to left.}
\end{center}
\end{figure}
}

The full spectrum of irreducible modules then consists of the spectral flow images $\sfmod{\ell}{\SLIrr{0}}$, $\sfmod{\ell}{\SLIrr{1}}$ and $\sfmod{\ell}{\SLTyp{\lambda}}$, where $\ell \in \ZZ$ and $\lambda \neq \pm \tfrac{1}{2} \in \RR / 2 \ZZ$.  We extend the notion of typicality and atypicality so that it is preserved by spectral flow.  We then also have four families of atypical indecomposables $\sfmod{\ell}{\SLTyp{\pm 1/2 ; \pm}}$.  As twisting by automorphisms preserves module structure, the members of these families are described by applying $\sfaut^{\ell}$ to each module in the appropriate exact sequence of \eqref{ses:SL2Typ}.  This is a rather large collection of modules, but it is still not quite complete.

We turn now to the fusion of the irreducible $\AKMA{sl}{2}_{-1/2}$-modules identified above.  Because there are uncountably many of these, this seems a rather daunting undertaking.  However, two things work in our favour:  First, fusion and spectral flow are strongly believed to play nicely together in the sense that
\begin{equation} \label{eq:FusionAssumption}
\sfmod{\ell_1}{\mathcal{M}} \fuse \sfmod{\ell_2}{\mathcal{N}} = \sfmod{\ell_1 + \ell_2}{\brac{\mathcal{M} \fuse \mathcal{N}}}
\end{equation}
for all (physical) modules $\mathcal{M}$ and $\mathcal{N}$.  We know of no proof for this relation despite much evidence in its favour, but we will assume it in the following, so that we may restrict the fusion rules to the ``untwisted'' sector $\ell_1 = \ell_2 = 0$.  The second boon is that the \svs{} of the typical modules $\SLTyp{\lambda}$ are expressible as polynomial functions of $\lambda$.  This allows us to compute their fusion decompositions explicitly as functions of $\lambda$.

The fusion rules of the untwisted irreducibles were computed in \cite{RidFus10} using the Nahm-Gaberdiel-Kausch algorithm.  As expected, $\SLIrr{0}$ is the fusion identity and the other products are
\begin{equation} \label{FR:SL2Irr}
\begin{gathered}
\SLIrr{1} \fuse \SLIrr{1} = \SLIrr{0}, \qquad 
\SLIrr{1} \fuse \SLTyp{\lambda} = \SLTyp{\lambda + 1}, \\
\SLTyp{\lambda} \fuse \SLTyp{\mu} = 
\begin{cases}
\SLStag{\lambda + \mu} & \text{if \(\lambda + \mu = 0,1 \bmod{2}\),} \\
\sfmod{1}{\SLTyp{\lambda + \mu + 1/2}} \oplus \sfmod{-1}{\SLTyp{\lambda + \mu - 1/2}} & \text{otherwise.}
\end{cases}
\end{gathered}
\end{equation}
Here, the labels on the typicals $\SLTyp{\lambda}$ and the new atypicals $\SLStag{\lambda}$ must be taken modulo $2$.  These two additional modules $\SLStag{0}$ and $\SLStag{1}$ are indecomposables whose Loewy diagrams are given in \figref{fig:SL2Loewy}.  Their fusion rules are
\begin{equation} \label{FR:SL2Stag}
\SLIrr{1} \fuse \SLStag{\lambda} = \SLStag{\lambda + 1}, \qquad 
\begin{aligned}
\SLTyp{\lambda} \fuse \SLStag{\mu} &= \sfmod{-2}{\SLTyp{\lambda + \mu + 1}} \oplus 2 \: \SLTyp{\lambda + \mu} \oplus \sfmod{2}{\SLTyp{\lambda + \mu + 1}}, \\
\SLStag{\lambda} \fuse \SLStag{\mu} &= \sfmod{-2}{\SLStag{\lambda + \mu + 1}} \oplus 2 \: \SLStag{\lambda + \mu} \oplus \sfmod{2}{\SLStag{\lambda + \mu + 1}},
\end{aligned}
\end{equation}
showing that the fusion ring generated by the irreducibles closes upon adding $\SLStag{0}$, $\SLStag{1}$ and their spectral flows.

\begin{figure}
\begin{center}
\begin{tikzpicture}[thick,>=latex,
	nom/.style={circle,draw=black!20,fill=black!20,inner sep=1pt}
	]
\node (top0) at (0,1.5) [] {$\SLIrr{0}$};
\node (left0) at (-1.5,0) [] {$\sfmod{-2}{\SLIrr{1}}$};
\node (right0) at (1.5,0) [] {$\sfmod{2}{\SLIrr{1}}$};
\node (bot0) at (0,-1.5) [] {$\SLIrr{0}$};
\node (top1) at (5,1.5) [] {$\SLIrr{1}$};
\node (left1) at (3.5,0) [] {$\sfmod{-2}{\SLIrr{0}}$};
\node (right1) at (6.5,0) [] {$\sfmod{2}{\SLIrr{0}}$};
\node (bot1) at (5,-1.5) [] {$\SLIrr{1}$};
\node at (0,0) [nom] {$\SLStag{0}$};
\node at (5,0) [nom] {$\SLStag{1}$};
\draw [->] (top0) -- (left0);
\draw [->] (top0) -- (right0);
\draw [->] (left0) -- (bot0);
\draw [->] (right0) -- (bot0);
\draw [->] (top1) -- (left1);
\draw [->] (top1) -- (right1);
\draw [->] (left1) -- (bot1);
\draw [->] (right1) -- (bot1);
\end{tikzpicture}
\hspace{0.05\textwidth}
\begin{tikzpicture}[thick,>=latex,
	nom/.style={circle,draw=black!20,fill=black!20,inner sep=1pt}
	]
\node (top) at (0,1.5) [] {$\BGIrr{}$};
\node (left) at (-1.5,0) [] {$\sfmod{-2}{\BGIrr{}}$};
\node (right) at (1.5,0) [] {$\sfmod{2}{\BGIrr{}}$};
\node (bot) at (0,-1.5) [] {$\BGIrr{}$};
\node at (0,0) [nom] {$\BGStag{}$};
\draw [->] (top) -- (left);
\draw [->] (top) -- (right);
\draw [->] (left) -- (bot);
\draw [->] (right) -- (bot);
\end{tikzpicture}
\caption{\label{fig:SL2Loewy} Loewy diagrams for the indecomposable $\AKMA{sl}{2}_{-1/2}$-modules $\SLStag{0}$ and $\SLStag{1}$ (left) and the indecomposable $\beta \gamma$ ghost module $\BGStag{}$ (right).}
\end{center}
\end{figure}

We remark that $\SLStag{0}$ has the vacuum module $\SLIrr{0}$ as its socle.  One can check from \figref{fig:SLSpec} that the composition factors $\sfmod{\mp 2}{\SLIrr{1}}$ (which lie immediately above $\SLIrr{0}$ in the Loewy diagram of $\SLStag{0}$) have states $\ket{\theta^{\pm}}$ of weight $\pm 2$ and conformal dimension $-1$.  Indeed, the explicit construction of $\SLStag{0}$ as a fusion product allows us to choose normalisations for states so that $f_{-1} \ket{\theta^+} = e_{-1} \ket{\theta^-} = \ket{0}$.  We also note from this construction that the vacuum $\ket{0}$ is part of a rank $2$ Jordan block for $L_0$.  We denote the generalised eigenvector in this block by $\ket{\Omega}$, normalising it so that $L_0 \ket{\Omega} = \ket{0}$.  One can show that this fixes the structure of the indecomposable $\SLStag{0}$ uniquely \cite{RidFus10}.  In particular, one derives that $e_1 \ket{\Omega} = -\tfrac{1}{4} \ket{\theta^+}$ and $f_1 \ket{\Omega} = -\tfrac{1}{4} \ket{\theta^-}$.  In any case, the non-diagonalisable action of $L_0$ on $\SLStag{0}$ leads to logarithmic singularities in the two-point function of $\func{\Omega}{z}$.  $L_0$ is similarly non-diagonalisable when acting upon $\SLStag{1}$.  This confirms that $\AKMA{sl}{2}_{-1/2}$ is a \lcft{}.

Note that $\SLIrr{1}$, with its two ground states of dimension $\tfrac{1}{2}$, is observed to be a simple current of order $2$.  The corresponding simple current extension of $\AKMA{sl}{2}_{-1/2}$ is, of course, the $\beta \gamma$ ghost algebra \cite{RidSL208}.  Combining fusion orbits under $\SLIrr{1}$, we deduce that the ghost vacuum module $\BGIrr{} = \SLIrr{0} \oplus \SLIrr{1}$ has an indecomposable cover $\BGStag{} = \SLStag{0} \oplus \SLStag{1}$ upon which the ghosts act with half-integer moding (this is the untwisted sector).  This is analogous to the case of symplectic fermions analysed in \secref{sec:SF} and we give the Loewy diagram of $\BGStag{}$ in \figref{fig:SL2Loewy}.  Again, $L_0$ acts non-diagonalisably on $\BGStag{}$, hence the $\beta \gamma$ ghost theory is also logarithmic.

However, the untwisted sector for $\beta \gamma$ ghosts also contains other modules thanks to the existence of spectral flow.  In fact, it contains $\sfmod{2m}{\BGIrr{}}$, $\sfmod{2m}{\BGStag{}}$ and $\sfmod{2m+1}{\BGTyp{\lambda}}$, for $\lambda \in \RR / \ZZ$ and $m \in \ZZ$ (and $\BGTyp{\lambda} = \SLTyp{\lambda} \oplus \SLTyp{\lambda + 1}$).  Similarly, the twisted sector with integer moding contains $\sfmod{2m+1}{\BGIrr{}}$, $\sfmod{2m+1}{\BGStag{}}$ and $\sfmod{2m}{\BGTyp{\lambda}}$, again for $\lambda \in \RR / \ZZ$ and $m \in \ZZ$.  If one allows more general moding, then there is a continuum of sectors, each with a similar spectrum obtained by relaxing the constraint on $m$.

\subsection{Modular Transformations and the Verlinde Formula} \label{sec:SL2Mod}

The $\beta \gamma$ ghost theory is free, so it is easy to write the characters of its modules using a Poincar\'{e}-Birkhoff-Witt basis.  Indeed, the vacuum character is
\begin{equation}
\ch{\BGIrr{}} = \traceover{\BGIrr{}} y^k z^{h_0} z^{L_0 - c/24} = y^{-1/2} \prod_{n=1}^{\infty} \frac{1}{\brac{1 - z^{-1} q^{n-1/2}} \brac{1 - z q^{n-1/2}}} = y^{-1/2} \frac{\func{\eta}{q}}{\Jth{4}{z;q}},
\end{equation}
where we have included the $\AKMA{sl}{2}$-weight and level in preparation for deriving the corresponding affine characters.  The characters of its images under spectral flow may be obtained from
\begin{equation} \label{eqn:CharSF}
\fch{\sfmod{\ell}{\mathcal{M}}}{y;z;q} = \fch{\mathcal{M}}{y z^{\ell} q^{\ell^2 / 4} ; z q^{\ell / 2} ; q},
\end{equation}
which holds for an arbitrary module $\mathcal{M}$.  Combining this with the identifications \eqref{eq:SL2AtypSF} and the assertion that $\AKMA{sl}{2}_{-1/2}$ is the $\ZZ_2$-orbifold of the $\beta \gamma$ ghosts, we obtain the characters of the admissible $\AKMA{sl}{2}_{-1/2}$-modules:
\begin{equation}
\begin{aligned}
\ch{\SLIrr{0}} &= \frac{y^{-1/2}}{2} \sqbrac{\frac{\func{\eta}{q}}{\Jth{4}{z ; q}} + \frac{\func{\eta}{q}}{\Jth{3}{z ; q}}}, \\
\ch{\SLIrr{1}} &= \frac{y^{-1/2}}{2} \sqbrac{\frac{\func{\eta}{q}}{\Jth{4}{z ; q}} - \frac{\func{\eta}{q}}{\Jth{3}{z ; q}}},
\end{aligned}
\qquad
\begin{aligned}
\ch{\SLDisc{-1/2 ; +}} &= \frac{y^{-1/2}}{2} \sqbrac{\frac{-\ii \func{\eta}{q}}{\Jth{1}{z ; q}} + \frac{\func{\eta}{q}}{\Jth{2}{z ; q}}}, \\
\ch{\SLDisc{-3/2 ; +}} &= \frac{y^{-1/2}}{2} \sqbrac{\frac{-\ii \func{\eta}{q}}{\Jth{1}{z ; q}} - \frac{\func{\eta}{q}}{\Jth{2}{z ; q}}}.
\end{aligned}
\end{equation}
Unlike the triplet algebra characters studied in \secref{sec:Triplet}, these characters have good modular properties.

Indeed, with $y = \ee^{2 \pi \ii t}$, $z = \ee^{2 \pi \ii u}$ and $q = \ee^{2 \pi \ii \tau}$ as usual, the S-matrix with respect to the (ordered) basis $\set{\ch{\SLIrr{0}}, \ch{\SLIrr{1}}, \ch{\SLDisc{-1/2 ; +}}, \ch{\SLDisc{-3/2 ; +}}}$ is found to be symmetric and unitary:
\begin{equation}
\modS = \frac{1}{2} 
\begin{pmatrix}
1 & -1 & 1 & -1 \\
-1 & 1 & 1 & -1 \\
1 & 1 & \ii & \ii \\
-1 & -1 & \ii & \ii
\end{pmatrix}
.
\end{equation}
However, when one computes the conjugation matrix $\modC = \modS^2$ and the fusion matrices (using the standard Verlinde formula), trouble arises in the form of negative multiplicities \cite{KohFus88}:
\begin{equation}
\modC = 
\begin{pmatrix}
1 & 0 & 0 & 0 \\
0 & 1 & 0 & 0 \\
0 & 0 & 0 & -1 \\
0 & 0 & -1 & 0
\end{pmatrix}
, \quad 
\fusmat{\SLDisc{-1/2 ; +}} = 
\begin{pmatrix}
0 & 0 & 1 & 0 \\
0 & 0 & 0 & 1 \\
0 & -1 & 0 & 0 \\
-1 & 0 & 0 & 0
\end{pmatrix}
\quad \fusmat{\SLDisc{-3/2 ; +}} = 
\begin{pmatrix}
0 & 0 & 0 & 1 \\
0 & 0 & 1 & 0 \\
-1 & 0 & 0 & 0 \\
0 & -1 & 0 & 0
\end{pmatrix}
.
\end{equation}

To understand this paradox of negative multiplicities, one should first recall that the fusion ring generated by the admissibles and their conjugates contains all their spectral flows $\sfmod{\ell}{\SLIrr{0}}$ and $\sfmod{\ell}{\SLIrr{1}}$ as well.  However, the well known periodicities of theta functions lead to the following character identities:
\begin{equation} \label{eq:CharLinDep}
\ch{\sfmod{\ell - 1}{\SLIrr{0}}} + \ch{\sfmod{\ell + 1}{\SLIrr{1}}} = 
\ch{\sfmod{\ell - 1}{\SLIrr{1}}} + \ch{\sfmod{\ell + 1}{\SLIrr{0}}} = 0.
\end{equation}
The characters of the $\sfmod{\ell}{\SLIrr{\lambda}}$ are therefore linearly dependent \cite{LesSU202} and a basis for the span of these characters is precisely given by the characters of the admissibles.  The resolution to the paradox \cite{RidSL208} is then that the characters do not completely specify the module, because of \eqref{eq:CharLinDep}, hence the map from modules to characters is a projection.  For example, the conjugate of $\SLDisc{-1/2 ; +}$ is $\SLDisc{1/2 ; -}$, but we have
\begin{equation}
\ch{\SLDisc{1/2 ; -}} = \ch{\sfmod{-1}{\SLIrr{0}}} = -\ch{\sfmod{1}{\SLIrr{1}}} = -\ch{\SLDisc{-3/2 ; +}},
\end{equation}
which explains why $\modC$ asserts that the conjugate to $\SLDisc{-1/2 ; +}$ is $-\SLDisc{-3/2 ; +}$ (it is, as far as the characters are concerned).  Similarly, the fusion rule $\SLDisc{-3/2 ; +} \fuse \SLDisc{-3/2 ; +} = \sfmod{2}{\SLIrr{0}}$ translates into the Grothendieck rule
\begin{equation}
\ch{\SLDisc{-3/2 ; +}} \grfuse \ch{\SLDisc{-3/2 ; +}} = \ch{\sfmod{2}{\SLIrr{0}}} = -\ch{\SLIrr{1}},
\end{equation}
which explains why the Verlinde formula gives $\fuscoeff{\SLDisc{-3/2 ; +} \SLDisc{-3/2 ; +}}{\SLIrr{1}} = -1$.

Having recognised the source of negative multiplicities (the characters are not linearly independent), we can turn to fixing it.  As the reader may have guessed, this will involve the typical modules $\SLTyp{\lambda}$ introduced in \secref{sec:SL2Rep}.  The characters of these may be computed \cite{RidSL210} from those of the $\beta \gamma$ ghost modules $\BGTyp{\lambda}$ which are, again, easily deduced from Poincar\'{e}-Birkhoff-Witt bases:
\begin{equation}
\ch{\SLTyp{\lambda}} = \frac{y^{-1/2} z^{\lambda}}{\func{\eta}{q}^2} \sum_{n \in \ZZ} z^{2n} \qquad \Rightarrow \qquad \ch{\sfmod{\ell}{\SLTyp{\lambda}}} = \frac{y^{-1/2} z^{\lambda - \ell / 2} q^{\ell \brac{\lambda - \ell / 4} / 2}}{\func{\eta}{q}^2} \sum_{n \in \ZZ} z^{2n} q^{\ell n}.
\end{equation}
It is important to note that this character is not convergent for any $z \neq 0$, hence it does not define a meromorphic function.  Instead, we shall treat it as a distribution in the variables $t$, $u$ and $\tau$:
\begin{align} \label{ch:SL2Typ}
\ch{\sfmod{\ell}{\SLTyp{\lambda}}} 
&= \frac{\ee^{-\ii \pi t} \ee^{2 \pi \ii \brac{\lambda - \ell / 2} u} \ee^{\ii \pi \ell \brac{\lambda - \ell / 4} \tau}}{\func{\eta}{\tau}^2} \sum_{n \in \ZZ} \ee^{2 \pi \ii \brac{2u + \ell \tau} n} \notag \\
&= \frac{\ee^{-\ii \pi t} \ee^{\ii \pi \ell^2 \tau / 4}}{\func{\eta}{\tau}^2} \sum_{m \in \ZZ} \ee^{\ii \pi \brac{\lambda - \ell / 2} m} \func{\delta}{m = 2u + \ell \tau}.
\end{align}

This has a rather satisfying interpretation \cite{CreMod12}:  The character of $\sfmod{\ell}{\SLIrr{\lambda}}$ is, as a meromorphic function of $z$, only convergent in the annulus $\abs{q}^{-\ell + 1} < \abs{z} < \abs{q}^{-\ell - 1}$ (we also need $\abs{q} < 1$) because of poles at the annulus' boundaries \cite{LesSU202}.  One therefore realises that the linear dependencies \eqref{eq:CharLinDep} amount to summing characters which are defined on \emph{disjoint} annuli of convergence and obtaining $0$.  On the other hand, applying spectral flow to the exact sequences \eqref{ses:SL2Typ} gives the character relations
\begin{equation} \label{eq:CharLinIndep}
\ch{\sfmod{\ell - 1}{\SLIrr{0}}} + \ch{\sfmod{\ell + 1}{\SLIrr{1}}} = \ch{\sfmod{\ell}{\SLTyp{+1/2 ; \pm}}}, \qquad 
\ch{\sfmod{\ell - 1}{\SLIrr{1}}} + \ch{\sfmod{\ell + 1}{\SLIrr{0}}} = \ch{\sfmod{\ell}{\SLTyp{-1/2 ; \pm}}},
\end{equation}
from which we deduce that, in the distributional setting, summing these characters does not give $0$, but rather gives a distribution whose support is precisely the pole separating the two annuli of convergence.

In any case, \eqref{eq:CharLinIndep} shows that the atypical irreducible characters are not linearly dependent when treated as distributions.  To obtain distributional formulae for these characters, we splice the exact sequences \eqref{ses:SL2Typ} together to obtain resolutions for the atypical irreducibles $\sfmod{\ell}{\SLIrr{0}}$ and $\sfmod{\ell}{\SLIrr{1}}$ in terms of the atypical indecomposables $\sfmod{\ell}{\SLTyp{\pm 1/2 ; +}}$.  As with the singlet algebra $\SingAlg{1,2}$ in \secref{sec:Singlet}, these resolutions translate into alternating sums for the atypical characters:
\begin{equation} \label{ch:SL2Atyp}
\ch{\sfmod{\ell}{\SLIrr{\lambda}}} = \sum_{\ell' = 0}^{\infty} \brac{-1}^{\ell'} \ch{\sfmod{\ell + 2 \ell' + 1}{\SLTyp{\lambda + \ell' + 1/2 ; +}}} \qquad \text{(\(\lambda = 0,1\)).}
\end{equation}
Our (topological) basis for the space of characters will therefore be chosen to consist of those of the typical irreducibles $\sfmod{\ell}{\SLTyp{\lambda}}$, $\lambda \neq \pm \tfrac{1}{2} \bmod{2}$, supplemented by those of the atypical indecomposables $\sfmod{\ell}{\SLTyp{\pm 1/2 ; +}}$.

We therefore turn to S-transforming the characters \eqref{ch:SL2Typ} of the $\sfmod{\ell}{\SLTyp{\lambda}}$.  We find that
\begin{equation} \label{eqn:SL2TypS}
\Sch{\sfmod{\ell}{\SLTyp{\lambda}}} = \sum_{\ell' \in \ZZ} \int_{-1}^1 \modS_{\brac{\ell , \lambda} \brac{\ell' , \lambda'}} \ch{\sfmod{\ell'}{\SLTyp{\lambda'}}} \: \dd \lambda', \quad 
\modS_{\brac{\ell , \lambda} \brac{\ell' , \lambda'}} = \frac{1}{2} \frac{\abs{\tau}}{-\ii \tau} \ee^{\ii \pi \brac{\ell \ell' / 2 - \ell \lambda' - \ell' \lambda}},
\end{equation}
which is easily verified by expanding both sides, integrating and then summing.  Once again, the S-matrix is seen to be symmetric and unitary.  The $\tau$-dependent factor in the S-matrix entries originates from the homogeneity of delta functions, $\func{\delta}{ax=0} = \abs{a}^{-1} \func{\delta}{x=0}$, and is not particularly worrisome because it will cancel when transforming bulk characters (partition functions) and when applying the Verlinde formula.  Mathematically, it signifies that the space of characters carries a \emph{projective} representation of the modular group $\SLG{SL}{2 ; \ZZ}$ rather than a genuine one.

The modular S-transformation of the atypical irreducible characters is now computed from \eqref{ch:SL2Atyp} exactly as we did for the atypical singlet modules in \secref{sec:SingMod}.  The resulting S-matrix entries take the form
\begin{equation} \label{eqn:SL2AtypS}
\modS_{\atyp{\brac{\ell , \lambda}} \brac{\ell' , \lambda'}} = \sum_{\ell' = 0}^{\infty} \brac{-1}^{\ell'} \modS_{\brac{\ell + 2 \ell' + 1} \brac{\lambda + \ell' + 1/2}} = \frac{1}{2} \frac{\abs{\tau}}{-\ii \tau} \frac{\ee^{\ii \pi \brac{\ell \ell' / 2 - \ell \lambda' - \ell' \lambda}}}{2 \func{\cos}{\pi \lambda'}}.
\end{equation}
Here, as before, we distinguish atypical from typical labels by underlining the former (so $\atyp{\brac{\ell , \lambda}}$ stands for $\sfmod{\ell}{\SLIrr{\lambda}}$ with $\lambda = 0,1$).  Applying the continuous Verlinde formula, we can now easily obtain the Grothendieck fusion coefficients.  The calculations are straight-forward and very similar to those presented in \secref{sec:SingMod}, so we only report the resulting Grothendieck fusion rules:
\begin{equation}
\begin{gathered}
\ch{\sfmod{\ell}{\SLIrr{\lambda}}} \grfuse \ch{\sfmod{m}{\SLIrr{\mu}}} = \ch{\sfmod{\ell + m}{\SLIrr{\lambda + \mu}}}, \qquad 
\ch{\sfmod{\ell}{\SLIrr{\lambda}}} \grfuse \ch{\sfmod{m}{\SLTyp{\mu}}} = \ch{\sfmod{\ell + m}{\SLTyp{\lambda + \mu}}}, \\
\ch{\sfmod{\ell}{\SLTyp{\lambda}}} \grfuse \ch{\sfmod{m}{\SLTyp{\mu}}} = \ch{\sfmod{\ell + m+1}{\SLTyp{\lambda + \mu + 1/2}}} + \ch{\sfmod{\ell + m-1}{\SLTyp{\lambda + \mu - 1/2}}}.
\end{gathered}
\end{equation}
These agree perfectly with the Grothendieck versions of the fusion rules \eqref{FR:SL2Irr} for the irreducibles and may be checked to imply the Grothendieck versions of those \eqref{FR:SL2Stag} of the indecomposables.  We remark that these decompositions also confirm the Grothendieck version of the conjectured relation \eqref{eq:FusionAssumption}.

\subsection{Bulk Modular Invariants and State Spaces} \label{sec:SL2Bulk}

The symmetries of the S-matrix imply, as usual, that the diagonal partition function (and its charge-conjugate variant) is modular invariant:
\begin{equation} \label{eq:SL2PartFunc}
\func{\partfunc{diag.}}{q,\ahol{q}} = \sum_{\ell \in \ZZ} \int_{-1}^1 \ahol{\ch{\sfmod{\ell}{\SLTyp{\lambda}}}} \ch{\sfmod{\ell}{\SLTyp{\lambda}}} \: \dd \lambda.
\end{equation}
At the level of the quantum state space, this suggests a splitting into typical and atypical sectors:
\begin{equation}
\bulkstatespace = \Biggl[ \bigoplus_{\ell \in \ZZ} \underset{\lambda \neq \pm 1/2}{\directint_{-1}^1} \brac{\sfmod{\ell}{\SLTyp{\lambda}} \otimes \sfmod{\ell}{\SLTyp{\lambda}}} \: \dd \lambda \Biggr] \oplus \SLBulk{\text{atyp}}.
\end{equation}
Once again, the character of the atypical contribution $\SLBulk{\text{atyp}}$ may be put in the suggestive form
\begin{align}
\ch{\SLBulk{\text{atyp}}} &= \sum_{\ell \in \ZZ} \brac{\ahol{\ch{\sfmod{\ell}{\SLIrr{0}}}} \ch{\sfmod{\ell}{\SLStag{0}}} + \ahol{\ch{\sfmod{\ell}{\SLIrr{1}}}} \ch{\sfmod{\ell}{\SLStag{1}}}} \notag \\
&= \sum_{\ell \in \ZZ} \brac{\ahol{\ch{\sfmod{\ell}{\SLStag{0}}}} \ch{\sfmod{\ell}{\SLIrr{0}}} + \ahol{\ch{\sfmod{\ell}{\SLStag{1}}}} \ch{\sfmod{\ell}{\SLIrr{1}}}},
\end{align}
but, unlike the bulk triplet module of \secref{sec:TripBulk}, $\SLBulk{\text{atyp}} = \SLBulk{0} \oplus \SLBulk{1} \oplus \SLBulk{-1/2 ; +} \oplus \SLBulk{-3/2 ; +}$ is decomposable.  This follows immediately from considering the weights, modulo $2$, and conformal dimensions, modulo $1$, of the atypical states.  We have chosen to label the bulk atypicals as we did the admissible irreducibles $\SLIrr{0}$, $\SLIrr{1}$, $\SLDisc{-1/2 ; +}$ and $\SLDisc{-3/2 ; +}$ because the bulk modules are distinguished by which of the four admissibles has its character squared contributing to the bulk character.

A natural proposal for the structure of these bulk atypical modules is then to draw the Loewy diagrams of the contributing indecomposables $\sfmod{\ell}{\SLStag{\lambda}}$, tensoring each factor (on the left) with $\sfmod{\ell}{\SLIrr{\lambda}}$.  This defines the holomorphic structure of the bulk atypical and the antiholomorphic structure is added by identifying factors which combine to give $\sfmod{\ell}{\SLStag{\lambda}} \otimes \sfmod{\ell}{\SLIrr{\lambda}}$.  The resulting bulk Loewy diagram differs from that of the triplet model, pictured in \figref{fig:TripBulkMod}, in that there are now infinitely many composition factors ($\AKMA{sl}{2}_{-1/2}$ behaves more like the singlet model in this respect).  We illustrate a part of it in \figref{fig:SL2BulkLoewy}.  Note that the proposed structure suggests that spectral flow acts periodically on the atypical sector:
\begin{equation}
\cdots \overset{\sfaut}{\lra} \SLBulk{0} \overset{\sfaut}{\lra} \SLBulk{-1/2 ; +} \overset{\sfaut}{\lra} \SLBulk{1} \overset{\sfaut}{\lra} \SLBulk{-3/2 ; +} \overset{\sfaut}{\lra} \SLBulk{0} \overset{\sfaut}{\lra} \cdots.
\end{equation}
We further remark that these bulk modules are manifestly local ($L_0 - \ahol{L}_0$ may be diagonalised).

\begin{figure}
\begin{center}
\begin{tikzpicture}[thick,>=latex,scale=1.4]
\node (tr) at (12,4) [] {$\scriptstyle \sfmod{\ell - 2}{\SLIrr{\lambda - 1}} \otimes \sfmod{\ell - 2}{\SLIrr{\lambda - 1}}$};
\node (tm) at (8,4) [] {$\scriptstyle \sfmod{\ell}{\SLIrr{\lambda}} \otimes \sfmod{\ell}{\SLIrr{\lambda}}$};
\node (tl) at (4,4) [] {$\scriptstyle \sfmod{\ell + 2}{\SLIrr{\lambda + 1}} \otimes \sfmod{\ell + 2}{\SLIrr{\lambda + 1}}$};
\node (br) at (12,2) [] {$\scriptstyle \sfmod{\ell - 2}{\SLIrr{\lambda - 1}} \otimes \sfmod{\ell - 2}{\SLIrr{\lambda - 1}}$};
\node (bm) at (8,2) [] {$\scriptstyle \sfmod{\ell}{\SLIrr{\lambda}} \otimes \sfmod{\ell}{\SLIrr{\lambda}}$};
\node (bl) at (4,2) [] {$\scriptstyle \sfmod{\ell + 2}{\SLIrr{\lambda + 1}} \otimes \sfmod{\ell + 2}{\SLIrr{\lambda + 1}}$};
\node at (13,4) [] {$\cdots$};
\node (mrr) at (13,3) [] {$\cdots$};
\node at (13,2) [] {$\cdots$};
\node (mrl) at (11,3) [] {$\scriptstyle \sfmod{\ell - 2}{\SLIrr{\lambda - 1}} \otimes \sfmod{\ell}{\SLIrr{\lambda}}$};
\node (mmr) at (9,3) [] {$\scriptstyle \sfmod{\ell}{\SLIrr{\lambda}} \otimes \sfmod{\ell - 2}{\SLIrr{\lambda - 1}}$};
\node (mml) at (7,3) [] {$\scriptstyle \sfmod{\ell}{\SLIrr{\lambda}} \otimes \sfmod{\ell + 2}{\SLIrr{\lambda + 1}}$};
\node (mlr) at (5,3) [] {$\scriptstyle \sfmod{\ell + 2}{\SLIrr{\lambda + 1}} \otimes \sfmod{\ell}{\SLIrr{\lambda}}$};
\node at (3,4) [] {$\cdots$};
\node (mll) at (3,3) [] {$\cdots$};
\node at (3,2) [] {$\cdots$};
\draw [->] (tr) -- (mrr);
\draw [->] (mrr) -- (br);
\draw [->] (tr) -- (mrl);
\draw [->] (mrl) -- (br);
\draw [->] (tm) -- (mmr);
\draw [->] (mmr) -- (bm);
\draw [->] (tm) -- (mml);
\draw [->] (mml) -- (bm);
\draw [->] (tl) -- (mlr);
\draw [->] (mlr) -- (bl);
\draw [->] (tl) -- (mll);
\draw [->] (mll) -- (bl);
\draw [->,dotted] (tr) -- (mmr);
\draw [->,dotted] (mmr) -- (br);
\draw [->,dotted] (tm) -- (mrl);
\draw [->,dotted] (mrl) -- (bm);
\draw [->,dotted] (tm) -- (mlr);
\draw [->,dotted] (mlr) -- (bm);
\draw [->,dotted] (tl) -- (mml);
\draw [->,dotted] (mml) -- (bl);
\end{tikzpicture}
\caption{\label{fig:SL2BulkLoewy}A part of the (proposed) Loewy diagrams for the atypical bulk modules.  Taking $\lambda = 0 \bmod{2}$ and $\ell = 0, 1, 2, 3 \bmod{4}$ gives the diagrams for $\SLBulk{0}$, $\SLBulk{-1/2 ; +}$, $\SLBulk{1}$ and $\SLBulk{-3/2 ; +}$, respectively.}
\end{center}
\end{figure}
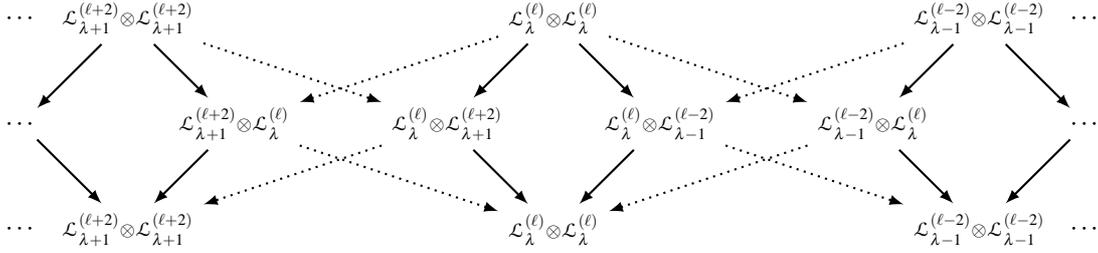

There are, of course, many simple currents with which we can try to construct additional modular invariants.  Indeed, the fusion rules \eqref{FR:SL2Irr}, coupled with \eqref{eq:FusionAssumption}, show that every atypical irreducible $\sfmod{\ell}{\SLIrr{\lambda}}$ is a simple current.  Requiring that the simple currents have fields of integer dimension restricts this to two families, $\sfmod{4m}{\SLIrr{0}}$ and $\sfmod{4m+2}{\SLIrr{1}}$, for $m \in \ZZ$.  However, the latter family is generated by the spectral flows of the bosonic ghost fields.  Because the ghost fields have dimension $\tfrac{1}{2}$ and their spectral flows have integer dimension, these flowed fields will be \emph{mutually fermionic}.  We do not therefore expect them to give rise to modular invariants.

The construction of the extended algebras proceeds as with the singlet algebra (\secref{sec:SingModInv}).  For each $n \in \ZZ$, there is an extended algebra $\ExtAlg{n} = \bigoplus_{m \in \ZZ} \sfmod{4mn}{\SLIrr{0}}$ and its untwisted typical modules are the $\ExtSLTyp{n ; \ell}{\lambda} = \bigoplus_{m \in \ZZ} \sfmod{4mn + \ell}{\SLTyp{\lambda}}$ with $2n \lambda \in \ZZ$ and $\ell = 0, 1, \ldots, 4m-1$.  The extended S-matrix is then
\begin{equation}
\Sch{\ExtSLTyp{n ; \ell}{j/2n}} = \sum_{j', \ell' = 0}^{4n-1} \ExtmodS{n}_{\brac{j,\ell} \brac{j',\ell'}} \ch{\ExtSLTyp{n ; \ell'}{j'/2n}}, \qquad \ExtmodS{n}_{\brac{j,\ell} \brac{j',\ell'}} = \frac{1}{4n} \frac{\abs{\tau}}{-\ii \tau} \ee^{\ii \pi \brac{\ell \ell' n - j' \ell - j \ell'} / 2n},
\end{equation}
and its evident symmetries immediately imply the modular invariance of the extended partition function
\begin{equation}
\partfunc{diag.}^{\brac{n}} = \sum_{j=0}^{4n-1} \sum_{\ell = 0}^{4n-1} \ahol{\ch{\ExtSLTyp{n ; \ell}{j/2n}}} \ch{\ExtSLTyp{n ; \ell}{j/2n}}.
\end{equation}
Unfortunately, the modular properties of the atypical extended characters again remain out of reach.

We conclude by mentioning that these extended partition functions may describe the level $-\tfrac{1}{2}$ \WZW{} models on the non-compact simple Lie group $\SLG{SL}{2 ; \RR}$.  More precisely, we recall that this group has centre $\ZZ_2$ and fundamental group $\ZZ$, so there are an infinite number of Lie groups having $\SLA{sl}{2 ; \RR}$ as their Lie algebra.  This includes the simply-connected universal cover of $\SLG{SL}{2 ; \RR}$, often referred to as $\grp{AdS}_3$, and the centreless adjoint group $\SLG{PSL}{2 ; \RR}$.  We propose that the partition function $\partfunc{diag.}$ of \eqref{eq:SL2PartFunc} describes strings on $\grp{AdS}_3$ with level $-\tfrac{1}{2}$, whereas $\partfunc{diag.}^{\brac{1}}$ describes strings on $\SLG{PSL}{2 ; \RR}$ (or $\SLG{SL}{2 ; \RR}$).  Whether this proposal is true or not, we remark that one can arrive at a consistent structure for the atypical sector of the quantum state space of these extended theories by identifying factors in \figref{fig:SL2BulkLoewy} which are identical except that their spectral flow indices differ by a multiple of $4n$.  This imposes a periodicity on the infinitely wide bulk Loewy diagrams so that the resulting diagrams have a finite number of composition factors.  For $n=1$, it is easy to check that each extended atypical diagram has eight factors and that the structure looks identical to the triplet bulk atypical pictured in \figref{fig:TripBulkMod}.

\subsection{Correlation Functions} \label{sec:corrsl2}

Correlation functions for $\AKMA{sl}{2}_{-1/2}$ may be computed using a strategy that is almost identical to that used for symplectic fermions.  We start by recalling the bosonisation of the ghost fields $\tfunc{\beta}{z}$ and $\tfunc{\gamma}{z}$.  Let $\tfunc{\varphi_L}{z}$ and $\tfunc{y_L}{z}$ be free bosons with \opes{}
\begin{equation}
\func{\varphi_L}{z} \func{\varphi_L}{w} \sim \log \brac{z-w}, \qquad 
\func{y_L}{z} \func{y_L}{w} \sim -\log \brac{z-w}.
\end{equation}
Then, the bosonisation amounts to the identifications
\begin{equation}
\func{\beta}{z} = \normord{\ee^{-\brac{\func{\varphi_L}{z} + \func{y_L}{z}}}}, \qquad 
\func{\gamma}{z} = \normord{\pd \func{\varphi_L}{z} \ee^{\func{\varphi_L}{z} + \func{y_L}{z}}}.
\end{equation}
We observe that these fields commute with the zero-mode of $\ee^{-\varphi_L(z)}$, so we conclude that we can use the same screening charge \eqref{eq:DefScreen} as for symplectic fermions.  The $\AKMA{sl}{2}_{-1/2}$ currents now take the bosonised form
\begin{equation}
\begin{gathered}
\func{e}{z} = \frac{1}{2} \normord{\func{\beta}{z} \func{\beta}{z}} = \frac{1}{2} \normord{\ee^{-2 \bigl( \func{\varphi_L}{z} + \func{y_L}{z} \bigr)}}, \qquad 
\func{h}{z} = \frac{1}{2} \normord{\func{\beta}{z} \func{\gamma}{z}} = \pd \func{y_L}{z}, \\
\func{f}{z} = \frac{1}{2} \normord{\func{\gamma}{z} \func{\gamma}{z}} = \frac{1}{2} \normord{\brac{\pd \func{\varphi}{z} \pd \func{\varphi}{z} - \pd^2 \func{\varphi}{z}} \ee^{2 \bigl( \func{\varphi_L}{z} + \func{y_L}{z} \bigr)}}.
\end{gathered}
\end{equation}
We will neglect the antiholomorphic currents which are constructed similarly.  In this free field realisation, it is natural to consider the following bulk fields:
\begin{equation}
\func{V_{\ell, \lambda; n}}{z,\ahol{z}} = \normord{\ee^{\brac{-\lambda+1/2} \func{\varphi}{z,\ahol{z}} + \brac{-\lambda + \ell/2} \func{y}{z,\ahol{z}} - n \tbrac{\func{\varphi_L}{z} + \func{y_L}{z}}}}.
\end{equation}
Here we have chosen a convenient labelling which requires $\ell \in \ZZ$ and $n \in 2 \ZZ$.  We compute the following \opes{} with the currents:
\begin{equation}
\begin{gathered}
\func{e}{z} \func{V_{\ell,\lambda;0}}{w,\ahol{w}} = \frac{\frac{1}{2}\func{V_{\ell,\lambda;2}}{w,\ahol{w}}}{\brac{z-w}^{1-\ell}} + \cdots, \qquad 
\func{h}{z} \func{V_{\ell,\lambda;0}}{w,\ahol{w}} = \frac{\brac{\lambda - \ell/2} \func{V_{\ell,\lambda;0}}{w,\ahol{w}}}{z-w} + \cdots, \\
\func{f}{z} \func{V_{\ell,\lambda;0}}{w,\ahol{w}} = \frac{\frac{1}{2} \brac{-\lambda + \tfrac{1}{2}} \brac{-\lambda + \tfrac{3}{2}} \func{V_{\ell,\lambda;-2}}{w,\ahol{w}}}{\brac{z-w}^{1+\ell}} + \cdots, \\
\func{T}{z} \func{V_{\ell,\lambda;0}}{w,\ahol{w}} = \frac{-\tfrac{1}{8} \brac{1 + \ell^2 - 4 \ell \lambda} \func{V_{\ell,\lambda;0}}{w,\ahol{w}}}{\brac{z-w}^2} + \frac{\pd \func{V_{\ell,\lambda;0}}{w,\ahol{w}}}{z-w} + \cdots.
\end{gathered}
\end{equation}
This means that the field $\tfunc{V_{\ell,\lambda;0}}{w,\ahol{w}}$ transforms as a ground state in the representation $\sfmod{\ell}{\SLTyp{\lambda}}$.  In general, $n$ parametrises the ground states of these representations, but we will restrict to $n=0$ for simplicity.

Correlation functions are defined exactly as in the symplectic fermion case, see \eqref{eq:DefCorr}, with correlators in the free theory again being free boson correlators, subject to a charge conservation condition:
\begin{multline}
\corrfn{\func{\nonch V_{\ell_1,\lambda_1;0}}{z_1,\ahol{z}_1} \cdots \func{\nonch V_{\ell_n,\lambda_n;0}}{z_n,\ahol{z}_n}}_0 \\
= -\delta_{\lambda_1 + \cdots + \lambda_n = \tfrac{1}{2} \brac{n-2}} \delta_{\ell_1 + \cdots + \ell_n = n-2} \prod_{i<j} \abs{z_i - z_j}^{2 \tbrac{ \brac{\lambda_i - 1/2} \brac{\lambda_j - 1/2} - \brac{\lambda_i - \ell_i / 2} \brac{\lambda_j - \ell_j / 2}}}.
\end{multline}
In this way, we arrive at the following non-zero one- and two-point functions:
\begin{equation}
\begin{split}
\corrfn{\func{\nonch V_{\ell,\lambda;0}}{z,\ahol{z}}} &= -\delta_{\lambda=-1/2} \delta_{\ell = -1}, \\
\corrfn{\func{\nonch V_{\ell,\lambda;0}}{z,\ahol{z}} \func{\nonch V_{\ell',\lambda';0}}{w,\ahol{w}}} &= -\delta_{\lambda + \lambda' = 0} \delta_{\ell + \ell' = 0} \abs{z-w}^{2 \brac{\brac{\lambda - \ell / 2}^2 - \lambda^2 + 1/4}},
\end{split}
\end{equation}
We remark that this one-point function is consistent with the vacuum appearing as a composition factor of the indecomposables $\sfmod{-1}{\SLTyp{-1/2 ; \pm}}$ and that the two-point function agrees with our notion of conjugation.  For the three-point functions, we again use the Fateev-Dotsenko integral formula to obtain
\begin{multline}
\corrfn{\func{\nonch V_{\ell,\lambda;0}}{z,\ahol{z}} \func{\nonch V_{m,\mu;0}}{1,1} \func{\nonch V_{n,\nu;0}}{0,0}} = \frac{-1}{\abs{z-1}^{2 \brac{h_{\ell,\lambda} + h_{m,\mu} - h_{n,\nu}}} \abs{z}^{2 \brac{h_{\ell,\lambda} - h_{m,\mu} + h_{n,\nu}}}} \\
\cdot \sqbrac{\delta_{\lambda + \mu + \nu = 1/2} \delta_{\ell + m + n = 1} + \frac{\func{\Gamma}{\tfrac{1}{2} + \lambda} \func{\Gamma}{\tfrac{1}{2} + \mu}\func{\Gamma}{\tfrac{1}{2} + \nu}}{\func{\Gamma}{\tfrac{1}{2} - \lambda} \func{\Gamma}{\tfrac{1}{2} - \mu} \func{\Gamma}{\tfrac{1}{2} - \nu}} \delta_{\lambda + \mu + \nu = -1/2} \delta_{\ell + m + n = -1}},
\end{multline}
where $h_{\ell, \lambda} = \tfrac{1}{2} \tbrac{\brac{\lambda - \tfrac{1}{2}} \brac{\lambda + \tfrac{1}{2}} - \brac{\lambda - \tfrac{1}{2} \ell}^2} = \tfrac{1}{2} \brac{\ell \lambda - \tfrac{1}{4} \brac{\ell^2 - 1}}$ is the conformal dimension of the (twisted) ground state of weight $\lambda - \ell/2$ in $\sfmod{\ell}{\SLTyp{\lambda}}$.  This time, there are singularities in the three-point function whenever $\lambda \in \ZZ + \tfrac{1}{2}$.  Regularising as in the case of symplectic fermions would give logarithmic correlators, thereby confirming the presence of the indecomposable modules $\sfmod{\ell}{\SLStag{0}}$ and $\sfmod{\ell}{\SLStag{1}}$.  These correlation functions may be checked to be consistent with the fusion rules \eqref{FR:SL2Irr}.

\subsection{Further Developments} \label{sec:SL2Future}

Unlike the triplet theories discussed in \secref{sec:TripFuture} and the superalgebra theories that we will consider in \secref{sec:GL11Future}, the \lcfts{} with admissible level affine algebra symmetries remain relatively unexplored.  Aside from $\AKMA{sl}{2}_{-1/2}$, reviewed above, the only other admissible theory to have received comparable treatment is $\AKMA{sl}{2}_{-4/3}$ \cite{GabFus01,CreMod12}.  Here, there are three admissible \hwms{}, all irreducible, which we denote (with the same conventions as used above) by $\SLIrr{0}$, $\SLDisc{-2/3 ; +}$ and $\SLDisc{-4/3 ; +}$.  Again, spectral flow acts and we obtain two infinite families because $\SLDisc{-4/3 ; +} = \sfmod{1}{\SLIrr{0}}$.  There are also typical irreducibles $\SLTyp{\lambda}$, with $\lambda \neq \pm \tfrac{2}{3} \bmod{2}$, and atypical indecomposables $\SLTyp{\pm 2/3 ; \pm}$, both of whose ground states have conformal dimension $-\tfrac{1}{3}$, as well as their spectral flows.

The Nahm-Gaberdiel-Kausch algorithm gives (untwisted) fusion rules including \cite{GabFus01}
\begin{equation} \label{FR:SL2Gab}
\SLDisc{2/3 ; -} \fuse \SLDisc{-2/3 ; +} = \SLIrr{0} \oplus \SLTyp{0}, \qquad 
\SLDisc{-2/3 ; +} \fuse \SLTyp{0} = \SLStag{-2/3 ; +}, \qquad 
\SLTyp{0} \fuse \SLTyp{0} = \SLTyp{0} \oplus \SLStag{0},
\end{equation}
with the vacuum module $\SLIrr{0}$ acting again as the fusion identity.  Here, $\SLStag{0}$ and $\SLStag{-2/3 ; +}$ are indecomposables with respective socles $\SLIrr{0}$ and $\SLDisc{-2/3 ; +}$ and the familiar diamond-shaped Loewy diagrams (see \cite{CreMod12}).  Both exhibit a non-diagonalisable action of $L_0$.  The modular properties of the characters of the typicals and atypicals are derived as for $k=-\tfrac{1}{2}$ with the resulting S-matrix entries being \cite{CreMod12}
\begin{equation}
\begin{gathered}
\modS_{\brac{\ell,\lambda} \brac{\ell',\lambda'}} = \frac{1}{2} \frac{\abs{\tau}}{-\ii \tau} \ee^{\ii \pi \brac{4 \ell \ell' / 3 - \ell \lambda' - \ell' \lambda}}, \\
\modS_{\atyp{\brac{\ell,0}} \brac{\ell',\lambda'}} = \frac{1}{2} \frac{\abs{\tau}}{-\ii \tau} \frac{\ee^{\ii \pi \ell \brac{4 \ell' / 3 - \lambda'}}}{1 + 2 \func{\cos}{\pi \lambda'}}, \qquad
\modS_{\atyp{\brac{\ell,-2/3}} \brac{\ell',\lambda'}} = \frac{\abs{\tau}}{-\ii \tau} \frac{\ee^{\ii \pi \brac{\ell + 1/2} \brac{4 \ell' / 3 - \lambda'}} \func{\cos}{\pi \lambda' / 2}}{1 + 2 \func{\cos}{\pi \lambda'}}.
\end{gathered}
\end{equation}
Applying the Verlinde formula then leads to
\begin{equation}
\SLTyp{\lambda} \fuse \SLTyp{\mu} = \sfmod{-1}{\SLTyp{\lambda + \mu - 4/3}} \oplus \SLTyp{\lambda + \mu} \oplus \sfmod{1}{\SLTyp{\lambda + \mu + 4/3}} \qquad \text{(\(\lambda + \mu \neq 0 , \pm \tfrac{2}{3} \bmod{2}\)).}
\end{equation}
At the atypical points $0$ and $\pm \tfrac{2}{3}$, the natural prediction is instead that
\begin{equation}
\SLTyp{\lambda} \fuse \SLTyp{-\lambda} = \SLTyp{0} \oplus \SLStag{0}, \qquad 
\SLTyp{\lambda} \fuse \SLTyp{-\lambda + 2/3} = \sfmod{1}{\SLTyp{0}} \oplus \SLStag{2/3 ; -}, \qquad
\SLTyp{\lambda} \fuse \SLTyp{-\lambda} = \sfmod{-1}{\SLTyp{0}} \oplus \SLStag{-2/3 ; +},
\end{equation}
where $\SLStag{2/3 ; -} = \sfmod{-1}{\SLStag{-2/3 ; +}} = \conjmod{\SLStag{-2/3 ; +}}$.  Moreover, it also suggests the fusion rule
\begin{equation}
\SLDisc{-2/3 ; +} \fuse \SLStag{0} = \sfmod{-1}{\SLTyp{0}} \oplus \SLStag{-2/3 ; +} \oplus \sfmod{2}{\SLTyp{0}},
\end{equation}
from which the remaining rules follow by using associativity.

The fact that Verlinde formulae for $\AKMA{sl}{2}_{-1/2}$ and $\AKMA{sl}{2}_{-4/3}$ have been successfully derived using the above formalism suggests that this will generalise to all admissible levels.  This will be detailed in \cite{CreMod13}.  Other affine algebras at admissible levels have not yet received much attention, though $\AKMA{sl}{3}$ was briefly addressed \cite{FurFus98}, in the days before spectral flow and indecomposability were realised to be critical, and some structure theory for admissible level $\AKMSA{sl}{2}{1}$ may be found in \cite{BowRep97,BowCha98,SemHig05}.  The link with the \WZW{} model on $\SLG{SL}{2 ; \RR}$ is interesting because it suggests that these theories may be logarithmic for general values of the level.  The famous articles \cite{MalStr01,MalStr01b,MalStr02} of Maldacena and Ooguri suggest no logarithmic structure for this model, though it could be argued that they did not look for any (see also \cite{BarMod11}).  Indeed, recent mathematical work \cite{FjeDua11} suggests that there may be more to this picture than was previously realised.

\section{\WZW{} Models with $\AKMSA{gl}{1}{1}$ Symmetry} \label{sec:GL11}

The \WZW{} model on $\SLSG{GL}{1}{1}$ is by far the best understood \cft{} associated to Lie supergroups. It was first studied by Rozansky and Saleur two decades ago \cite{Rozansky:1992rx,Rozansky:1992td} in two of the original landmark \lcft{} papers. More recently, this model was reconsidered by Saleur and Schomerus \cite{Schomerus:2005bf} who were able to compute correlation functions and propose a structure for the full bulk theory. Their computations revealed a striking similarity to the twist field correlators of the symplectic fermion theory, an observation that was explained in \cite{Creutzig:2008an}, which in turn was motivated by \cite{LeClair:2007aj}. The correlation functions of Saleur and Schomerus implicitly suggested fusion rules, which were then confirmed in \cite{Creutzig:2011cu}. This theory is also one of the few for which the boundary theory is thoroughly studied \cite{Creutzig:2007jy,Creutzig:2008ek,Creutzig:2009zz}, meaning that D-branes have been classified and boundary three-point correlation functions and bulk-boundary two-point functions are known. We follow \cite{Creutzig:2011cu} in reviewing this example.

\subsection{$\SLSA{gl}{1}{1}$ and its Representations} \label{sec:GL11Fin}

The Lie superalgebra $\SLSA{gl}{1}{1}$ is generated by two bosonic elements $N$ and $E$ and by two fermionic ones $\psi^\pm$ subject to the relations
\begin{equation} \label{eqngl11Rels}
\comm{N}{\psi^{\pm}} = \pm \psi^{\pm}, \qquad \acomm{\psi^+}{\psi^-} = E.
\end{equation}
It naturally acts on the super vector space $\CC^{1 \vert 1}$ and its elements are identified with supermatrices as follows:
\begin{equation} \label{eq:GL11DefRep}
N = \frac{1}{2} 
\begin{pmatrix}
1 & 0 \\
0 & -1
\end{pmatrix}
, \qquad E = 
\begin{pmatrix}
1 & 0 \\
0 & 1
\end{pmatrix}
, \qquad \psi^+ = 
\begin{pmatrix}
0 & 1 \\
0 & 0
\end{pmatrix}
, \qquad \psi^- = 
\begin{pmatrix}
0 & 0 \\
1 & 0
\end{pmatrix}
.
\end{equation}
The Killing form, corresponding to the supertrace form in the adjoint representation, is degenerate for $\SLSA{gl}{1}{1}$.  To obtain a non-degenerate bilinear form $\killing{\cdot}{\cdot}$, one instead takes the supertrace in the defining representation \eqref{eq:GL11DefRep}:
\begin{equation}
\killing{N}{E} = \killing{E}{N} = 1, \qquad \killing{\psi^+}{\psi^-} = -\killing{\psi^-}{\psi^+} = 1.
\end{equation}

Verma modules $\FinGLVer{n,e}$ are constructed from \hwss{} $\ket{\finite{v}_{n,e}}$ satisfying
\begin{equation}
E \ket{\finite{v}_{n,e}} = e \ket{\finite{v}_{n,e}}, \qquad N \ket{\finite{v}_{n,e}} = \brac{n+\tfrac{1}{2}} \ket{\finite{v}_{n,e}}, \qquad \psi^+ \ket{\finite{v}_{n,e}} = 0
\end{equation}
in the usual manner. However, as $\psi^-$ squares to $0$, all Verma modules are two-dimensional.%
\footnote{We remark at this point that the label $n$ of the Verma module $\FinGLVer{n,e}$ refers to the \emph{average} of the $N$-eigenvalues for this representation.  This average labelling convention for $n$ will be adhered to for all $\SLSA{gl}{1}{1}$-modules.}  %
Moreover, as $\psi^+ \psi^- \ket{\finite{v}_{n,e}} = e \ket{\finite{v}_{n,e}}$, the Verma module $\FinGLVer{n,e}$ is reducible if and only if $e=0$. Once again, irreducibility is the generic situation and hence modules with $e \neq 0$ are referred to as being typical. In the atypical case where $e=0$, there is an irreducible one-dimensional submodule spanned by $\psi^- \ket{\finite{v}_{n,0}}$ --- we shall denote it by $\FinGLAtyp{n-1/2}$ --- and the quotient $\FinGLAtyp{n+1/2}$ is also one-dimensional and irreducible: 
\begin{equation}
\dses{\FinGLAtyp{n-1/2}'}{}{\FinGLVer{n,0}}{}{\FinGLAtyp{n+1/2}}.
\end{equation}
Here, the prime attached to the submodule $\FinGLAtyp{n-1/2}$ serves to remind us that its \hws{} has the opposite parity to the highest weight generators of the other modules appearing in this sequence.

Typical irreducibles turn out to be projective in the category of all finite-dimensional modules upon which $N$ and $E$ act diagonalisably.  The atypical irreducibles have projective covers $\FinGLStag{n}$ which are generated by a vector $\ket{\finite{w}_n}$ satisfying $E \ket{\finite{w}_n} = 0$ and $N \ket{\finite{w}_n} = n \ket{\finite{w}_n}$. The fermionic generators $\psi^+$ and $\psi^-$ act freely on $\ket{\finite{w}_n}$, resulting in the four-dimensional representation illustrated in \figref{fig:fingl11}.  These atypical projectives naturally appear in the representation ring generated by the irreducibles under the (graded) tensor product:
\begin{equation} \label{eq:GL11RepRing}
\begin{gathered}
\FinGLAtyp{n} \otimes \FinGLAtyp{n'} = \FinGLAtyp{n+n'}, \qquad
\FinGLAtyp{n} \otimes \FinGLVer{n',e'} = \FinGLVer{n+n',e'}, \qquad
\FinGLAtyp{n} \otimes \FinGLStag{n'} = \FinGLStag{n+n'}, \\
\FinGLVer{n,e} \otimes \FinGLVer{n',e'} = 
\begin{cases}
\FinGLStag{n+n'}' & \text{if \(e+e'=0\),} \\
\FinGLVer{n+n'+1/2,e+e'} \oplus \FinGLVer{n+n'-1/2,e+e'}' & \text{otherwise,}
\end{cases}
\\
\FinGLVer{n,e} \otimes \FinGLStag{n'} = \FinGLVer{n+n'+1,e}' \oplus 2 \: \FinGLVer{n+n',e} \oplus \FinGLVer{n+n'-1,e}', \qquad
\FinGLStag{n} \otimes \FinGLStag{n'} = \FinGLStag{n+n'+1}' \oplus 2 \: \FinGLStag{n+n'} \oplus \FinGLStag{n+n'-1}'.
\end{gathered}
\end{equation}
The prime on the indecomposables $\FinGLStag{n}$ refers to the relative parity of the generating state $\ket{\finite{w}_n}$.  We remark that the Casimir $Q = NE + \psi^- \psi^+$ acts non-diagonalisably on $\FinGLStag{n}$:  $Q \ket{\finite{w}_n} = \psi^- \psi^+ \ket{\finite{w}_n}$.

\begin{figure}
\begin{center}
\begin{tikzpicture}[thick,>=latex,
	nom/.style={circle,draw=black!20,fill=black!20,inner sep=1pt}
	]
\node (top0) at (0,1.5) [] {$\ket{\finite{w}_n}$};
\node (left0) at (-1.5,0) [] {$\psi^+ \ket{\finite{w}_n}$};
\node (right0) at (1.5,0) [] {$\psi^- \ket{\finite{w}_n}$};
\node (bot0) at (0,-1.5) [] {$\psi^- \psi^+ \ket{\finite{w}_n}$};
\node (top1) at (6,1.5) [] {$\FinGLAtyp{n}$};
\node (left1) at (4.5,0) [] {$\FinGLAtyp{n+1}$};
\node (right1) at (7.5,0) [] {$\FinGLAtyp{n-1}$};
\node (bot1) at (6,-1.5) [] {$\FinGLAtyp{n}$};
\node at (0,0) [nom] {$\FinGLStag{n}$};
\node at (6,0) [nom] {$\FinGLStag{n}$};
\draw [->] (top0) to (left0);
\draw [->] (top0) to (right0);
\draw [->] (left0) to (bot0);
\draw [->] (right0) to (bot0);
\draw [->] (top1) -- (left1);
\draw [->] (top1) -- (right1);
\draw [->] (left1) -- (bot1);
\draw [->] (right1) -- (bot1);
\end{tikzpicture}
\caption{\label{fig:fingl11}The projective cover $\FinGLStag{n}$ of the $\SLSA{gl}{1}{1}$-module $\FinGLAtyp{n}$, illustrated by its states on the left and by its Loewy diagram on the right.}
\end{center}
\end{figure}
 
\subsection{The $\SLSG{GL}{1}{1}$ \WZW{} Model}

\WZW{} models on compact reductive Lie groups give rise to a natural family of rational \cfts{}.  Models on Lie supergroups may be constructed in the same manner. One starts with a supergroup-valued field $g$ and parametrises it using a ``Gauss-like'' decomposition. For $\SLSG{GL}{1}{1}$, this corresponds to
\begin{equation}
g = \ee^{c_- \psi^-} \ee^{XE + YN} \ee^{-c_+ \psi^+},
\end{equation}
so the fields of the theory are the two bosonic fields $\func{X}{z,\ahol{z}}$ and $\func{Y}{z,\ahol{z}}$ and the two fermionic fields $\func{c_{\pm}}{z,\ahol{z}}$. The standard \WZW{} action is then reduced, using the Polyakov-Wiegmann identity, to
\begin{equation} \label{eq:SWZW}
S_{\text{WZW}}[g] = \frac{k}{4 \pi} \int \brac{-\pd X \apd Y - \pd Y \apd X + 2 \ee^Y \pd c_+ \apd c_- } \: \dd z \dd \ahol{z},
\end{equation}
where $k$ is the level. Varying this action leads to the expected equations of motion:
\begin{equation}
\apd \func{J}{z,\ahol{z}} = 0, \qquad \pd \func{\ahol{J}}{z,\ahol{z}} = 0.
\end{equation}
Here, $J = k \pd g g^{-1}$ and $\ahol{J} = -k g^{-1} \apd g$ are Lie superalgebra-valued currents. In component form, they become
\begin{equation} \label{eq:holomorphiccurrents}
J^E = -k \pd Y, \qquad 
J^N = -k \pd X + k c_- \pd c_+ \ee^Y, \qquad 
J^- = k \ee^Y \pd c_+, \qquad 
J^+ = -k \pd c_- - k c_- \pd Y,
\end{equation}
and similarly for the anti-holomorphic current $\ahol{J}$.

Upon quantising, the modes of the holomorphic current satisfy the relations of the affine Kac-Moody superalgebra $\AKMSA{gl}{1}{1}$ at level $k$ (we will mostly ignore the antiholomorphic sector as usual): 
\begin{equation} \label{eq:affine}
\comm{J^E_r}{J^N_s} = r \delta_{r+s=0} k, \qquad 
\comm{J^N_r}{J^{\pm}_s} = \pm J^{\pm}_{r+s},\qquad 
\acomm{J^+_r}{J^-_s} = J^E_{r+s} + r \delta_{r+s=0} k.
\end{equation}
As in the case of the free boson (\secref{sec:FreeBoson}), any non-zero level $k$ can be rescaled to $1$.  We will assume in what follows that such a rescaling has been made.

The energy-momentum tensor $\func{T}{z}$ has $c=0$ and is given by a variant \cite{Rozansky:1992rx} of the Sugawara construction:%
\footnote{The Lie superalgebra $\SLSA{gl}{1}{1}$ is not simple and the space of invariant bilinear forms turns out to be two-dimensional. However, there is a unique choice that leads to a Virasoro field.}
\begin{equation} \label{eqnDefT}
\func{T}{z} = \frac{1}{2} \func{\normord{J^N J^E + J^E J^N - J^+ J^- + J^- J^+}}{z} + \frac{1}{2} \func{\normord{J^E J^E}}{z}.
\end{equation}
We note that, as in the case of $\AKMA{sl}{2}$ discussed in \secref{sec:SL2}, $\AKMSA{gl}{1}{1}$ possesses a family of spectral flow automorphisms $\sfaut^{\ell}$ which are indispensable to understanding its representation theory:
\begin{equation} \label{eq:GL11SpecFlow}
\func{\sfaut^{\ell}}{J^N_r} = J^N_r, \qquad 
\func{\sfaut^{\ell}}{J^E_r} = J^E_r - \ell  \delta_{r,0}, \qquad 
\func{\sfaut^{\ell}}{J^{\pm}_r} = J^{\pm}_{r \mp \ell}, \qquad 
\func{\sfaut^{\ell}}{L_0} = L_0 - \ell J^N_0.
\end{equation}
As before, these automorphisms may be used to construct new modules $\sfmod{\ell}{\mathcal{M}}$ from an arbitrary $\AKMSA{gl}{1}{1}$-module $\mathcal{M}$ by twisting the action on the states as in \eqref{eq:SL2AutActions}.

\subsection{Representation Theory of $\AKMSA{gl}{1}{1}$}

The representation theory of the affine algebra $\AKMSA{gl}{1}{1}$ is very similar to that of its horizontal subalgebra $\SLSA{gl}{1}{1}$.  We define affine Verma modules $\GLVer{n,e}$ and their irreducible quotients, as before, by defining a \hws{} $\ket{v_{n,e}}$ to be one satisfying%
\footnote{As with $\SLSA{gl}{1}{1}$-modules, the label $n$ parametrising modules refers to the average $J^N_0$-eigenvalue of the ground states.}
\begin{equation}
J^N_0 \ket{v_{n,e}} = \brac{n + \tfrac{1}{2}} \ket{v_{n,e}}, \qquad 
J^E_0 \ket{v_{n,e}} = e \ket{v_{n,e}}, \qquad 
J^+_0 \ket{v_{n,e}} = J^{\pm}_r \ket{v_{n,e}} = 0 \quad \text{(\(r>0\)).}
\end{equation} 
The conformal dimension of the \hws{} $\ket{v_{n,e}}$ is then (recalling that $k$ has been set to $1$)
\begin{equation} \label{eqnConfDim}
L_0 \ket{v_{n,e}} = \Delta_{n,e} \ket{v_{n,e}} = \tbrac{ne + \tfrac{1}{2} e^2} \ket{v_{n,e}}.
\end{equation}

It follows that every \sv{} of $\GLVer{n,0}$ has dimension $0$, leading to the non-split exact sequence
\begin{subequations} \label{ses:GL11Atyp}
\begin{equation} \label{ses:GL11Atyp1}
\dses{\GLAtyp{n-1/2,0}'}{}{\GLVer{n,0}}{}{\GLAtyp{n+1/2,0}},
\end{equation}
where the prime again indicates that the submodule's \hws{} has parity opposite to those of the other modules.  A simple counting argument \cite{Creutzig:2011cu} now shows that the $\GLVer{n,e}$ with $0 < \abs{e} < 1$ are irreducible.  By employing the spectral flow automorphisms of \eqref{eq:GL11SpecFlow}, one concludes that the Verma modules $\GLVer{n,e}$ with $e \notin \ZZ$ are irreducible (typical) and that the (atypical) case $e \in \ZZ$ yields reducible Verma modules.  Along with \eqref{ses:GL11Atyp1}, the (non-split) exact sequences turn out to be
\begin{equation}
\begin{aligned}
&\dses{\GLAtyp{n-1,e}'}{}{\GLVer{n,e}}{}{\GLAtyp{n,e}} & &\text{(\(e \in \ZZ_-\)),} \\
&\dses{\GLAtyp{n+1,e}'}{}{\GLVer{n,e}}{}{\GLAtyp{n,e}} & &\text{(\(e \in \ZZ_+\)).}
\end{aligned}
\end{equation}
\end{subequations}
Note that, once again, the vacuum module $\GLAtyp{0,0}$ is atypical.

We remark that, in contrast with the action of spectral flow on $\AKMA{sl}{2}$-modules, spectral flows of $\AKMSA{gl}{1}{1}$-modules do not have states whose conformal dimensions are unbounded below.  This can be traced back to the fact that each of the modes $J^{\pm}_r$ squares to $0$.%
\footnote{The corresponding statement for integrable $\affine{\alg{g}}$-modules, with $\affine{\alg{g}}$ a Kac-Moody algebra, may similarly be traced back to the special form of the \svs{} of the Verma covers of these modules.}  %
Indeed, character methods and some analysis of Verma modules and their contragredient duals allow one to identify the result of applying $\sfaut$:
\begin{equation}
\sfmod{1}{\GLVer{n,e}} = \GLVer{n-1,e+1}', \quad \text{if \(e \notin \ZZ\),} \qquad 
\sfmod{1}{\GLAtyp{n,e}} = 
\begin{cases}
\GLAtyp{n-1/2,0}', & \text{if \(e=-1\),} \\
\GLAtyp{n-1/2,1}, & \text{if \(e=0\),} \\
\GLAtyp{n-1,e+1}', & \text{otherwise.}
\end{cases}
\end{equation}
Note that this immediately explains how the irreducibility of the Verma modules with $e \notin \ZZ$ could be deduced from that of those with $0 < \abs{e} < 1$.

Finally, the fusion rules follow readily from the principles behind the Nahm-Gaberdiel-Kausch algorithm and the tensor product rules \eqref{eq:GL11RepRing} of $\SLSA{gl}{1}{1}$ (we ignore parity to avoid an overabundance of cases):
\begin{equation} \label{FR:GL11}
\begin{gathered}
\GLAtyp{n,e} \fuse \GLAtyp{n',e'} = \GLAtyp{n+n'-\func{\eps}{e,e'},e+e'}, \qquad 
\GLAtyp{n,e} \fuse \GLVer{n',e'} = \GLVer{n+n'-\func{\eps}{e},e+e'}, \\
\GLVer{n,e} \fuse \GLVer{n',e'} = 
\begin{cases}
\GLStag{n+n'+\func{\eps}{e+e'},e+e'}, & \text{if \(e+e' \in \ZZ\),} \\
\GLVer{n+n'+1/2,e+e'} \oplus \GLVer{n+n'-1/2,e+e'} & \text{otherwise,}
\end{cases}
\\
\begin{aligned}
\GLAtyp{n,e} \fuse \GLStag{n',e'} &= \GLStag{n+n'-\func{\eps}{e,e'},e+e'}, \\
\GLVer{n,e} \fuse \GLStag{n',e'} &= \GLVer{n+n'+1-\func{\eps}{e'},e+e'} \oplus 2 \: \GLVer{n+n'-\func{\eps}{e'},e+e'} \oplus \GLVer{n+n'-1-\func{\eps}{e'},e+e'}, \\
\GLStag{n,e} \fuse \GLStag{n',e'} &= \GLStag{n+n'+1-\func{\eps}{e,e'},e+e'} \oplus 2 \: \GLStag{n+n'-\func{\eps}{e,e'},e+e'} \oplus \GLStag{n+n'-1-\func{\eps}{e,e'},e+e'}.
\end{aligned}
\end{gathered}
\end{equation} 
Here, we have defined
\begin{equation}
\tfunc{\eps}{e} = \tfrac{1}{2} \tfunc{\sgn}{e}, \qquad 
\tfunc{\eps}{e,e'} = \tfunc{\eps}{e} + \tfunc{\eps}{e'} - \tfunc{\eps}{e+e'},
\end{equation}
with the convention that the sign function satisfies $\tfunc{\sgn}{0} = 0$.

The modules $\GLStag{n,e}$, with $e \in \ZZ$, generated by the above fusion rules may be constructed by inducing the $\SLSA{gl}{1}{1}$-module $\FinGLStag{n}$ and applying spectral flow.  They are indecomposable and, just as $\FinGLStag{n}$ carries a non-diagonalisable action of the Casimir $Q$, so the action of $L_0$ on $\GLStag{n,e}$ is non-diagonalisable.  The Loewy diagrams are given in \figref{fig:GL11AffLoewy}.  It is often stated that the $\GLStag{n,e}$ are the projective covers of the $\GLAtyp{n,e}$, presumably in the category of vertex algebra modules upon which $J^N_0$ and $J^E_0$ act diagonalisably, but we are not aware of a proof of this statement.

\begin{figure}
\begin{center}
\begin{tikzpicture}[thick,>=latex,
	nom/.style={circle,draw=black!20,fill=black!20,inner sep=1pt}
	]
\node (top1) at (6,1.5) [] {$\GLAtyp{n,e}$};
\node (left1) at (4.5,0) [] {$\GLAtyp{n+1,e}'$};
\node (right1) at (7.5,0) [] {$\GLAtyp{n-1,e}'$};
\node (bot1) at (6,-1.5) [] {$\GLAtyp{n,e}$};
\node at (6,0) [nom] {$\GLStag{n,e}$};
\draw [->] (top1) -- (left1);
\draw [->] (top1) -- (right1);
\draw [->] (left1) -- (bot1);
\draw [->] (right1) -- (bot1);
\end{tikzpicture}
\caption{\label{fig:GL11AffLoewy}The Loewy diagram of the indecomposable $\AKMSA{gl}{1}{1}$-module $\GLStag{n,e}$.}
\end{center}
\end{figure}

\subsection{Modular Transformations and the Verlinde Formula} \label{sec:GL11Mod}

For superalgebras, it is natural to work with supercharacters rather than characters.  Those of the Verma modules $\GLVer{n,e}$ are given by
\begin{align} \label{eq:GL11CharVerma}
\fsch{\GLVer{n,e}}{x;y;z;q} &= \straceover{\GLVer{n,\ell}} x y^{J^E_0} z^{J^N_0} q^{L_0 - c/24} = x y^e z^{n+1/2} q^{\Delta_{n,e}} \prod_{i=1}^{\infty} \frac{\brac{1 - z q^i} \brac{1 - z^{-1} q^{i-1}}}{\brac{1 - q^i}^2} \notag \\
&= \ii x y^e z^n q^{\Delta_{n,e}} \frac{\Jth{1}{z;q}}{\func{\eta}{q}^3},
\end{align}
where we recall that the level $k$ (of which $x$ is supposed to keep track) has been set to $1$.  With $x = \ee^{2 \pi \ii t}$, $y = \ee^{2 \pi \ii u}$, $z = \ee^{2 \pi \ii v}$ and $q = \ee^{2 \pi \ii \tau}$ and the S-transformation $\modS \colon \left( t \middle\vert u \middle\vert v \middle\vert \tau \right) \to \left( t - uv / \tau \middle\vert u / \tau \middle\vert v / \tau \middle\vert -1 / \tau \right)$, the induced transformation of the supercharacters may be computed by a double Gaussian integral. The resulting S-matrix entries are 
\begin{equation}
\Ssch{\GLVer{n,e}} = \int_{-\infty}^{\infty} \int_{-\infty}^{\infty} \modS_{(n,e),(n',e')} \sch{\GLVer{n',e'}} \: \dd n' \dd e', \qquad 
\modS_{(n,e),(n',e')} = -\ii \omega \ee^{-2 \pi \ii \brac{ne'+n'e+ee'}},
\end{equation}
where $\omega$ is a square root of minus one that depends upon how the analytic continuation is performed.  We note that the sign $-\ii \omega$ will cancel in bulk modular invariants and the Verlinde formula.

The exact sequences \eqref{ses:GL11Atyp} may now be used to construct resolutions for the atypicals in terms of Verma modules.  We consider only the vacuum module for brevity:
\begin{equation} \label{eq:GL11VacRes}
\cdots \lra \GLVer{-7/2,0}' \lra \GLVer{-5/2,0} \lra \GLVer{-3/2,0}' \lra \GLVer{-1/2,0} \lra \GLAtyp{0,0} \lra 0.
\end{equation}
Because the relative parities of the Verma modules alternate, the vacuum supercharacter is not an alternating sum of Verma characters.  The vacuum S-matrix entries are therefore
\begin{align}
\modS_{\atyp{\brac{0,0}}, \brac{n',e'}} &= \sum_{j=0}^{\infty} \modS_{\brac{-j-1/2, 0}, \brac{n',e'}} = -\ii \omega \sum_{j=0}^{\infty} \ee^{2 \pi \ii \brac{j+1/2} e'} = \frac{\ii \omega}{\ee^{\ii \pi e'} - \ee^{-\ii \pi e'}} = \frac{\omega}{2 \sin \sqbrac{\pi e'}}.
\end{align}
From this, and the continuous Verlinde formula, we can compute the Grothendieck fusion of the typicals:
\begin{align}
\fuscoeff{\brac{n,e}, \brac{n',e'}}{\brac{n'',e''}} &= -\int_{-\infty}^{\infty} \int_{-\infty}^{\infty} \ee^{-2 \pi \ii \bigl( (n+n'-n'') E + (e+e'-e'') (N+E) \bigr)} \brac{\ee^{\ii \pi E} - \ee^{\ii \pi E}} \: \dd N \dd E \notag \\
&= \func{\delta}{n''=n+n'+1/2} - \func{\delta}{n''=n+n'-1/2} \\
\Rightarrow \qquad \sch{\GLVer{n,e}} \grfuse &\sch{\GLVer{n',e'}} = \sch{\GLVer{n+n'+1/2,e+e'}} - \sch{\GLVer{n+n'-1/2,e+e'}}.
\end{align}
The minus sign appearing in this result indicates that the parity of $\GLVer{n+n'-1/2,e+e'}$ is opposite to that of $\GLVer{n+n'+1/2,e+e'}$, hence we would affix a prime to the former module.  This is therefore in perfect agreement with the genuine typical fusion rule reported in \eqref{FR:GL11}, the case when $e+e' \in \ZZ$ following from the character identity
\begin{equation}
\sch{\GLVer{n+n'+1/2,e+e'}} + \sch{\GLVer{n+n'-1/2,e+e'}'} = \sch{\GLStag{n+n'+\func{\eps}{e+e'},\ell+\ell'}},
\end{equation}
as may be checked using the exact sequences \eqref{ses:GL11Atyp} and the Loewy diagram in \figref{fig:GL11AffLoewy}.  The other Grothendieck fusion products may be checked in a similar fashion.  We remark that one can also use the modular properties of the characters, rather than the supercharacters, to compute Grothendieck fusion rules.  However, closure under the (projective) modular group action then requires the consideration of the characters and supercharacters of twisted $\AKMSA{gl}{1}{1}$-modules on which the generators act with half-integer moding.

\subsection{Bulk Modular Invariants and State Spaces}
 
The bulk state space of the $\SLSG{GL}{1}{1}$ \WZW{} model was first proposed by Saleur and Schomerus in \cite{Schomerus:2005bf}.  The result is
\begin{equation}
\bulkstatespace = \Biggl[ \: \underset{e \notin \ZZ}{\directint_{\RR^2}} \Bigl[ \brac{\GLVer{n,e} \otimes \GLVer{-n,-e}'} \Bigr] \: \dd n \dd e \Biggr] \oplus \GLBulk{\text{atyp}} ,
\end{equation}
where the atypical contributions further decompose as
\begin{equation}
\GLBulk{\text{atyp}} = \bigoplus_{e \in \ZZ} \directint_0^1 \GLBulkatyp{n,e} \: \dd n.
\end{equation}
We illustrate the structure of the indecomposable bulk atypicals $\GLBulkatyp{n,e}$ in \figref{fig:bulkGL}, noting that this proposal corresponds to the charge-conjugate partition function, rather than the diagonal one.  The modular invariance of the corresponding bulk super partition function now follows from the unitarity of the S-matrix and the symmetry $\modS_{\brac{n,e} \brac{n',e'}} = \modS_{\brac{-n,-e} \brac{-n',-e'}}$.
The diagonal super partition function is similarly invariant.  We remark that constructing an ordinary (non-super) modular invariant partition function requires introducing half-integer moded sectors \cite{Creutzig:2011cu} and the result is an invariant for the bosonic subtheory (an orbifold) of $\SLSG{GL}{1}{1}$.

\begin{figure}
\begin{center}
\begin{tikzpicture}[thick,>=latex,scale=1.4]
\node (tr) at (12,4) [] {$\scriptstyle\GLAtyp{n-1,\ell}  \otimes \GLAtyp{-n+1,-\ell}$};
\node (tm) at (8,4) [] {$\scriptstyle \GLAtyp{n,\ell}\otimes \GLAtyp{-n,-\ell}$};
\node (tl) at (4,4) [] {$\scriptstyle \GLAtyp{n+1,\ell} \otimes \GLAtyp{-n-1,-\ell}$};
\node (br) at (12,2) [] {$\scriptstyle \GLAtyp{n-1,\ell}  \otimes \GLAtyp{-n+1,-\ell}$};
\node (bm) at (8,2) [] {$\scriptstyle  \GLAtyp{n,\ell}\otimes \GLAtyp{-n,-\ell}$};
\node (bl) at (4,2) [] {$\scriptstyle \GLAtyp{n+1,\ell} \otimes \GLAtyp{-n-1,-\ell}$};
\node at (13,4) [] {$\cdots$};
\node (mrr) at (13,3) [] {$\cdots$};
\node at (13,2) [] {$\cdots$};
\node (mrl) at (11,3) [] {$\scriptstyle \GLAtyp{n,\ell} \otimes \GLAtyp{-n+1,-\ell}$};
\node (mmr) at (9,3) [] {$\scriptstyle \GLAtyp{n-1,\ell} \otimes \GLAtyp{-n,-\ell}$};
\node (mml) at (7,3) [] {$\scriptstyle \GLAtyp{n+1,\ell} \otimes \GLAtyp{-n,-\ell}$};
\node (mlr) at (5,3) [] {$\scriptstyle \GLAtyp{n,\ell} \otimes \GLAtyp{-n-1,-\ell}$};
\node at (3,4) [] {$\cdots$};
\node (mll) at (3,3) [] {$\cdots$};
\node at (3,2) [] {$\cdots$};
\draw [->] (tr) -- (mrr);
\draw [->] (mrr) -- (br);
\draw [->] (tr) -- (mrl);
\draw [->] (mrl) -- (br);
\draw [->] (tm) -- (mmr);
\draw [->] (mmr) -- (bm);
\draw [->] (tm) -- (mml);
\draw [->] (mml) -- (bm);
\draw [->] (tl) -- (mlr);
\draw [->] (mlr) -- (bl);
\draw [->] (tl) -- (mll);
\draw [->] (mll) -- (bl);
\draw [->,dotted] (tr) -- (mmr);
\draw [->,dotted] (mmr) -- (br);
\draw [->,dotted] (tm) -- (mrl);
\draw [->,dotted] (mrl) -- (bm);
\draw [->,dotted] (tm) -- (mlr);
\draw [->,dotted] (mlr) -- (bm);
\draw [->,dotted] (tl) -- (mml);
\draw [->,dotted] (mml) -- (bl);
\end{tikzpicture}
\caption{\label{fig:bulkGL}A part of the Loewy diagrams for the atypical bulk modules $\GLBulkatyp{n,e}$ corresponding to the charge-conjugate modular invariant partition function.  We have neglected to indicate relative parities because the conjugate of an atypical $\GLAtyp{n,e}$ changes parity unless $e=0$.}
\end{center}
\end{figure}
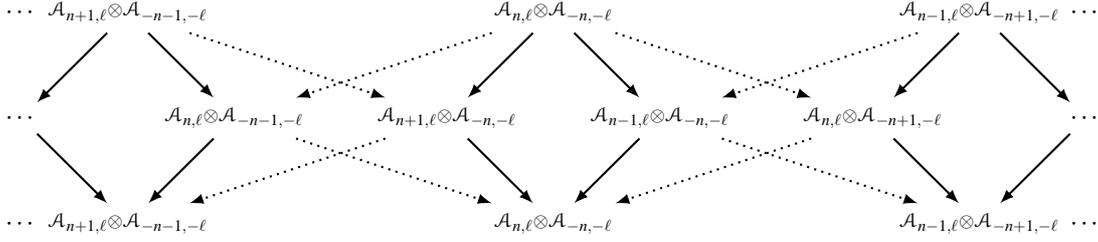

Other super partition functions which are modular invariant can be found using extended algebras. Every atypical irreducible $\GLAtyp{n,e}$ is a simple current, by \eqref{FR:GL11}.  One can therefore construct a large variety of extended algebras, among the most interesting being \cite{Creutzig:2011cu,Creutzig:2011np} the tensor product of the $\beta\gamma$ ghost algebra and that of a pair of free fermions, the affine Kac-Moody superalgebra $\AKMSA{sl}{2}{1}$ at levels $-\tfrac{1}{2}$ and $1$, as well as an infinite series of superalgebras containing, as subalgebras, the $N=2$ superconformal algebra, the Bershadsky-Polyakov algebra $\mathcal{W}^{(2)}_3$, and its generalisations, the Feigin-Semikhatov algebras $\mathcal{W}^{(2)}_n$.

One very interesting observation is that the supercharacters of the atypical irreducibles of these extended algebras turn out to be \emph{mock modular forms}.  These are familiar, but mysterious, objects in number theory whose modular transformations may be expressed in terms of an integral, the Mordell integral \cite{Zwegers}.  One can evaluate the extended algebras' Verlinde formulae directly by using this integral \cite{Alfes:2012pa}.  It would therefore be extremely interesting to rederive this using the method of atypical resolutions.  We note that mock modular forms in general seem to be closely tied to atypical characters of affine Lie superalgebras, see \cite{Kac:1994kn,Kac:2001,SemHig05}.

\subsection{Correlation Functions}

The three-point functions of the $\SLSG{GL}{1}{1}$ \WZW{} model were computed in \cite{Schomerus:2005bf}. The results bear a striking resemblance to those of symplectic fermions, an observation which was explained in \cite{Creutzig:2008an}. We briefly summarise the computations, generalising the method used for both symplectic fermions and $\AKMA{sl}{2}_{-1/2}$.  First, consider a set of three free bosons, $\tfunc{\nonch\varphi}{z,\bar z}, \tfunc{\nonch y}{z,\bar z}$ and $\tfunc{\nonch x}{z,\bar z}$, whose non-regular \opes{} take the form
\begin{equation}
\func{\nonch\varphi}{z, \ahol{z}} \func{\nonch\varphi}{w, \ahol{w}} = \log \abs{z-w}^2 + \cdots, \qquad 
\func{\nonch y}{z,\ahol{z}} \func{\nonch x}{w, \ahol{w}} = \log \abs{z-w}^2 + \cdots.
\end{equation}
As before, we denote the chiral part of fields by a subscript $L$.  Define four holomorphic fields
\begin{equation}
\func{J^E}{z} = -\func{\pd y}{z}, \quad 
\func{J^N}{z} = -\func{\pd x}{z}, \quad 
\func{J^-}{z} = -\normord{\ee^{\func{y_L}{z} - \func{\varphi_L}{z}}}, \quad 
\func{J^+}{z} = -\normord{\pd \func{\varphi}{z} \ee^{-\func{y_L}{z} + \func{\varphi_L}{z}}},
\end{equation}
and their antiholomorphic analogues similarly.  The non-regular \opes{} of these fields are
\begin{equation}
\func{J^E}{z} \func{J^N}{w} \sim \frac{1}{(z-w)^2}, \quad 
\func{J^N}{z} \func{J^\pm}{w} \sim \frac{\pm \func{J^\pm}{w}}{z-w}, \quad 
\func{J^+}{z} \func{J^-}{w} \sim \frac{1}{(z-w)^2} + \frac{\pm\func{J^E}{w}}{z-w},
\end{equation}
so we have a free field realisation of the currents of $\AKMSA{gl}{1}{1}$.  As these fields again commute with the zero-mode of $e^{-\varphi_L(z)}$, we can use the screening charge \eqref{eq:DefScreen} once again.

In this free field realisation, the interesting bulk fields are
\begin{equation}
\begin{aligned}
V^{--}_{-e,-n+1/2} &= \normord{\ee^{e \brac{\varphi + x} + n y}}, \\
V^{-+}_{-e,-n+1/2} &= \normord{\ee^{e \brac{\varphi + x} + n y + \varphi_R - y_R}},
\end{aligned}
\qquad
\begin{aligned}
V^{+-}_{-e,-n+1/2} &= \normord{\ee^{e \brac{\varphi + x} + n y + \varphi_L - y_L}}, \\
V^{++}_{-e,-n+1/2} &= \normord{\ee^{e \brac{\varphi + x} + n y + \varphi - y}}.\\
\end{aligned}
\end{equation}
We find the following non-regular \opes{} with the currents:
\begin{equation}
\begin{aligned}
\func{J^E}{z} \func{\nonch V^{- \pm}_{-e,-n+1/2}}{w,\ahol{w}} &\sim 
\frac{-e \func{\nonch V^{- \pm}_{-e,-n+1/2}}{w,\ahol{w}}}{z-w}, \\
\func{J^N}{z} \func{\nonch V^{- \pm}_{-e,-n+1/2}}{w,\ahol{w}} &\sim 
\frac{-n \func{\nonch V^{- \pm}_{-e,-n+1/2}}{w,\bar w}}{z-w}, \\
\func{J^+}{z} \func{\nonch V^{- \pm}_{-e,-n+1/2}}{w,\ahol{w}} &\sim 
\frac{e \func{\nonch V^{+ \pm}_{-e,-n+1/2}}{w,\ahol{w}}}{z-w},
\end{aligned}
\quad
\begin{aligned}
\func{J^E}{z} \func{\nonch V^{+ \pm}_{-e,-n+1/2}}{w,\ahol{w}} &\sim 
\frac{-e \func{\nonch V^{+ \pm}_{-e,-n+1/2}}{w,\ahol{w}}}{z-w}, \\
\func{J^N}{z} \func{\nonch V^{+ \pm}_{-e,-n+1/2}}{w,\ahol{w}} &\sim 
\frac{-\brac{n-1} \func{\nonch V^{+ \pm}_{-e,-n+1/2}}{w,\ahol{w}}}{z-w}, \\
\func{J^-}{z} \func{\nonch V^{+ \pm}_{-e,-n+1/2}}{w,\ahol{w}} &\sim 
\frac{\func{\nonch V^{- \pm}_{-e,-n+1/2}}{w,\ahol{w}}}{z-w},
\end{aligned}
\end{equation}
which imply that these fields correspond to the primaries of $\GLVer{-n+1/2,-e}$.

Correlation functions are now defined in almost exactly the same manner as for symplectic fermions and $\AKMA{sl}{2}_{-1/2}$. The only difference is that there are now two free bosons in addition to the screened boson $\tfunc{\varphi}{z,\bar z}$. The results for the three-point functions include, for example,
\begin{multline}
\corrfn{\func{\nonch V^{--}_{-e_1,-n_1+1/2}}{z,\ahol{z}} \func{\nonch V^{++}_{-e_2,-n_2+1/2}}{1,1} \func{\nonch V^{++}_{-e_3,-n_3+1/2}}{0,0}} \\
= -\frac{\func{\Gamma}{1-e_1} \func{\Gamma}{-e_2} \func{\Gamma}{-e_3}}{\func{\Gamma}{e_1} \func{\Gamma}{1+e_2} \func{\Gamma}{1+e_3}} \frac{\delta_{n_1 + n_2 + n_3 = 2} \delta_{e_1 + e_2 + e_3 = 0}}{\abs{z-1}^{2 \brac{e_2 (1-e_1-n_1) + e_1 (1-n_2)}} \abs{z}^{2 \brac{e_3 (1-e_1-n_1) + e_1 (1-n_3)}}},
\end{multline}
from which we again observe singularities at the atypical points $e \in \ZZ$. As before, regularising leads to logarithms signifying the presence of the indecomposables $\GLStag{n,e}$ and the results are consistent with the fusion rules \eqref{FR:GL11}.

\subsection{Further Developments} \label{sec:GL11Future}

Conformal field theories associated to affine Kac-Moody superalgebras provide a rich source of interesting new logarithmic conformal field theories.  Unfortunately, the only examples that are understood in great detail are those with $\AKMSA{psl}{1}{1}$ and $\AKMSA{gl}{1}{1}$ symmetries.  However, this does not mean that no progress has been made on more sophisticated superalgebra models.

As with $\AKMSA{gl}{1}{1}$, the spectrum of the the bulk theory may be conjectured using a combination of harmonic analysis and a first order formulation. By the latter, one means a perturbative description of the \WZW{} model in terms of its bosonic subtheory coupled to free $bc$-ghosts.  One can then compute correlation functions and so on as above. This approach has been considered, with varying degrees of detail, for the bulk theories corresponding to the \WZW{} models on $\SLSG{SL}{2}{1}$ \cite{Saleur:2006tf}, $\SLSG{PSL}{2}{2}$ \cite{Gotz:2006qp} and for general type I supergroups in \cite{Quella:2007hr}.  A different approach determines, and then exploits, correspondences with super-Liouville theories \cite{Hikida:2007sz,Creutzig:2011qm}. Both methods appear to generalise to the boundary \cft{}, but there is an obstacle amounting to identifying the appropriate boundary screening charges. This has so far only been achieved for $\SLSG{GL}{1}{1}$ \cite{Creutzig:2008ek} and $\SLSG{OSP}{1}{2}$ \cite{Creutzig:2010zp}. In both cases, the boundary screening charge was found to be essentially given by the square root of the bulk screening charge. This behaviour is very similar to that observed for matrix factorization in Landau-Ginzburg theories \cite{Kapustin:2002bi} where the boundary and bulk screening charges seem to be similarly related, at least for $\SLSG{GL}{n}{n}$ \cite{Creutzig:2010ne}.

Some supergroup \WZW{} theories and their cosets have remarkable properties that lead to interesting applications in physics. The key point is the rather innocent-seeming observation that the Killing form, the supertrace in the adjoint representation, vanishes identically for the simple Lie superalgebras $\SLSA{psl}{n}{n}$, $\SLSA{osp}{2n+2}{2n}$ and $\SLA{d}{2,1;\alpha}$.  The corresponding \WZW{} models have been argued to possess exactly marginal perturbations \cite{Bershadsky:1999hk,Quella:2007sg,Candu:2010yg} including, as a special case, the \emph{principal chiral model}.  In the case of $\SLSG{PSU}{1,1}{2}$, this describes (the target space supersymmetric part of) superstring theory on $\grp{AdS}_3 \times S^3 \times X$, where $X$ is some four-dimensional manifold \cite{Berkovits:1999im}.  One prediction of the celebrated \emph{AdS/CFT correspondence} is that this string theory is dual to the two-dimensional \cft{} associated with certain symmetric orbifolds of the four-manifold $X$. The superstring theory dual to four-dimensional conformal gauge theory is likewise described by the \cft{} associated to the coset \cite{Maldacena:1997re}
\[
\frac{\SLSG{PSU}{2,2}{4}}{\SLG{SU}{2,2} \times \SLG{SU}{4}}.
\] 
This coset differs from the gauged \WZW{} model, but it can still be argued to be conformally invariant due to the fact that the Killing form of the numerator vanishes \cite{Kagan:2005wt,Babichenko:2006uc,Candu:2010yg}.

Conformal field theories with superalgebra symmetries also appear in statistical physics. Supersymmetric disordered systems are described by perturbations of $n$ pairs of free fermions and $\beta \gamma$ ghosts. The bilinears in these fields are well known to define the currents of $\AKMSA{gl}{n}{n}$ at level $k=1$. The associated disordered system is given by the corresponding current-current perturbation \cite{Guruswamy:1999hi} (although in this case it seems that the perturbed theory is not quite conformal). One of the important open questions in this area is that of finding an effective field theory for the transitions between plateaux in the integer quantum Hall effect. Such a theory might have $\SLSG{GL}{n}{n}$ symmetry and sigma models related to $\SLSG{PSL}{2}{2}$ have also been argued to appear in this context \cite{Zirnbauer:1999ua}. 

\section{Staggered Modules} \label{sec:Stag}

We have seen in the previous sections that \lcfts{} all have certain types of reducible, but indecomposable, modules in their spectra.  The action of the Virasoro zero-mode $L_0$ is not diagonalisable on these modules, leading to logarithmic singularities in correlators.  In this section, we shall discuss the mathematical structure of the simplest class of modules on which $L_0$ cannot be diagonalised, the \emph{staggered modules}.  These were so named by Rohsiepe in his study \cite{RohRed96} of indecomposable Virasoro modules formed by glueing several \hwms{} together.%
\footnote{The ``staggering'' presumably derives from a useful pictorial representation in which the vertical positions of the \hwss{} are ordered according as to their conformal dimensions.  The result bears a passing similarity to the staggered starting positions customarily used for runners racing around a track.}  %
Here, we will discuss indecomposables formed by glueing modules from more general, but still structurally well-understood, classes.  However, we will restrict to glueing only two of these well-understood modules together.  These are the most common types of indecomposables encountered in \lcft{} and the majority of the best understood examples of these theories feature only this type of reducible indecomposable (those studied in \secTref{sec:Trip}{sec:SL2}{sec:GL11} for example).  A further advantage is that we expect a classification of such staggered modules to be feasible, see \cite{RohRed96,RidSta09} for the Virasoro case.

\subsection{Staggered Modules for Associative Algebras} \label{sec:StagMod}

In preparation for defining staggered modules in some generality, we first declare that for a given Lie algebra (or more generally, associative algebra%
\footnote{It is more convenient to consider the \uea{} of the Lie algebra for the mathematical development that follows.}%
), we will choose a collection of \emph{standard modules} whose structure is reasonably well-understood:  For $\AKMSA{gl}{1}{1}$, the standard modules are the Verma modules; for $\AKMA{sl}{2}_{-1/2}$, standard means the spectral flows of the relaxed \hwms{} which we have denoted by $\SLTyp{\lambda}$; for $\SingAlg{1,2}$, standard means the Feigin-Fuchs modules $\SingTyp{\mu}$.  These are the modules whose characters have the most satisfactory modular transformation properties.  We remark that these notions of being ``standard'' may be lifted to the respective simple current extensions by declaring that an extended module is standard whenever it is the orbit, under fusing with the simple current, of a standard module.  In this sense, $\TripIrr{-1/8}$ and $\TripIrr{3/8}$ are standard $\TripAlg{1,2}$-modules, while $\TripIrr{0}$ and $\TripIrr{1}$ are not.

Note that this notion is more general than \lcft{}.  For example, we may choose the standard modules of a (type I) Lie superalgebra like $\SLSA{gl}{1}{1}$ to be its Kac modules \cite{KacCha77}.%
\footnote{Actually, we would have liked to refer to standard modules in general as ``Kac modules''.  However, this term is already in use for \lcfts{} with Virasoro algebra symmetries \cite{RasFus07}.  It is not clear to us at present if these Kac modules for the Virasoro algebra are good candidates for standard modules in the sense we wish.  Moreover, a theory of Virasoro staggered modules, with standard meaning highest weight, has already been developed \cite{RidSta09}.}  %
Similarly, a good choice for quantum groups like $\mathcal{U}_q \brac{\SLA{sl}{2}}$ is that standard means highest weight.  Even the diagram algebras, such as the Temperley-Lieb algebra, which crop up in statistical lattice models have standard modules, in this case the right choices are the cell modules of Graham and Lehrer \cite{GraCel96} (also known as standard modules; we have borrowed the nomenclature from this example).

The feature that these collections have in common is that standard modules are always indecomposable and, moreover, are naturally parametrised so that they are generically irreducible.  In this review, we have referred to irreducible standard modules as typical and reducible ones as atypical.  Furthermore, and this is crucial for the next definition, the central elements of the associative algebra all act diagonalisably on standard modules.  We now define a staggered module $\VirStag{}$ to be one which is (isomorphic to) an extension of a standard module by another upon which there is a central element $Q$ acting non-diagonalisably:
\begin{equation} \label{ses:DefStag}
\dses{\LMod}{\iota}{\VirStag{}}{\pi}{\RMod}.
\end{equation}
In other words, $\VirStag{}$ has a submodule isomorphic to a standard module $\LMod$, which we shall refer to as the \emph{left module}, and the quotient $\VirStag{} / \func{\iota}{\LMod}$ is isomorphic to another standard module $\RMod$, which we shall refer to as the \emph{right module}.  For $\SLSA{gl}{1}{1}$ and $\mathcal{U}_q \brac{\SLA{sl}{2}}$, we may take $Q$ to be quadratic Casimir; for the Temperley-Lieb algebra, $Q$ can be taken to be the ``braid transfer matrix'' \cite{RidTL12}.  For the algebras arising in \lcft{}, it is $L_0$ which acts non-diagonalisably --- to get a central element, one can act with $\ee^{2 \pi \ii L_0}$ instead.  It is not hard to see that $L_0$ and $\ee^{2 \pi \ii L_0}$ may be interchanged with only minor modifications to the arguments that follow.%
\footnote{One could also imagine \cfts{} on which other zero-modes act non-diagonalisably, however very few examples appear to be known (see \cite{BabTak12} for one).}

The first result to note about staggered modules is that they only exist if $Q$ acts on both the left and right modules as the same multiple $\lambda$ of the identity (\appref{app:Exact}).  The second result to note is that the non-trivial Jordan blocks of $Q$ all have rank $2$.  This follows directly from the exactness of \eqref{ses:DefStag}:  If $\ket{v} \in \VirStag{}$ is an arbitrary element of a non-trivial Jordan block for $Q$ of (generalised) eigenvalue $\lambda$, then $\brac{Q - \lambda} \ket{v}$ need not be $0$, but $\pi \brac{Q - \lambda} \ket{v} = \brac{Q - \lambda} \pi \ket{v} = 0$, since $Q$ is diagonalisable on $\RMod$.  Exactness then gives $\brac{Q - \lambda} \ket{v} = \iota \ket{w}$, for some $\ket{w} \in \LMod$ of $Q$-eigenvalue $\lambda$, whence $\brac{Q - \lambda}^2 \ket{v} = \brac{Q - \lambda} \iota \ket{w} = \iota \brac{Q - \lambda} \ket{w} = 0$, because $Q$ is also diagonalisable on $\LMod$.

To further investigate these staggered modules, we introduce some more notation.  Let $\ket{\theta} \in \VirStag{}$ be a generalised eigenvector of $Q$ of eigenvalue $\lambda$, so that $\ket{\chi} = \brac{Q - \lambda} \ket{\theta} \neq 0$.  Then, $\pi \ket{\theta} \in \RMod$ is non-zero because $\pi \ket{\theta} = 0$ implies that $\ket{\theta} \in \tfunc{\iota}{\LMod}$, by exactness, hence that $\ket{\theta}$ is a genuine eigenvector of $Q$.  Now, suppose that there is an element $U$ of our associative algebra such that $U \pi \ket{\theta} = 0$.  Then, $\pi U \ket{\theta} = 0$, hence $U \ket{\theta} \in \tfunc{\iota}{\LMod}$.  Because of this, the centrality of $Q$ now gives
\begin{equation}
U \ket{\chi} = U \brac{Q - \lambda} \ket{\theta} = \brac{Q - \lambda} U \ket{\theta} = 0,
\end{equation}
the last equality again following from the fact that elements of $\tfunc{\iota}{\LMod}$ are genuine eigenvectors of $Q$.  In other words, any $U$ annihilating $\pi \ket{\theta} \in \RMod$ also annihilates $\ket{\chi} \in \tfunc{\iota}{\LMod}$.

This is a bit abstract, so let us consider an important specialisation that occurs for a Lie (super)algebra with standard meaning highest weight and $\ket{\theta}$ chosen to project onto the \hws{} of $\RMod$.  Then, $\ket{\chi} = \brac{Q - \lambda} \ket{\theta}$ is necessarily non-zero because $Q$ would be diagonalisable on all of $\VirStag{}$ otherwise.  But, $\ket{\theta}$ is annihilated by all positive modes $U$, hence we can conclude that $\ket{\chi}$ is too.  It follows that in a staggered module over a Lie (super)algebra, $\ket{\chi}$ is a \emph{non-zero \sv{}} of $\LMod$.  Similarly, if $\pi \ket{\theta}$ is a relaxed \hws{}, then so is $\ket{\chi}$ (we might call it a relaxed \sv{}).

Taking this a step further, we may suppose that $\RMod$ is generated by a set $\tset{\ket{\Theta^j}}_{j \in J}$ and then choose elements $\ket{\theta^j} \in \VirStag{}$ so that $\pi \ket{\theta^j} = \ket{\Theta^j}$.  For each $j \in J$, the elements of our associative algebra which annihilate $\ket{\Theta^j}$ form an ideal whose generators we denote by $U^j_i$, $i \in I$.  We now define $\ket{\omega^j_i} \in \LMod$ by
\begin{equation}
\iota \ket{\omega^j_i} = U^j_i \ket{\theta^j}
\end{equation}
(applying $\pi$ shows that $U^j_i \ket{\theta^j} \in \tfunc{\iota}{\LMod}$).  Because there is an ambiguity in our choice of each $\ket{\theta^j}$ up to adding arbitrary elements $\ket{\eta^j} \in \LMod$, there is a similar ambiguity in our definition of the $\ket{\omega^j_i}$:
\begin{equation}
\ket{\theta^j} \lra \ket{\theta^j} + \iota \ket{\eta^j} \qquad \Rightarrow \qquad 
\ket{\omega^j_i} \lra \ket{\omega^j_i} + U^j_i \ket{\eta^j}.
\end{equation}
With this setup, we can prove that the $\ket{\omega^j_i}$ determine the isomorphism class of a staggered module, generalising the Virasoro result given in \cite[Prop.~3.6]{RidSta09}.

\begin{thm}
Let $\VirStag{}$ and $\overline{\VirStag{}}$ be staggered modules with the same left module $\LMod$ and the same right module $\RMod$.  Then, there exists an isomorphism $\psi \colon \overline{\VirStag{}} \to \VirStag{}$ making the diagram
\begin{equation}
\begin{CD}
0 @>>> \LMod @>{\iota}>> \VirStag{} @>{\pi}>> \RMod @>>> 0 \\
@. @| @AA{\psi}A @| @. \\
0 @>>> \LMod @>{\overline{\iota}}>> \overline{\VirStag{}} @>{\overline{\pi}}>> \RMod @>>> 0
\end{CD}
\end{equation}
commute if and only if there exist $\ket{\eta^j} \in \LMod$, for each $j \in J$, such that
\begin{equation} \label{eq:StagToBeProved}
\ket{\overline{\omega}^j_i} = \ket{\omega^j_i} + U^j_i \ket{\eta^j},
\end{equation}
for all $i \in I$.
\end{thm}
\begin{proof}
If we have such an isomorphism $\psi$, then
\begin{equation}
\pi \tbrac{\psi \ket{\overline{\theta}^j} - \ket{\theta^j}} = \overline{\pi} \ket{\overline{\theta}^j} - \pi \ket{\theta^j} = \ket{\Theta^j} - \ket{\Theta^j} = 0,
\end{equation}
hence $\psi \ket{\overline{\theta}^j} - \ket{\theta^j} = \iota \ket{\eta^j}$ for some $\ket{\eta^j} \in \LMod$.  Applying $U^j_i$ now gives \eqref{eq:StagToBeProved}, as required.

Conversely, if \eqref{eq:StagToBeProved} holds for some $\ket{\eta^j} \in \LMod$, then define $\psi$ to be $\iota \circ \overline{\iota}^{-1}$ on $\tfunc{\overline{\iota}}{\LMod}$ and by
\begin{equation}
\psi \ket{\overline{\theta}^j} = \ket{\theta^j} + \iota \ket{\eta^j}, \qquad \psi U \ket{\overline{\theta}^j} = U \psi \ket{\overline{\theta}^j}
\end{equation}
otherwise.  Now, $\psi$ is clearly a homomorphism and it is easy to invert.  All we need to check is that it is well-defined because it may happen that $U \ket{\overline{\theta}^j} \in \tfunc{\overline{\iota}}{\LMod}$, in which case the two definitions for $\psi$ must agree.  But, $U \ket{\overline{\theta}^j} \in \tfunc{\overline{\iota}}{\LMod}$ leads to $U \ket{\overline{\Theta}^j} = 0$, hence $U$ belongs to the (left) ideal that annihilates $\ket{\overline{\Theta}^j}$.  Thus, we may write $U = \sum_i V^j_i U^j_i$, where the $V^j_i$ are elements of our associative algebra.  The first definition for $\psi$ now gives
\begin{equation}
\psi U \ket{\overline{\theta}^j} = \iota \overline{\iota}^{-1} \sum_i V^j_i \overline{\iota} \ket{\overline{\omega}^j_i} = \iota \sum_i V^j_i \brac{\ket{\omega^j_i} + U^j_i \ket{\eta^j}} = \iota \sum_i V^j_i \ket{\omega^j_i} + \iota U \ket{\eta^j},
\end{equation}
while the second yields
\begin{equation}
\psi U \ket{\overline{\theta}^j} = U \brac{\ket{\theta^j} + \iota \ket{\eta^j}} = \sum_i V^j_i \iota \ket{\omega^j_i} + \iota U \ket{\eta^j},
\end{equation}
completing the proof.
\end{proof}

We illustrate this theorem with staggered Virasoro modules for which standard means Verma.  Then, $\RMod$ is generated by its \hws{} $\ket{\Theta}$ (so $\abs{J} = 1$) and the annihilator of this state is generated by $L_1$, $L_2$ and $L_0 - h$, where $h$ is the conformal dimension of $\ket{\Theta}$.  We therefore obtain three vectors
\begin{equation}
\ket{\omega_0} = \brac{L_0 - h} \ket{\theta}, \qquad 
\ket{\omega_1} = L_1 \ket{\theta}, \qquad \ket{\omega_2} = L_2 \ket{\theta},
\end{equation}
where $\pi \ket{\theta} = \ket{\Theta}$.  Up to the ambiguity in choosing $\ket{\theta}$, these three vectors completely specify the isomorphism class of a staggered module.  In fact, because $\ket{\omega_0}$ is (a rescaling of) $\ket{\chi}$ and Virasoro singular vectors of a given conformal dimension are unique, this vector is already determined by $\LMod$ and $\RMod$, so the staggered module is characterised by $\ket{\omega_1}$ and $\ket{\omega_2}$.%
\footnote{This remains true if standard means, instead, highest weight because the additional generators of the annihilating ideal lead to additional states $\ket{\omega_i}$ which can be computed given the $\ket{\omega_1}$ and $\ket{\omega_2}$ defined in the Verma case \cite[Prop.~3.4]{RidSta09}.}

This theorem therefore allows us to reduce the problem of deciding if two staggered modules are isomorphic to computing a set of vectors $\ket{\omega^j_i} \in \LMod$ and seeing if there exist $\ket{\eta^j} \in \LMod$ such that \eqref{eq:StagToBeProved} holds.  If one can determine which sets of $\ket{\omega^j_i}$ do actually correspond to a staggered module --- this is the existence problem and it is decidedly difficult in general --- then the question of counting the number of isomorphism classes of staggered modules with given left and right modules becomes an exercise in linear algebra.  We remark that this is not quite the same as computing $\ExtGrp{1}{\RMod}{\LMod}$ (see \appref{app:Exact}) because we are restricting to extensions with a non-diagonalisable action of the centre.

\subsection{Logarithmic Couplings} \label{sec:Beta}

Because the counting of staggered module isomorphism classes may be reduced, modulo the existence problem, to a question of linear algebra, it seems plausible that the space of isomorphism classes will be a vector space (or affine space).  It seems reasonable to ask if there is a natural means to parametrise this space.  This is the idea behind logarithmic couplings:  Instead of characterising a staggered module by a collection of vectors $\ket{\omega^j_i}$, subject to the ambiguities \eqref{eq:StagToBeProved}, we try to find a collection of numbers which likewise characterise the staggered module but which are \emph{invariant} under \eqref{eq:StagToBeProved}.

This programme has only been studied in any detail for the Virasoro algebra (related computations for $\AKMA{sl}{2}_{-1/2}$ were detailed in \cite{RidFus10} and for $\AKMA{sl}{2}_{-4/3}$ in \cite{GabFus01,CreMod12}).  When standard means Verma, it turns out \cite[Thm.~6.4 and Thm.~6.14]{RidSta09} that the vector space of isomorphism classes of staggered modules has dimension $0$ when $\ket{\chi}$ is the (generating) \hws{} of $\LMod$ (staggered modules are unique) and has dimension $1$ when $\ket{\chi}$ is a \emph{principal} \sv{}, meaning that it is descended from no other proper \sv{}.  These are the most useful cases, though it is also possible for the dimension to be $2$, and we shall consider the principal case exclusively for the remainder of this section.

We therefore need a single number to identify a staggered Virasoro module, up to isomorphism, when $\ket{\chi}$ is principal.  A method to compute this number was originally proposed by Gaberdiel and Kausch \cite{GabInd96} in the course of explicitly constructing certain $c=-2$ and $c=-7$ staggered modules using their fusion algorithm.  They chose $\ket{\theta}$ to satisfy $L_n \ket{\theta} = 0$, for all $n>1$, and then defined $\beta \in \CC$ by $L_1^{\ell} \ket{\theta} = \beta \ket{\xi}$, where $\ket{\xi}$ denotes the \hws{} of $\LMod$ and $\ell$ is the difference between the conformal dimensions of $\ket{\theta}$ and $\ket{\xi}$.%
\footnote{We remark that $\ket{\theta}$ is naturally restricted to being a generalised eigenvalue of $L_0$ in this development.}  %
One can check that $\beta$ does not depend upon the choice of $\ket{\theta}$, assuming that we only choose among the $\ket{\theta}$ which are annihilated by the $L_n$ with $n>1$.  It follows that isomorphic staggered modules will have the same $\beta$.

More recently, these numbers $\beta$ were generalised to other staggered modules, most notably in \cite{EbeVir06} where it was realised that Gurarie and Ludwig's anomaly numbers (see \eqnref{eq:AnomalyNumbers}) for percolation and dilute polymers were just (differently normalised versions of) the $\beta$ for the $c=0$ staggered modules $\VirStag{1,5}$ and $\VirStag{3,1}$ (see \secDref{sec:NGK}{sec:PercFuture}).  Unfortunately, attention was not always paid to the crucial requirement that $\beta$ not depend upon the choice of $\ket{\theta}$.  In particular, when one considers staggered modules more general than those considered in \cite{GabInd96}, it is not usually possible to find any $\ket{\theta}$ satisfying $L_n \ket{\theta} = 0$ for all $n>1$.  It follows that the proposed recipe for computing $\beta$ does not make sense for general staggered modules.

This was corrected in \cite{RidPer07} where a definition of $\beta$ was given for any staggered Virasoro module with $\ket{\chi}$ principal (see \cite{RidSta09} for the general case).  First, we normalise \cite{AstStr97} the \sv{} $\ket{\chi}$ so that
\begin{equation} \label{eq:NormOmega0}
\ket{\chi} = U \ket{\xi}, \qquad U = L_{-1}^{\ell} + \cdots,
\end{equation}
where the omitted terms are Virasoro monomials involving the $L_{-n}$ with $n \geqslant 1$ and at least one $n>1$.  We then define $\beta$ by
\begin{equation} \label{eq:DefBeta}
U^{\dag} \ket{\theta} = \beta \ket{\xi},
\end{equation}
where $L_n^{\dag} = L_{-n}$ is the usual adjoint (lifted to the \uea{}).  If $\ket{\xi}$ is given norm $1$, then applying $\bra{\xi}$ to both sides of this definition leads to
\begin{equation}
\beta = \bracket{\xi}{U^{\dag}}{\theta} = \braket{\chi}{\theta},
\end{equation}
which is obviously invariant under $\ket{\theta} \to \ket{\theta} + \iota \ket{\eta}$, $\ket{\eta} \in \LMod$, because $\ket{\chi}$ is singular in $\tfunc{\iota}{\LMod}$.  The quantities $\beta$ were christened \emph{logarithmic couplings} in \cite{RidPer07}, though the terms \emph{beta-invariants} \cite{RidSta09} and \emph{indecomposability parameters} \cite{DubCon10} have also been used since.

The logarithmic coupling defined in \eqref{eq:DefBeta} has the property that two staggered modules with the same left and right modules will be non-isomorphic if their couplings are different.  The converse is also true \cite[Thm.~6.15]{RidSta09}:  If the logarithmic couplings of such staggered modules coincide, then the modules are isomorphic.  In other words, $\beta$ is a complete invariant of the space of isomorphism classes of staggered Virasoro modules.  The same is true for staggered Virasoro modules with standard meaning highest weight, though it is then no longer true that every $\beta$ need correspond to a staggered module.

A downside to this theory is that $\beta$ is not particularly easy to compute in general.  One can explicitly construct the staggered module, for example using the Nahm-Gaberdiel-Kausch algorithm \cite{GabInd96,EbeVir06,RidPer07} (or as a space of local martingales for a Schramm-Loewner evolution process \cite{KytFro08}).  In this way, the logarithmic couplings $\beta_{r,s}$ of some of the $c=0$ staggered modules $\VirStag{r,s}$ considered in \secref{sec:Perc} were computed:%
\footnote{We remark that the discrepancy between the logarithmic couplings $\beta_{1,5} = -\tfrac{5}{18}$, $\beta_{3,1} = \tfrac{10}{27}$ and the anomaly numbers $b_{1,5} = -\tfrac{5}{8}$, $b_{3,1} = \tfrac{5}{6}$ of Gurarie and Ludwig (see \eqnref{eq:AnomalyNumbers}) is just a matter of normalisation.  We have chosen to (canonically) normalise $\ket{\chi}$ as $\brac{L_{-1}^2 - \tfrac{2}{3} L_{-2}} \ket{\xi}$.  Identifying $\ket{\xi}$ with the vacuum $\ket{0}$, we find that $\ket{T} = -\tfrac{3}{2} \ket{\chi}$ and $\ket{t} = -\tfrac{3}{2} \ket{\theta}$, hence the anomaly numbers are obtained from the logarithmic couplings by multiplying by $\brac{-\tfrac{3}{2}}^2$.}%
\begin{equation} \label{eq:c=0Betas}
\beta_{1,4} = -\tfrac{1}{2}, \quad \beta_{1,5} = -\tfrac{5}{18}, \quad \beta_{1,7} = -420, \quad \beta_{1,8} = -\tfrac{10780000}{243}, \quad \beta_{3,1} = \tfrac{10}{27}.
\end{equation}
A slightly more efficient method \cite{RidLog07} is to check if the existence of \svs{} in staggered modules fixes $\beta$.  However, the most efficient known seems to be the proposal of \cite{VasInd11} in which a limit formula is obtained for $\beta$ as a byproduct of cancelling divergences in Virasoro primary \opes{} as the conformal dimensions and central charges tend to their required values in a controlled manner.  Surprisingly, the logarithmic couplings of certain classes of staggered modules with $\RDim - \LDim$ small can be computed as a function of the central charge \cite{RidLog07,KytFro08,RidNon12,GaiPhy12}.  These functions are reasonably simple and conjectures for more general formulae have been made.  However, there appears to have been no progress as yet on resolving these conjectures.

\subsection{More Logarithmic Correlation Functions} \label{sec:StagLogCorr}

We now briefly reconsider the (chiral) two-point function calculations of \secDref{sec:LogCorr}{sec:PercLogCorr}.%
\footnote{Three-point correlators may likewise be computed, assuming that one has already determined the three-point coupling constants between primary fields.}  %
Given the states $\ket{\theta}$ and $\ket{\chi} = \brac{L_0 - h} \ket{\theta}$, the global invariance of the vacuum always leads to the following form for the two-point functions of the corresponding fields:
\begin{equation} \label{eq:2ptVirLogCorr}
\corrfn{\func{\chi}{z} \func{\chi}{w}} = 0, \qquad \corrfn{\func{\chi}{z} \func{\theta}{w}} = \frac{B}{\brac{z-w}^{2h}}, \qquad \corrfn{\func{\theta}{z} \func{\theta}{w}} = \frac{A - 2B \func{\log}{z-w}}{\brac{z-w}^{2h}},
\end{equation}
where $A$ and $B$ are constants.  If $\ell = 0$, so $\ket{\theta}$ has the same conformal dimension as $\ket{\xi}$, then $\ket{\chi} = \ket{\xi}$ and $B = \braket{\xi}{\theta}$.  In this case, $\ket{\xi}$ has norm zero, so we are free to normalise the hermitian form so that $B=1$.  Note that $B$ does not depend upon the choice of $\ket{\theta}$, whereas $A$ does.  We remark that if the vacuum $\ket{0}$ has a logarithmic partner $\ket{\Omega}$, so $L_0 \ket{\Omega} = \ket{0}$, then this analysis shows that the one-point function of the identity field vanishes whereas that of its partner will be constant.

When $\ell = 1$, $\ket{\chi}$ can only be $L_{-1} \ket{\xi}$ which is only singular when $\ket{\xi}$ has dimension $0$ and $\ket{\theta}$ has dimension $1$.  It follows that the hermitian form may be normalised via $\corrfn{\func{\xi}{z} \func{\xi}{w}} = \braket{\xi}{\xi} = 1$.  Because $\ket{\omega_1} = L_1 \ket{\theta} = \beta \ket{\xi}$, we can compute $B$ using the \pde{} derived from the $L_1$-invariance of the vacuum.  This equation has inhomogeneous terms proportional to the logarithmic coupling $\beta$ of the staggered module which determine the constant of integration in $\corrfn{\func{\xi}{z} \func{\theta}{w}}$:
\begin{equation} \label{eq:L=1CorrFns}
\corrfn{\func{\xi}{z} \func{\theta}{w}} = \frac{\beta}{z-w} \qquad \Rightarrow \qquad \corrfn{\func{\chi}{z} \func{\theta}{w}} = \frac{-\beta}{\brac{z-w}^2} \qquad \Rightarrow \qquad B = -\beta.
\end{equation}
The constant $A$ appearing in $\corrfn{\func{\theta}{z} \func{\theta}{w}}$ again varies with the choice of $\ket{\theta}$, unless $\beta = 0$.

For $\ell > 1$, computing the proportionality constant between $B$ and $\beta$ is a little more cumbersome.  As above, one way is to determine $\corrfn{\func{\chi}{z} \func{\theta}{w}}$ from $\corrfn{\func{\xi}{z} \func{\theta}{w}}$.  For this, we derive in the usual fashion, for $n \geqslant 2$ and an arbitrary field $\func{\phi}{z}$, the following relation:
\begin{equation} \label{eq:2PtLogDesc}
\corrfn{\func{\brac{L_{-n} \phi}}{z} \func{\theta}{w}} = \frac{\brac{-1}^n}{\brac{z-w}^{n-1}} \sqbrac{\pd_w + \frac{h \brac{n-1}}{z-w}} \corrfn{\func{\phi}{z} \func{\theta}{w}} + \brac{-1}^n \sum_{k=1}^{\ell} \binom{k+n-1}{n-2} \frac{\corrfn{\func{\phi}{z} \func{\omega_k}{w}}}{\brac{z-w}^{n+k}}.
\end{equation}
Here, we have defined $\ket{\omega_k} = L_k \ket{\theta}$.  Moreover, we note that when $\ell > 1$, we may always choose $\ket{\theta}$ so that $\ket{\omega_1} = 0$.%
\footnote{This follows from counting arguments.  If $\func{p}{\ell}$ denotes the number of partitions of $\ell$, then $\ket{\omega_1}$ belongs to a space of dimension $\func{p}{\RDim - \LDim - 1}$, whereas one has $\func{p}{\RDim - \LDim} - 1$ effective independent shifts $\ket{\theta} \to \ket{\theta} + \ket{\eta}$.  Because $\func{p}{\ell} - \func{p}{\ell - 1} \geqslant 1$ whenever $\ell > 2$, we conclude that $\ket{\theta}$ may always be shifted so that $\ket{\omega_1} = 0$.  Strictly speaking, we should also allow for the possibility that $\LMod$ has a vanishing \sv{} of dimension less than that of $\ket{\omega_0}$, but this turns out not to change the result.}  %
For this choice of $\ket{\theta}$, we deduce (from the global conformal invariance of the vacuum) that $\corrfn{\func{\xi}{z} \func{\theta}{w}} = 0$.  (This is certainly not true when $\ell = 1$ --- see \eqref{eq:L=1CorrFns}.)

To illustrate the method, consider the $c=0$ staggered module $\VirStag{1,7}$ that was briefly considered in \secref{sec:NGK}.  Recalling \eqref{ses:S17S18} and \eqref{eq:c=0Betas}, we note that $\ket{\xi}$ has dimension $2$, $\ket{\theta}$ has dimension $h = 5$, $\beta = -420$, and $\ket{\chi} = \brac{L_{-1}^3 - 6 L_{-2} L_{-1} + 6 L_{-3}} \ket{\xi}$.  Setting $\ket{\omega_1} = 0$ forces $\ket{\omega_2} = -\tfrac{1}{48} \beta L_{-1} \ket{\xi}$ and $\ket{\omega_3} = \tfrac{1}{12} \beta \ket{\xi}$.  Using the vanishing of $\corrfn{\func{\xi}{z} \func{\theta}{w}}$ and \eqref{eq:2PtLogDesc}, we now obtain
\begin{equation}
\begin{aligned}
\corrfn{\func{\brac{L_{-3} \xi}}{z} \func{\theta}{w}} &= -4 \frac{\corrfn{\func{\xi}{z} \func{\omega_2}{w}}}{\brac{z-w}^5} - 5 \frac{\corrfn{\func{\xi}{z} \func{\omega_3}{w}}}{\brac{z-w}^6} = \frac{-\beta / 12}{\brac{z-w}^{10}}, \\
\corrfn{\func{\brac{L_{-2} L_{-1} \xi}}{z} \func{\theta}{w}} &= \frac{\corrfn{\func{\brac{L_{-1} \xi}}{z} \func{\omega_2}{w}}}{\brac{z-w}^5} = \frac{5 \beta / 12}{\brac{z-w}^{10}},
\end{aligned}
\end{equation}
and $\corrfn{\func{\brac{L_{-1}^3 \xi}}{z} \func{\theta}{w}} = 0$.  Therefore,
\begin{equation}
\corrfn{\func{\chi}{z} \func{\theta}{w}} = \frac{-3 \beta}{\brac{z-w}^{10}}
\end{equation}
and $B = -3 \beta = 1260$.  Unfortunately, we are not aware of any general results concerning the ratio $B / \beta$.  However, the limit formula for $\beta$ given in \cite{VasInd11} has a variant which gives $B$ directly \cite{GaiLat13}.

Whether one prefers $\beta$ or $B$, it is clear that such a parameter is mathematically important and physically relevant.  For example, the logarithmic couplings for the $c=0$ staggered modules containing the vacuum distinguish the percolation and dilute polymers theories.  In this case, the staggered modules have different right modules, but this might be too difficult to check explicitly in a more general physical situation.  Another example occurs in the symplectic fermions theory in which one can identify $c=-2$ staggered Virasoro modules where $\ket{\xi}$ and $\ket{\theta}$ have dimensions $0$ and $1$, respectively.  For each such module, one example being $\ket{\theta} = J^+_{-1} \ket{\Omega}$ and $\ket{\xi} = -J^+_0 \ket{\Omega}$ (see \secref{sec:SF} for notation), $\beta$ is found to be $-1$.  However, there is another $c=-2$ theory, the \emph{abelian sandpile model} (see \cite{DhaAbe99} for example), in which a staggered module with $\ket{\xi}$ dimension $0$ and $\ket{\theta}$ dimension $1$ is present.  However, the logarithmic coupling here has been measured to be $\beta = \tfrac{1}{2}$ \cite{JenHei06}, indicating that this theory is not equivalent to symplectic fermions.

\subsection{Further Developments} \label{sec:StagFuture}

As noted at the beginning, the staggered modules which have occupied us throughout this section are among the simplest, structurally, which give rise to logarithmic singularities in two-point functions.  However, more complicated structures may arise:  Examples of indecomposable Virasoro modules on which $L_0$ acts with Jordan blocks of rank $3$ were first discovered in \cite{EbeVir06} through fusion.  These appear structurally as a glueing of four \hwms{} and their existence was posited more generally in \cite{RasFus07,RasFus07b} using a (conjectured) means of analysing Virasoro representations using lattice techniques.  Moreover, indecomposables formed by glueing three \hwms{}, but with only rank $2$ Jordan blocks, have been shown to arise in the percolation \cft{} \cite{RidPer08}.  More recent lattice computations \cite{GaiPhy12} suggest that indecomposables with Jordan blocks of arbitrary rank are physically relevant.  Unfortunately, there is almost nothing known about the finer mathematical structure of these more complicated indecomposables.

For completeness, we consider an example illustrating a Virasoro indecomposable with a rank $3$ Jordan block for $L_0$.  It may be realised as the fusion product of the $c=0$ irreducibles $\VirIrr{1/8}$ and $\VirIrr{-1/24}$, but has only been explicitly constructed to grade $6$ \cite{EbeVir06}, so its deeper structure remains unknown (see \cite{RasFus07} for a conjectured character).  There are two ground states of dimension $0$ and they form a Jordan block for $L_0$.  We denote the eigenvector by $\ket{\xi}$ and its partner by $\ket{\theta}$, normalised so that $L_0 \ket{\theta} = \ket{\xi}$.  At dimension $1$, one finds four states with all but one forming a rank $3$ Jordan block.  $L_{-1} \ket{\xi}$ is found to be non-zero, hence is the $L_0$-eigenstate belonging to the Jordan block; the other eigenstate will be denoted by $\ket{\xi'}$.  The generator of the Jordan block will be denoted by $\ket{\zeta}$ and, along with $L_{-1} \ket{\theta}$, this completes a basis for the dimension $1$ subspace.

The states $\ket{\zeta}$ and $\ket{\xi'}$ must satisfy the following relations:
\begin{equation}
\brac{L_0 - 1} \ket{\zeta} = a_1 L_{-1} \ket{\theta} + a_2 \ket{\xi'} + a_3 L_{-1} \ket{\xi}, \quad 
L_1 \ket{\zeta} = b_1 \ket{\xi} + b_2 \ket{\theta}, \quad 
L_1 \ket{\xi'} = c_1 \ket{\xi} + c_2 \ket{\theta}.
\end{equation}
However, $a_1 \neq 0$ as $\ket{\zeta}$ generates a rank $3$ block, hence we may scale $\ket{\zeta}$ so that $a_1 = 1$.  Moreover, $\ket{\xi'}$ is an $L_0$-eigenvector, hence $0 = L_1 \brac{L_0 - 1} \ket{\xi'} = L_0 L_1 \ket{\xi'} = c_1 L_0 \ket{\xi} + c_2 L_0 \ket{\theta} = c_2 \ket{\xi}$, giving $c_2 = 0$.  One can check explicitly that $c_1$ is non-zero, hence we may normalise $\ket{\xi'}$ so that $c_1 = 1$.  We now make use of the freedom we have in defining $\ket{\zeta}$:  Shifting by multiples of $L_{-1} \ket{\theta}$ and $\ket{\xi'}$ allows us to tune $a_3$ and $b_1$ to $0$.  Finally, $b_2 \ket{\xi} = L_0 L_1 \ket{\zeta} = L_1 \brac{L_0 - 1} \ket{\zeta} = \brac{a_2 + 2} \ket{\xi}$ reduces us to a single unknown:
\begin{equation}
\brac{L_0 - 1} \ket{\zeta} = L_{-1} \ket{\theta} + \brac{b_2 - 2} \ket{\xi'}, \qquad 
L_1 \ket{\zeta} = b_2 \ket{\theta}, \qquad 
L_1 \ket{\xi'} = \ket{\xi}.
\end{equation}
These conclusions may be checked explicitly with the result that the fusion product has $b_2 = -\tfrac{1}{12}$.  It is now straight-forward to calculate the corresponding two-point correlators:
\begin{equation}
\begin{gathered}
\corrfn{\func{\xi}{z} \func{\xi}{w}} = \corrfn{\func{\xi}{z} \func{\xi'}{w}} = 0, \quad 
\corrfn{\func{\xi}{z} \func{\theta}{w}} = 1, \quad 
\corrfn{\func{\theta}{z} \func{\theta}{w}} = A - 2 \func{\log}{z-w}, \\
\corrfn{\func{\theta}{z} \func{\xi'}{w}} = \frac{1}{z-w}, \quad 
\corrfn{\func{\xi'}{z} \func{\xi'}{w}} = \frac{1}{\brac{z-w}^2}, \quad
\corrfn{\func{\xi}{z} \func{\zeta}{w}} = \frac{b_2}{z-w}, \\
\corrfn{\func{\theta}{z} \func{\zeta}{w}} = \frac{B - 2 b_2 \func{\log}{z-w}}{z-w}, \quad 
\corrfn{\func{\xi'}{z} \func{\zeta}{w}} = \frac{C - \brac{b_2 - 3} \func{\log}{z-w}}{\brac{z-w}^2}, \\
\corrfn{\func{\zeta}{z} \func{\zeta}{w}} = \frac{D + E \func{\log}{z-w} - \brac{b_2 - 1} \brac{b_2 - 6} \func{\log^2}{z-w}}{\brac{z-w}^2}.
\end{gathered}
\end{equation}
The constants $A$, $B$, $C$, $D$ and $E$ all depend upon the precise choices that we have made for the fields, while $b_2$ does not.  We remark that a $\log^2$ term in two-point functions is indicative of a rank $3$ Jordan block for $L_0$.

We conclude by noting that there are large gaps in our knowledge regarding how one can completely specify the isomorphism class of staggered modules for general chiral algebras and for generalised staggered modules such as the one studied above.  While the ``generalised logarithmic coupling'' $b_2$ seems to characterise the module above, we cannot say whether each choice of $b_2$ corresponds to a consistent indecomposable structure or whether $b_2 = -\tfrac{1}{12}$ is the only consistent choice.  A related issue is our lack of knowledge concerning logarithmic couplings for bulk modules.  We have throughout specified Loewy diagrams for atypical bulk indecomposables without considering whether these diagrams completely determine the module structure.  This seems to be a very difficult problem which has only received attention very recently \cite{VasPuz11,RidNon12,GaiPhy12}.  Hopefully, the years to come will improve our understanding of bulk atypicals considerably.

\section{Summary and Conclusions} \label{sec:Conc}

In this review, we have covered various aspects of \lcft{}.  With the exception of motivational sections like \secref{sec:Crossing}, we have worked exclusively in the continuum setting.  The indecomposable structures that arise in logarithmic theories were introduced for the Virasoro algebra through a discussion of the horizontal crossing probability for critical percolation.  While these structures lead directly to the type of singularities in correlation functions that logarithmic theories are named for, the percolation theory is still not particularly well understood.  In particular, we are still not sure of its spectrum.  To contrast this, we then embarked on detailed summaries of three of the best understood examples (we refer to them as ``archetypes'') of \lcfts{}.  Here, the indecomposability arose for modules over certain affine Kac-Moody superalgebras and we were able to provide a relatively complete picture including the spectrum and fusion rules, the chiral characters and their modular transformations, as well as bulk modular invariants and structures for the bulk state space.  We concluded with a brief introduction to the general (and as yet unknown) mathematical theory of the type of indecomposables responsible for logarithmic structure, the staggered modules.

The continuum approach to \lcft{} that we advocate here consists of two essential steps.  First, one needs to thoroughly understand the representation theory of the chiral algebra.  The observed pattern in each of our archetypes is that one has a large spectrum, preferably continuous, of standard modules which are generically irreducible.  The irreducible standard modules are said to be typical, whereas there are also atypical standard modules that are reducible, but indecomposable, the vacuum module being a notable example.  The crucial, but difficult in general, step is to understand the structure of the projective covers (in an appropriate category of vertex algebra modules) of these atypicals.  We expect that there is a generalisation of Bern\v{s}te\u{\i}n-Gel'fand-Gel'fand reciprocity \cite{BerCer76} at work here, where the structure of the atypical standards in terms of irreducibles is related to the structure of the projectives in terms of atypical standards.

The second step is the modular data.  Here, we have seen that the characters of the standard modules span a (projective) representation of the modular group and that, when the spectrum is continuous, this representation may be (topologically) completed so as to include the characters of the atypical irreducibles.  Then, we can apply the continuum Verlinde formula to get Grothendieck fusion coefficients, determine the bulk modular invariants, and discuss extended algebras.  These steps together constitute a procedure which is just a logarithmic adaptation of that detailed in \secref{sec:FreeBoson} for the free boson.  We also view it as a wide-reaching generalisation of the formalism proposed for type I supergroup \WZW{} models in \cite{Quella:2007hr}.  While it has thus far only been tested on logarithmic theories related to affine algebras, we are convinced that our procedure is much more generally applicable.  We therefore summarise it as follows:
\begin{enumerate}[leftmargin=*]
\item Start with a theory possessing a continuous spectrum of well understood standard modules whose characters are all linearly independent.
\item Compute the modular S-transforms of the characters of the standards, the typical irreducibles being a proper subset.
\item Splice short exact sequences to obtain resolutions for the atypical irreducibles in terms of standards. This allows one to express atypical characters as infinite sums of standard characters.
\item Use these sums to deduce the S-matrix entries between atypical irreducibles and standards.
\item Apply the natural continuum version of the Verlinde formula to compute the Grothendieck fusion rules, checking that the resulting multiplicities are non-negative integers.
\item Check the modular invariance of the diagonal bulk partition function and look for simple currents in the fusion rules.  Construct extended algebras with discrete spectra and modular invariant partition functions.
\item Use the Grothendieck fusion rules of the continuum theory to deduce those of its simple current extensions.
\end{enumerate}

We remark that the determination of the genuine fusion rules and the structure of the projective covers of the atypicals still remains.  For this, the Nahm-Gaberdiel-Kausch fusion algorithm will be computationally prohibitive in all but the easiest cases.  However, the Grothendieck rules indicate exactly where one should look for indecomposable structure, especially if one has mathematical knowledge of the irreducibles in the spectrum (extension groups in particular).  We therefore suggest that combining this knowledge with the computation of carefully chosen correlators may be sufficient to determine the structures of the indecomposables that arise by fusing irreducibles.  In this article, we have only partially capitalised upon the power of free field realisations.  All three archetypes admitted a free field realisation inside a lattice theory, the chiral algebra being identified with the kernel of a screening charge.  Such a description can be a powerful tool in representation-theoretic investigations and provide one of the only really general tools available for computing the correlators required to completely determine the fusion rules.

It is worth stressing that the derivation of atypical S-matrix elements, necessary for applying the Verlinde formula, requires \emph{non-periodic} resolutions of atypical irreducibles in terms of indecomposable standards.  The simple current extensions that we construct usually have discrete, or even finite, spectra and so the resolutions for the extended atypicals are periodic.  In many such examples, the S-transformations are known \cite{Flohr:1995ea} to depend upon the modular parameter $\tau$, invalidating a straight-forward application of the Verlinde formula.  There is a proposal \cite{FucNon04} for a generalisation of the Verlinde formula appropriate to certain logarithmic theories with discrete spectra (see also \cite{MiyMod04,FloVer07,GabFro08,GaiRad09,PeaGro10}), so it would be very interesting to see if this proposal can be derived using our extended algebra formalism.  (Of course, the fusion rules for an extended algebra are easily obtained from those of its parent.)

Finally, we remark that there is a real need to extend our detailed knowledge of \lcft{} to more non-trivial examples.  The three archetypes considered in this review all have similar logarithmic structures and are, in fact, very closely related.  Namely, one can construct $\AKMA{sl}{2}_{-1/2}$ as a $\AKMA{u}{1}$-coset of (an extension of) $\AKMSA{gl}{1}{1}$ \cite{Creutzig:2008an}, while the singlet algebra $\SingAlg{1,2}$ and triplet algebra $\TripAlg{1,2}$ are $\AKMA{u}{1}$-cosets of $\AKMA{sl}{2}_{-1/2}$ and its maximal extended algebra, respectively \cite{RidSL210}.%
\footnote{Interestingly, the latter observation is repeated for $\SingAlg{1,3}$, $\TripAlg{1,3}$ and $\AKMA{sl}{2}_{-4/3}$ \cite{AdaCon05}, but the generalisation is that $\SingAlg{1,p}$ and $\TripAlg{1,p}$ are $\AKMA{u}{1}$-cosets of certain fractional level Feigin-Semikhatov algebras $W_n^{(2)}$ \cite{CRWCos13}.}  %
In particular, all the atypical irreducibles of these theories may be resolved in terms of (indecomposable) standards.  The next step in terms of difficulty would be to consider logarithmic theories in which one has \emph{degree $2$} atypicals which are resolved in terms of degree $1$ atypicals, which are themselves resolved in terms of standards.  An example of such a theory has been worked out in \cite{BabTak12}.  We expect that the $\AKMSA{gl}{2}{2}$ and $\AKMSA{psl}{2}{2}$ \WZW{} models will provide further examples which, with applications to both string theory and statistical physics in mind, seem to us to be the most obvious candidates for future study.

\addtocontents{toc}{\setcounter{tocdepth}{-1}}
\section*{Acknowledgements}
\addtocontents{toc}{\setcounter{tocdepth}{2}}

We would like to thank the many many people who have contributed to the writing of this logarithmic review through teaching us, discussing with us and telling us when we were talking rubbish.  These people include, but are not limited to, Drazen Adamovi\'{c}, John Cardy, Matthias Gaberdiel, Azat Gainutdinov, Yasuaki Hikida, Kalle Kyt\"{o}l\"{a}, Andrew Linshaw, Pierre Mathieu, Antun Milas, Paul Pearce, Thomas Quella, J\o{}rgen Rasmussen, Peter R\o{}nne, Philippe Ruelle, Ingo Runkel, Yvan Saint-Aubin, Hubert Saleur, Volker Schomerus, Alyosha Semikhatov, Romain Vasseur and Simon Wood.  The research of DR is supported by an Australian Research Council Discovery Project DP1093910.

\appendix

\section{Homological Algebra:  A (Very Basic) Primer} \label{app:HomAlg}

The quantum state space of a (bulk) \cft{} is a module (representation space) over a symmetry algebra which must contain two commuting copies of the Virasoro algebra (one copy in the boundary case).  The \emph{raison d'\^{e}tre} of \lcft{} is that there are occasions for which this module cannot be written as a direct sum of irreducible submodules over the symmetry algebra, but rather that one must include, in the direct sum, some submodules which are reducible but indecomposable.  When discussing reducible but indecomposable modules, the language of homological algebra and category theory becomes very convenient, indeed almost unavoidable.  In this appendix, we will introduce some of the basic concepts and terminology that are used freely throughout the text.

\subsection{Exact Sequences and Extensions} \label{app:Exact}

First, recall that a module is said to be \emph{reducible} if it contains a non-trivial proper submodule and \emph{decomposable} if it may be written as the direct sum of two non-trivial submodules:  $M = M_1 \oplus M_2$ with $M_1, M_2 \neq \set{0}$.  We remark that in the latter case, $M_1$ and $M_2$ are called \emph{direct summands} of the decomposable module $M$.  An \emph{irreducible} (or \emph{simple}) module is then one which has only the trivial module and itself as submodules.  A direct sum of irreducible modules is said to be \emph{completely reducible} (or \emph{semisimple}).  Similarly, an \emph{indecomposable} module is one for which $M = M_1 \oplus M_2$ implies that either $M_1$ or $M_2$ is trivial.  Finally, we define a \emph{proper} submodule of $M$ to be one which is not all of $M$.  A \emph{maximal} proper submodule of $M$ is then a proper submodule $N$ for which $N \subset M' \subseteq M$ (with $M'$ also a submodule of $M$) implies that $M' = M$.

The structure-preserving maps between modules $M$ and $N$, the \emph{module homomorphisms} (or intertwiners), are defined to be those linear maps $f \colon M \to N$ which commute with the action of the symmetry algebra $A$:
\begin{equation}
a \cdot \func{f}{m} = \func{f}{a \cdot m} \qquad \text{for all \(a \in A\) and \(m \in M\).}
\end{equation}
Canonical examples include the identity map $\id \colon M \to M$, the inclusion map $\iota \colon M \to N$ of a submodule $M$ of $N$, and the canonical projection $\pi \colon N \to N / M$ onto the quotient by a submodule.  We note that both the kernel and image of a module homomorphism $f \colon M \to N$ are submodules (of $M$ and $N$, respectively).

One of the central notions of homological algebra is the \emph{exact sequence}.  This is a chain of modules connected by module homomorphisms,
\begin{equation}
\cdots \lra M_{-2} \overset{f_{-2}}{\lra} M_{-1} \overset{f_{-1}}{\lra} M_0 \overset{f_0}{\lra} M_1 \overset{f_1}{\lra} M_2 \lra \cdots,
\end{equation}
which is \emph{exact}:  At each position $n$ of the chain, the kernel of the outgoing homomorphism $f_n$ coincides with the image of the incoming one $f_{n-1}$, that is, $\ker f_n = \im f_{n-1}$.  One may also consider chains which terminate at either end, in which case one does not require exactness at the endpoints.  As examples, note that the identity, inclusion and quotient maps give rise to the following exact sequences:
\begin{equation}
0 \lra M \overset{\id}{\lra} M \lra 0, \qquad 0 \lra M \overset{\iota}{\lra} N, \qquad N \overset{\pi}{\lra} N/M \lra 0.
\end{equation}
Conversely, the first sequence says that $\id$ is bijective, the second that $\iota$ is injective and the third that $\pi$ is surjective.  We remark that it is standard to abbreviate the trivial module $\set{0}$ to just $0$.

A \emph{short exact sequence} is an exact sequence of the form
\begin{equation} \label{ses}
\dses{M}{\iota}{N}{\pi}{Q}
\end{equation}
(note that $\iota$ is automatically an injection and $\pi$ is automatically a surjection).  This concisely summarises the common situation in which $M$ is a submodule $N$ and the quotient $N/M$ is isomorphic to $Q$.  An obvious question that arises now is whether there is a submodule of $N$ isomorphic to $Q$ such that $N \cong M \oplus Q$ --- if so, then we say that the short exact sequence \emph{splits}.  Certainly, if $N$ is such a direct sum, then the sequence \eqref{ses} is exact.  We are therefore asking if a direct sum is the only possibility.  In the very important case in which both $M$ and $Q$ are indecomposable, this amounts to asking if $M$ and $Q$ can be ``glued together'' to form a new indecomposable $N$, or whether they may only be combined as $M \oplus Q$.

Answering this question is a very subtle and difficult business in general.  The general machinery of homological algebra reduces this to the computation of the first \emph{extension group} $\ExtGrp{1}{Q}{M}$ (when \eqref{ses} is exact, $N$ is said to be an \emph{extension} of $Q$ by $M$).  In particular, $\ExtGrp{1}{Q}{M} = 0$ guarantees that \eqref{ses} always splits, hence that $N \cong M \oplus Q$ is the only possibility.  We will not need to concern ourselves with the formal machinery required to compute extension groups in general.  Instead, we only remark that there is one easy test for deciding when a short exact sequence must split:  If some \emph{central} element of the symmetry algebra $A$ acts as a multiple of the identity on $M$ and as a \emph{different} multiple of the identity on $Q$, then it acts on $N$ as a linear map with two distinct eigenspaces which may easily be checked to be submodules isomorphic to $M$ and $Q$.  In other words, it follows that $N$ splits as $M \oplus Q$.

We will typically apply this test on modules that are graded by the (generalised) eigenvalues of a zero-mode such as that of the Virasoro algebra $L_0$.  In this case, the central element should be identified with $\ee^{2 \pi \ii L_0}$ and the test reduces to the remark that if the states of the indecomposable modules $M_1$ and $M_2$, have conformal dimensions ($L_0$-eigenvalues) in $\ZZ + h_1$ and $\ZZ + h_2$, respectively, then $\ExtGrp{1}{M_1}{M_2} = 0$ will follow if $h_1 \neq h_2 \bmod{\ZZ}$.

We also note that there may exist modules $Q$ for which $\ExtGrp{1}{Q}{M} = 0$ for all modules $M$ (in the category under consideration).  Such $Q$ are said to be \emph{projective}.  These modules give a rough measure of maximal complexity among indecomposable modules.  More precisely, they have the property that if they appear as a quotient of any module, then that module decomposes as the direct sum of $Q$ and something else.  We mention that if an irreducible module $Q$ is a quotient of an indecomposable projective module $P$, then we say that $P$ is the \emph{projective cover} of $Q$ (projectivity guarantees that these covers are unique when they exist).  This is a remarkably useful concept.  Unfortunately, progress in rigorously identifying projective covers for the module categories of interest in \lcft{} is lamentably slow.

\subsection{Splicing Exact Sequences} \label{app:Splicing}

One construction that we will take advantage of is that of \emph{splicing} two short exact sequences together.  This procedure starts from two short exact sequences of the forms
\begin{equation}
\dses{M_1}{\iota_1}{N_1}{\pi_1}{M_0}, \qquad \dses{M_2}{\iota_2}{N_2}{\pi_2}{M_1}
\end{equation}
and produces the sequence
\begin{equation}
0 \lra M_2 \overset{\iota_2}{\lra} N_2 \xrightarrow{\iota_1 \circ \pi_2} N_1 \overset{\pi_1}{\lra} M_0 \lra 0.
\end{equation}
It is a simple exercise to check that the resulting sequence is also exact.  Moreover, if one has a short exact sequence with $M_2$ in the third position, $\ses{M_3}{N_3}{M_2}$ say, then the splicing process may be repeated to obtain a new, longer exact sequence.  The most interesting case occurs when one can splice infinitely many times, thereby obtaining a long exact sequence called a \emph{resolution} of $M_0$ in terms of the $N_i$:
\begin{equation}
\cdots \lra N_5 \xrightarrow{\iota_4 \circ \pi_5} N_4 \xrightarrow{\iota_3 \circ \pi_4} N_3 \xrightarrow{\iota_2 \circ \pi_3} N_2 \xrightarrow{\iota_1 \circ \pi_2} N_1 \overset{\pi_1}{\lra} M_0 \lra 0.
\end{equation}
When the $N_i$ belong to a class of well-behaved modules, one can use this to understand the behaviour of $M_0$.

\subsection{Grothendieck Groups and Rings} \label{app:Grothendieck}

One point that deserves emphasising is that the indecomposable structures that modules may exhibit become irrelevant when modules are replaced by their underlying vector spaces.  This is because exact sequences of vector spaces always split.  It follows that quantities attached to modules which only depend upon their vector space structure, characters being prime examples, will be blind to any indecomposable structure.  One is therefore motivated to consider the effect of forgetting the indecomposable structure of modules such as the quantum state space.  Certainly, it will usually be easier to ignore questions of whether a module is completely reducible or not and for some applications, computing modular invariant partition functions for example, maintaining such ignorance is perfectly justified.

This ``forgetting'' may be formalised through the notion of a \emph{Grothendieck group}.  Here, one starts with a collection (category) of modules, preferably finitely-generated, and forms an abelian group $\Groth$  whose generators are formal elements $\tGr{M}$, where $M$ is a module from the collection, and whose relations are
\begin{equation}
\tGr{M} - \tGr{N} + \tGr{Q} = 0, \qquad \text{whenever} \quad \dses{M}{}{N}{}{Q} \quad \text{is exact.}
\end{equation}
The point of these relations is to ensure that all extensions of $Q$ by $M$ are identified with $M \oplus Q$ in $\Groth$.  In favourable circumstances, such as when the characters of the irreducible modules are linearly independent, then the abelian group generated by the characters will be isomorphic to the Grothendieck group.  Note that the collections of modules that we will be considering will always carry a tensor product structure (fusion) and that this structure will (almost) always induce a well-defined product on $\Groth$.  We will therefore usually refer to $\Groth$ as a \emph{Grothendieck ring} or Grothendieck fusion ring.

\subsection{Socle Series and Loewy Diagrams} \label{app:Socle}

Finally, it is often convenient to go beyond the formalism of exact sequences in order to visualise the structure of an indecomposable module.  One way to do this is through \emph{composition series}.  This is a filtration of a module $M$ by submodules,
\begin{equation} \label{filt}
0 = M_0 \subset M_1 \subset \cdots \subset M_{\ell - 1} \subset M_{\ell} = M,
\end{equation}
such that the \emph{subquotients} $C_i = M_i / M_{i-1}$ of $M$ are all irreducible.  Composition series are not unique, but if $M$ does possess a composition series, then the \emph{length} $\ell$ and the \emph{composition factors} $C_i$ are common to every composition series of $M$ (only the ordering of the $C_i$ may change).

A common variation on this theme is the \emph{socle series}.  Here, one again has a filtration \eqref{filt} by modules, but the condition that $M_i / M_{i-1}$ be irreducible is replaced by the requirement that $M_i / M_{i-1}$ be the \emph{socle} of $M / M_{i-1}$.  The socle of a module $M$ is defined to be its maximal completely reducible submodule.  Equivalently, it is the (direct) sum of the irreducible submodules of $M$.  As the socle is unique (if it exists), the same is true of socle series.  In essence, composition series describe how irreducible modules are glued together to form $M$, whereas socle series describe how completely reducible modules are glued and so are usually more efficient.  The dual notion to the socle is the maximal semisimple quotient, sometimes called the \emph{head}.

We conclude with an abstract example illustrating an indecomposable structure which commonly arises in \lcft{} (and elsewhere).  The socle series is
\begin{equation}
0 = M_0 \subset M_1 \subset M_2 \subset M_3 = M, \quad \text{with} \quad \text{\(M_1 \cong C_1\), \(M_2 / M_1 \cong C_2 \oplus C_3\) and \(M_3 / M_2 \cong C_4\),}
\end{equation}
where the $C_i$ are the irreducible subquotients (composition factors) of $M$.  A common way of visualising this information is through the associated \emph{Loewy diagram}.  This is constructed by ``layering'' with the $i$-th layer consisting of the direct summands comprising the $i$-th filtration quotient $M_i / M_{i-1}$.  In the example at hand, the bottom layer is $\soc M = M_1 \cong C_1$, the middle layer is $\soc \brac{M / M_1} = M_2 / M_1 = C_2 \oplus C_3$, and the top layer is $\soc \brac{M / M_2} = M / M_2 = C_4$.  This is illustrated in \figref{fig:Loewy} (left).

\begin{figure}
\begin{center}
\begin{tikzpicture}[thick,scale=0.8,>=latex]
\node (t) at (-6,1.5) [] {$C_4$};
\node (l) at (-7.5,0) [] {$C_2$};
\node (r) at (-4.5,0) [] {$C_3$};
\node (b) at (-6,-1.5) [] {$C_1$};
\node (top) at (0,1.5) [] {$C_4$};
\node (left) at (-1.5,0) [] {$C_2$};
\node (right) at (1.5,0) [] {$C_3$};
\node (bot) at (0,-1.5) [] {$C_1$};
\draw [->] (top) -- (left);
\draw [->] (top) -- (right);
\draw [->] (left) -- (bot);
\draw [->] (right) -- (bot);
\end{tikzpicture}
\caption{ \label{fig:Loewy} An example of a Loewy diagram (left) and its annotated version (right).}
\end{center}
\end{figure}
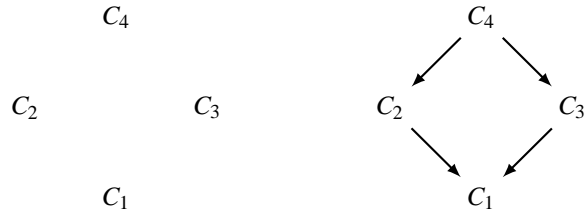

It is often convenient to annotate Loewy diagrams with arrows detailing the finer structure of the indecomposable module as in \figref{fig:Loewy} (right).  Such an arrow will always point down from a composition factor $C_j$ at layer $i$ to another $C_k$ at layer $i-1$.  Roughly speaking, both $C_j$ and $C_k$ may be associated to certain states of $M$ and the arrow indicates that the action of the algebra can take a state associated with $C_j$ to a state associated with $C_k$, but not vice-versa.  More precisely, it is possible to isolate $C_j$ and $C_k$ by canonically constructing a length $2$ subquotient of $M$ whose composition factors are precisely $C_j$ and $C_k$.  If this subquotient is indecomposable, then we draw an arrow from $C_j$ to $C_k$ (if $\ExtGrp{1}{C_j}{C_k}$ has dimension greater than $1$, we should also affix a label to the arrow to precisely identify the subquotient).  We will not try to demonstrate or verify this precise criterion for drawing arrows on Loewy diagrams here.  Suffice to say that it will be clear in the examples considered how composition factors are associated to states and arrows will be drawn on this rough basis.

\raggedright 


\end{document}